\documentclass[11pt]{article}

\usepackage{verbatim}
\usepackage{float}
\usepackage{amsmath}
\usepackage{amsthm}
\usepackage{amssymb}
\usepackage{thmtools}
\usepackage{graphicx}
\usepackage{mathtools}
\usepackage{enumitem}
\usepackage{array}
\usepackage{xspace}

\usepackage{tikz, relsize}
\usepackage{lmodern}
\usetikzlibrary{positioning, shapes, shadows, arrows}
\usetikzlibrary{decorations.pathreplacing}
	\tikzstyle{abstract}=[circle, draw=black, fill=white]
\tikzstyle{labelnode}=[circle, draw=white, fill=white]
\tikzstyle{line} = [draw, -latex']
\tikzset{fontscale/.style = {font=\relsize{#1}}}

\makeatletter

\floatstyle{ruled}
\newfloat{algorithm}{tbp}{loa}
\providecommand{\algorithmname}{Algorithm}
\floatname{algorithm}{\protect\algorithmname}

\usepackage{xcolor}
\usepackage{nameref}
\definecolor{ForestGreen}{rgb}{0.1333,0.5451,0.1333}
\definecolor{DarkRed}{rgb}{0.8,0,0}
\definecolor{Red}{rgb}{1,0,0}
\usepackage[linktocpage=true,
pagebackref=true,colorlinks,
linkcolor=DarkRed,citecolor=ForestGreen,
bookmarks,bookmarksopen,bookmarksnumbered]
{hyperref}

\usepackage{mathtools} % For long arrow with text above https://tex.stackexchange.com/questions/103988/rightarrow-with-text-above-it

\newenvironment{properties}[2][0]
{
\begin{enumerate} \setcounter{enumi}{#1}}{\end{enumerate}}

\usepackage{cleveref}

\declaretheorem[numberwithin=section,refname={Theorem,Theorems},Refname={Theorem,Theorems},name={Theorem}]{thm}
\declaretheorem[numberlike=thm,refname={Theorem,Theorems},Refname={Theorem,Theorems},name={Theorem}]{theorem}
\declaretheorem[numberlike=thm,refname={Claim,Claims},Refname={Claim,Claims},name={Claim}]{claim}
\declaretheorem[numberlike=thm,refname={Lemma,Lemmas},Refname={Lemma,Lemmas},name={Lemma}]{lem}
\declaretheorem[numberlike=thm,refname={Corollary,Corollaries},Refname={Corollary,Corollaries},name={Corollary}]{cor}

\declaretheorem[numberlike=thm,refname={Fact,Facts},Refname={Fact,Facts},name={Fact}]{fact}

\declaretheorem[numberlike=thm,refname={Proposition,Propositions},Refname={Proposition,Propositions},name={Proposition}]{prop}
\declaretheorem[numberlike=thm,refname={Observation,Observations},Refname={Observation,Observations}]{observation}

\declaretheorem[style=definition,numberlike=thm,refname={Definition,Definitions},Refname={Definition,Definitions},name={Definition}]{defn}

\AtBeginDocument{%

\let\ref\Cref
}

\usepackage{footnote}
\makesavenoteenv{tabular}
\makesavenoteenv{table}

\makeatletter
\newcommand\footnoteref[1]{\protected@xdef\@thefnmark{\labelcref{#1}}\@footnotemark}
\makeatother

\ifdefined\ShowComment
\usepackage{showlabels}
\newcommand{\todo}[1]{{\bf \color{red} TODO: #1}}

\newcommand{\danupon}[1]{{\color{red} DN: #1}}
\newcommand{\thatchaphol}[1]{\textcolor{blue}{TS: #1}}

\newcommand{\jl}[1]{{\bf \color{ForestGreen} JL: #1}}
\newcommand{\yg}[1]{{\bf \color{green} YG: #1}}

\newcommand{\yp}[1]{{\bf \color{green} RP: #1}}

\newcommand{\jnote}[1]{{\color{purple}{\sc\bf{[JC #1]}}}}
\newcommand{\takeout}[1]{{#1}}
\else
\newcommand{\todo}[1]{}
\newcommand{\danupon}[1]{}
\newcommand{\yp}[1]{}
\newcommand{\thatchaphol}[1]{}
\newcommand{\takeout}[1]{}

\newcommand{\yg}[1]{}

\newcommand{\jl}[1]{}
\newcommand{\jnote}[1]{}
\fi

\makeatother

\usepackage[english]{babel}

\newcommand{\reals}{{\mathbb R}}
\newcommand{\cset}{{\mathcal{C}}}

\newcommand{\updateflow}{\mathtt{UpdateFlow}}
\newcommand{\computematch}{\mathtt{ComputeMatching}}
\newcommand{\process}{\mathtt{Process}}
\newcommand{\cutorcert}{{\textsc{CutOrCertify}}\xspace}
\newcommand{\matchorcut}{{\textsc{MatchOrCut}}\xspace}
\newcommand{\routeorcut}{{\textsc{RouteOrCut}}\xspace}
\newcommand{\routeorcutp}{{\textsc{RouteOrCut-1Pair}}\xspace}
\newcommand{\constructexpander}{{\textsc{ConstructExpander}}\xspace}
\newcommand{\extractexpander}{{\textsc{ExtractExpander}}\xspace}
\newcommand{\reducedegree}{{\textsc{ReduceDegree}}\xspace}
\newcommand{\makecanonical}{{\textsc{MakeCanonical}}\xspace}

\newcommand{\tN}{\tilde N}
\newcommand{\ceil}[1]{\ensuremath{\left\lceil#1\right\rceil}}
\newcommand{\floor}[1]{\ensuremath{\left\lfloor #1\right\rfloor}}
\newcommand{\myparskip}{3pt}
\parskip \myparskip
\newcommand{\tset}{\mathcal{T}}

\newcommand{\cCMG}{\ensuremath{{c}_{\mbox{\tiny{\sc{CMG}}}}}\xspace}

\newcommand{\SC}{{\sf Sparsest Cut}\xspace}
\newcommand{\MCC}{{\sf Minimum-Conductance Cut}\xspace}
\newcommand{\SF}{{\sf SF}\xspace}
\newcommand{\MSF}{{\sf MSF}\xspace}
\newcommand{\BC}{{\sf Minimum Balanced Cut}\xspace}
\newcommand{\DC}{{\sf Dynamic Connectivity}\xspace}

\newcommand{\set}[1]{\left\{ #1 \right\}}
\newcommand{\pset}{{\mathcal{P}}}
\newcommand{\aset}{{\mathcal{A}}}
\newcommand{\hset}{{\mathcal{H}}}
\newcommand{\hG}{\hat G}
\newenvironment{proofof}[1]{\noindent{\bf Proof of #1.}}

\newcommand{\rootx}{\operatorname{Root}}
\newcommand{\mincost}{\operatorname{MinCost}}
\newcommand{\parent}{\operatorname{Parent}}
\newcommand{\update}{\operatorname{Update}}
\newcommand{\link}{\operatorname{Link}}
\newcommand{\evert}{\operatorname{Evert}}
\newcommand{\cut}{\operatorname{Cut}}
\newcommand{\out}{\operatorname{out}}

\global\long\def\match{\text{\ensuremath{\textsc{match}}}}
\global\long\def\cong{{\ensuremath{\eta}}}
\global\long\def\cut{\text{\ensuremath{\textsc{cut}}}}
\global\long\def\aug{\text{\ensuremath{\textsc{aug}}}}
\global\long\def\aug{\text{\ensuremath{\textsc{aug}}}}
\global\long\def\slow{\text{\ensuremath{\textsc{slow}}}}
\global\long\def\ab{\mathrm{ab}}
\global\long\def\ex{\mathrm{ex}}
\global\long\def\msf{\mathsf{MSF}}

\global\long\def\CA{\text{\ensuremath{\textsc{CutAugment}}}}

\global\long\def\PushMatch{\text{\ensuremath{\textsc{Push-CutMatch}}}}
\global\long\def\MaxflowMatch{\text{\ensuremath{\textsc{Maxflow-CutMatch}}}}
\global\long\def\Match{\text{\ensuremath{\textsc{CutMatch}}}}
\global\long\def\Maxflow{\text{\ensuremath{\textsc{MaxFlow}}}}
\global\long\def\Prune{\text{\ensuremath{\textsc{Prune}}}}
\global\long\def\Cut{\text{\ensuremath{\textsc{CutPrune}}}}
\newcommand{\BCut}{{\sf BalCutPrune}\xspace}
\global\long\def\vBCut{\text{\textsc{vBalCutPrune}}}
\global\long\def\RecBCut{\text{\textsc{RecurBalCutPrune}}}
\global\long\def\SlowBCut{\text{\textsc{SlowBalCutPrune}}}

\global\long\def\vol{\mathrm{Vol}}
\global\long\def\volG{\mathrm{Vol}(G)}

\global\long\def\poly{\mathrm{poly}}
\global\long\def\polylog{\mathrm{polylog}}
\global\long\def\P{\mathcal{\mathcal{P}}}

\global\long\def\dist{\mathrm{dist}}

\global\long\def\KKOV{\textsc{kkov}}

\global\long\def\Bhat{\hat{B}}
\global\long\def\Ahat{\hat{A}}
\global\long\def\That{\hat{T}}
\global\long\def\Ghat{\hat{G}}
\global\long\def\Shat{\hat{S}}

\global\long\def\Lhat{\widehat{L}}
\global\long\def\Ohat{\widehat{O}}
\global\long\def\Otil{\tilde{O}}

\newcommand\opt{\mathsf{opt}}

\global\long\def\Acal{\mathcal{A}}
\newcommand\abs[1]{{\left|#1\right|}}
\newcommand\ww{\boldsymbol{\mathit{w}}}
\newcommand\rr{\boldsymbol{\mathit{r}}}
\newcommand\xx{\boldsymbol{\mathit{x}}}
\newcommand\yy{\boldsymbol{\mathit{y}}}
\newcommand\bb{\boldsymbol{\mathit{b}}}
\newcommand\cc{\boldsymbol{\mathit{c}}}
\newcommand\zz{\boldsymbol{\mathit{z}}}

\newcommand\dd{\boldsymbol{\mathit{d}}}
\newcommand\defeq{\coloneqq}
\newcommand\eqdef{\coloneqq}

\newcommand\eps{\epsilon}
\newcommand\norm[1]{{\left\lVert#1\right\rVert}}
\title{A Deterministic Algorithm for Balanced Cut \\
        with Applications to Dynamic Connectivity, Flows, and Beyond
}

\author{
Julia Chuzhoy\\
TTIC
\and
Yu Gao\\
Georgia Tech\footnote{part of this work was done while visiting MSR Redmond}
\and
Jason Li\\
CMU\footnote{supported in part by NSF award CCF-1907820}
\and
Danupon Nanongkai\\
KTH
\and
Richard Peng\\
Georgia Tech\footnotemark[1]
\and
Thatchaphol Saranurak\\
TTIC
}

\begin{document}

\makeatletter{\renewcommand*{\@makefnmark}{} \footnotetext{
                These results were obtained independently by Chuzhoy, and the group
                consisting of Gao, Li, Nanongkai, Peng, and Saranurak.
                Chronologically, Gao et al. obtained their result in July 2019,
                while Chuzhoy's result was obtained in September 2019, but there was no
                communication between the groups until early October 2019.
        }\makeatother}

\maketitle

\pagenumbering{gobble}
We consider the classical Minimum Balanced Cut problem: given a graph $G$, compute a partition of its vertices into two subsets of roughly equal volume, while minimizing the number of edges connecting the subsets. We present the first {\em deterministic, almost-linear time} approximation algorithm for this problem. Specifically, our algorithm, given an $n$-vertex $m$-edge graph $G$ and any parameter $1\leq r\leq O(\log n)$, computes a $(\log m)^{r^2}$-approximation for Minimum Balanced Cut on $G$, in time $O\left ( m^{1+O(1/r)+o(1)}\cdot (\log m)^{O(r^2)}\right )$. In particular, we obtain  a $(\log m)^{1/\epsilon}$-approximation  in time $m^{1+O(1/\sqrt{\epsilon})}$ for any constant $\epsilon$, and a $(\log m)^{f(m)}$-approximation in time $m^{1+o(1)}$, for any slowly growing function $m$. We obtain deterministic algorithms with similar guarantees for the Sparsest Cut and the Lowest-Conductance Cut problems.

Our algorithm for the Minimum Balanced Cut problem in fact provides a stronger guarantee: it either returns a balanced cut whose value is close to a given target value, or it certifies that such a cut does not exist by exhibiting a large subgraph of $G$ that has high conductance.
We use this algorithm to obtain deterministic algorithms for dynamic connectivity and minimum spanning forest, whose worst-case update time on an $n$-vertex graph is $n^{o(1)}$, thus resolving a major open problem in the area of dynamic graph algorithms. Our work also implies deterministic algorithms for a host of additional problems, whose time complexities match,  up to subpolynomial in $n$ factors,  those of known randomized algorithms. The implications include almost-linear time deterministic algorithms for solving Laplacian systems and for approximating maximum flows in undirected graphs. 

%Laplacian systems, maximum flows on undirected graphs, minimum-cost bipartite perfect matching, spectral sparsifiers

\newpage
\tableofcontents{}
\newpage

\pagenumbering{arabic}
\section{Introduction}\label{sec:intro}

%\danupon{Notations to discuss:
%
% $\Phi_G$ or $\Phi(G)$? $\vol_G$ or $\vol(G)$? (Vishnoi's book Lx=b use %the latter for both cases \url{https://theory.epfl.ch/vishnoi/Lxb-Web.pdf}) 
%
% Should we called our problem Balanced Edge-Separator as in Vishnois' book?
% 
% $G-S$ should be defined?}

%\paragraph{The Balanced Sparse Cut Problem and Applications.} In this problem, 
%

In the classical \BC problem, the input is an $n$-vertex graph $G=(V,E)$, and the goal is to compute a partition of $V$ into two subsets $A$ and $B$ with $\vol_G(A),\vol_G(B)\geq \vol(G)/3$, while minimizing $|E_G(A,B)|$; here, $E_G(A,B)$ is the set of all edges with one endpoint in $A$ and another in $B$, and, for a set $S$ of vertices of $G$, $\vol_G(S)$ denotes its \emph{volume} -- the sum of the degrees of all vertices of $S$. Lastly, $\vol(G)=\vol_G(V)$ is the total volume of the graph.
The \BC problem is closely related to the \MCC problem, where the goal is to compute a subset $S$ of vertices of minimum \emph{conductance}, defined as $|E_G(S,V\setminus S)|/\min\{\vol_G(S),\vol_G(V\setminus S)\}$, and to the \SC problem, where the goal is to compute a subset $S$ of vertices of minimum \emph{sparsity}: $|E_G(S,V\setminus S)|/\min\{|S|,|V\setminus S|\}$.
While all three problems are known to be NP-hard, approximation algorithms for them are among the most central and widely used tools in algorithm design, especially due to their natural connections to the hierarchical divide-and-conquer paradigm 
~\cite{Racke02,SpielmanT04,Trevisan05,AroraHK10,RackeST14,KawarabayashiT19,NanongkaiSW17}. We note that approximation algorithms for \BC often consider a relaxed (or a bi-criteria) version, where we only require that the solution $(A,B)$ returned by the algorithm satisfies $\vol_G(A),\vol_G(B)\geq \vol(G)/4$, but the solution value is compared to that of the optimal balanced cut.

%
%we want to find a partition of the vertices of the graph
%so that (i) the two sides are of roughly equal size,
%and (ii) few edges are between the pieces.
%Such cuts are often referred to as {\em sparse cuts}
%because the two conditions together imply that the
%subset of vertices have small surface area to volume ratios.
%While both finding the most sparse and
%most balanced sparse cuts are NP-hard,
%approximate versions of these problems have widespread
%applications in algorithm design.
%In particular, approximate (balanced) sparse cut primitives
%are now central to the study of efficient graph algorithms
%due to their close connections with hierarchical
%divide-and-conquer schemes on general
%graphs~\cite{Racke02,SpielmanT04,Trevisan05,AroraHK10,RackeST14,KawarabayashiT19,NanongkaiSW17}.

%

The first approximation algorithm for \BC, whose running time is near-linear in the graph size, was developed in the seminal work of Spielman and Teng \cite{SpielmanT04}.  %
%,SpielmanT13-FirstJournal,SpielmanT11-SecondJournal,SpielmanT14-ThirdJournal}
This algorithm was used in \cite{SpielmanT04} in order to decompose a given graph into a collection of ``near-expanders'', which are then exploited in order to construct spectral sparsifiers, eventually leading to an algorithm for solving systems of linear equations in near-linear time.
Algorithms for \BC also served as crucial building blocks in the more recent breakthrough results that designed near- and almost-linear time\footnote{We informally say that an algorithm runs in near-linear time, if its running time is $O(m\cdot \poly\log n)$, where $m$ and $n$ are the number of edges and vertices in the input graph, respectively. We say that the running time is almost-linear, if it is bounded by $m^{1+o(1)}$.} approximation algorithms for a large class of flow and regression problems \cite{Sherman13,KelnerLOS14,Peng16,KyngPSW19} and faster exact algorithms for maximum flow, shortest paths with negative weights, and minimum-cost flow \cite{CohenMSV17,Madry16}. 
Spielman and Teng's expander decomposition was later strengthened by Nanongkai, Saranurak and Wulff-Nilsen \cite{NanongkaiSW17,Wulff-Nilsen17,NanongkaiS17}, who used it to obtain algorithms  for the dynamic minimum spanning forest problem with improved worst-case update time. The fastest current algorithm for computing expander decompositions is due to Saranurak and Wang \cite{SaranurakW19};  a similar decomposition was recently used by Chuzhoy and Khanna \cite{ChuzhoyK19} in their algorithm for the decremental single-source shortest paths problem, that in turn led to a faster algorithm for approximate vertex-capacitated maximum flow.

Unfortunately, all algorithms mentioned above are {\em randomized}. %This is mainly because Spielman and Teng's nearly-linear time algorithm for the balanced sparse cut problem is randomized. 
This is mainly because all existing almost- and near-linear time algorithms for \BC are randomized \cite{SpielmanT04,KhandekarRV09}.
A fundamental open question in this area is then: is there a {\em deterministic} algorithm for \BC with similar performance guarantees? Resolving this questions seems a key step to obtaining fast deterministic algorithms for all aforementioned problems, and to resolving one of the most prominent open problems in the area of dynamic graph algorithms, namely, whether there is a deterministic algorithm for \DC, whose worst-case update time is smaller than the classical $O(\sqrt{n})$ bound of Frederickson \cite{Frederickson85,EppsteinGIN97} by a factor that is polynomial in $n$. 

The best previous published bound on the running time of a determinsitic algorithm for \BC is $O(mn)$~\cite{AndersenCL07}.
A recent manuscript by a subset of the authors, together
with Yingchareonthawornchai~\cite{GaoLNPSY19}, obtains a running time of  $\min\set{n^{\omega+o(1)}, m^{1.5+o(1)}}$,
where $\omega < 2.372$ is the matrix multiplication exponent,
and $n$ and $m$ are the number of nodes and
edges of the input graph, respectively.
This algorithm is used in \cite{GaoLNPSY19} to obtain
faster deterministic algorithms for the vertex connectivity problem. 
However, the running time of the algorithm of \cite{GaoLNPSY19} for \BC is somewhat slow, and it just
falls short of breaking the $O(\sqrt{n})$
worst-case update time bound for \DC. 

\subsection{Our Results}
 We present a deterministic (bi-criteria) algorithm for \BC that, for any 
 parameter $r=O(\log  n)$, achieves an approximation factor $\alpha(r)=(\log m)^{r^2}$ in time $T(r)=O\left (m^{1+O(1/r)+o(1)}\cdot (\log m)^{O(r^2)}\right )$, where $n$ and $m$ are the number of vertices and edges in the input graph, respectively.
 In particular, for any constant $\epsilon$, the algorithm achieves $(\log m)^{1/\epsilon}$-approximation in time $O\left (m^{1+O(1/\sqrt{\epsilon})}\right )$. For any slowly growing function $f(m)$ (for example, $f(m)=\log\log m$ or $f(m)=\log^*m$), it achieves $(\log m)^{f(m)}$-approximation in time $m^{1+o(1)}$. 
 
 In fact our algorithm provides somewhat stronger guarantees: it either computes an almost balanced cut whose value is within an $\alpha(r)$ factor of a given target value $\eta$; or it certifies that every balanced cut in $G$ has value $\Omega(\eta)$, by producing a large sub-graph of $G$ that has a large conductance.
This algorithm implies fast deterministic algorithms for all
the above mentioned problems, including, in particular, improved worst-case update time guarantees for 
(undirected) \DC and {\sf Minimum Spanning Forest}.

\iffalse

The approximation nature of the algorithm means it returns
either a balanced cut, or a certificate that none exists
(in a step which we term `pruning').
This definition is identical to the one used in~\cite{SaranurakW19},
and we will define it formally in \Cref{def:intro:BCut} below\yp{give section pointer?}.
The formal statement of our almost-linear time algorithm
is in \Cref{thm:intro:main}. 

\begin{thm}[Main Result]\label{thm:intro:main}
        %There exists a constant $c<1$ such there for any $m$ and $\phi_L,\phi_U\in (0,1]$ where $\frac{\phi_U}{\phi_L}\geq (\log m)^{(\frac{\log m}{\log\log m})^c}$, we have
        For the $\BCut$ problem as defined below
        in \Cref{def:intro:BCut},
        there is a deterministic algorithm such that for any graph $G$ with $n$ vertices, $m$ edges,
        and lower/upper conductance thresholds $\phi_L,\phi_U\in (0,1]$ such that $\frac{\phi_U}{\phi_L}\geq 2^{O(\log^{2/3}n\cdot(\log\log n)^{1/3})}$,
        solves $\BCut$ in $m^{1+o(1)}$ time.
        %\danupon{Replace $1/100$ with some other number?}
\end{thm}
\fi

%\subsection{Formal Statement of Problem and Result}

In order to provide more details on our results and techniques, we need to introduce some notation. Throughout, we assume that we are given an $m$-edge, $n$-node undirected graph, denoted by $G = (V, E)$. A \emph{cut} in $G$ is a partition $(A,B)$ of $V$ into two non-empty subsets; abusing the notation, we will also refer to subsets $S$ of vertices with $S\neq \emptyset,V$ as cuts, meaning the partition $(S,V\setminus S)$ of $V$.
%We denote by $E(S, V \setminus S)$ the set of edges with exactly one endpoint in $S$, and we refer to these edges as \emph{boundary edges} of $S$. 
%\thatchaphol{I remove this because I defined in the first paragraph instead. The definition is needed even there.}
The \emph{conductance} of a cut $S$ in $G$, that was already mentioned above, is defined as:
\[
\Phi_G\left( S \right)
\eqdef
\frac{|E_G\left(S, V\setminus S\right)|}{\min\set{\vol_G\left(S\right), \vol_G\left(V\setminus S\right)}},
\]
and the conductance of a graph $G$,
that we denote by $\Phi(G)$, 
is the smallest conductance of any cut $S$ of $G$:
$
\Phi\left(G\right)
\eqdef
\min_{S \subsetneq V: S \neq \emptyset} \set{\Phi_G(S)}.
$

A notion that is closely related to conductance is that of sparsity. The \emph{sparsity} of a cut $S$ in $G$ is: $
\Psi_G\left( S \right)
\eqdef
\frac{|E_G\left(S, V\setminus S\right)|}{\min\set{|S|,|V\setminus S|}}$,
and the \emph{expansion} of the graph $G$ is the minimum sparsity of any cut $S$ in $G$: $
\Psi\left(G\right)
\eqdef
\min_{S \subsetneq V: S \neq \emptyset} \set{\Psi_G(S)}.
$

%That is, when the parameter in $(\cdot)$ is a subset of vertices,
%this notion measures the conductance of a set of vertices,
%but when its a graph, it's minimizing over all subsets of vertices.
%\yp{notion set up this way feels slightly confusing,
%        but I don't see another way to have degrees involved.} \jnote{if we switch to new notation we don't need this}
We say that a cut $S$ is \emph{balanced} if $\vol_G(S),\vol_G(V\setminus S)\geq \vol(G)/3$.
The main tool that we use in our approximation algorithm for the \BC problem is the $\BCut$ problem, that is defined next.  Informally, the problem seeks to either find a low-conductance balanced cut in a given graph, or to produce a certificate that every balanced cut has a high conductance, by exhibiting a large sub-graph of $G$ that has a high conductance. 

\begin{defn}[$\BCut$ problem]
	\label{def:intro:BCut}    
	The input to the $\alpha$-approximate $\BCut$ problem is a graph $G=(V,E)$, a conductance parameter $0<\phi\leq 1$, and an approximation factor $\alpha$. The goal is to compute a cut $(A,B)$ in $G$, 
	with $|E_G(A,B)|\leq \alpha \phi\cdot \vol(G)$, such that one of the following hold:
	either 
	%       Let $\Cut$ be an routine that, given $(G,\phi_{U},\phi_{L},\beta)$
	%       where $G=(V,E)$ is an $m$-edge graph, $\phi_{U},\phi_{L}\in[0,1]$
	%       where $\phi_{L}\le\phi_{U}$ and $\beta\le1/3$, returns either 
	\begin{enumerate}[noitemsep]
		\item \textbf{(Cut)} \label{enu:intro:CP cut} $\vol_{G}(A),\vol_G(B)\ge \vol(G)/3$; or
		\item \textbf{(Prune)} \label{enu:CP prune} $\vol_G(A)\geq \vol(G)/2$, and graph $G[A]$ has conductance at least $\phi$.
	\end{enumerate}
\end{defn}

%In fact, we will use a slightly stronger condition for the (Prune) case: Let $H$ be the graph obtained from $G$ by deleting all vertices of $S$ from it, and adding self-loops to the remaining vertices as needed, so that the degree of every vertex $v\in V(H)$ in $H$ is the same as its degree in $G$. We will replace the requirement $\Phi_{G}(G-S)\ge\alpha_{L}$ with the stronger requirement that $\Phi_H(H)\geq \alpha_L$.
Our main technical result is the following.

\begin{thm}[Main Result]\label{thm:intro:main}
	There is a deterministic algorithm, that, given a graph $G$ with $m$ edges,
	and parameters $\phi\in (0,1]$, $1\leq r\leq O(\log n)$, and $\alpha =(\log m)^{r^2}$,
	computes a solution to the $\alpha$-approximate $\BCut$ problem instance $(G,\phi)$ in time $O\left ( m^{1+O(1/r)+o(1)}\cdot (\log m)^{O(r^2)}\right )$.
	\thatchaphol{Maybe good to say that the factor $m^{o(1)}$ is $\exp\left(O\left( \left( \log{m} \log\log{m}\right)^{3/4} \right) \right)$.}
\end{thm}	
	
	 In particular, by  letting $r$ be a large constant, we obtain a $(\log m)^{1/\epsilon}$-approximation 
in time $m^{1+O(1/\sqrt{\epsilon})}$ for any constant $\epsilon$, 
and by letting $f(m)$ be any slowly growing function (for example, $f(m)=\log\log m$ or $f(m)=\log^*m$), and setting $r=\sqrt{f(m)}$, we obtain $(\log m)^{f(m)}$-approximation in time $m^{1+o(1)}$.

The algorithm from \ref{thm:intro:main} immediately implies a deterministic bi-criteria factor-$(\log n)^{r^2}$-approximation algorithm for \BC, with running time $O\left ( m^{1+O(1/r)+o(1)}\cdot (\log m)^{O(r^2)}\right )$ for any value of $r\leq O(\log m)$. Indeed, suppose we are given any conductance parameter $0<\phi<1$ for an input graph $G=(V,E)$. We apply the algorithm from  \ref{thm:intro:main} to graph $G$ with the parameter $\phi$, obtaining a cut $(A,B)$ in $G$, with $|E_G(A,B)|\leq \alpha \phi\cdot \vol(G)$. If $\vol_G(A),\vol_G(B)\geq \vol(G)/4$, then we obtain  an  (almost) balanced cut $(A,B)$ of conductance at most $\alpha \phi$. Otherwise,  we are guaranteed that
$\vol_G(A)\geq 3\vol(G)/4$, and graph $G[A]$ has conductance at least $\phi$.
We claim that in this case, for any balanced cut $(A',B')$ in $G$, $|E_G(A',B')|\geq \Omega(\phi\cdot \vol(G))$ holds. This is because any such partition $(A',B')$ of $V$ defines a partition $(X,Y)$ of $A$, with $\vol_G(X),\vol_G(Y)\geq \Omega(\vol(G))$, and, since $\Phi(G[A])\ge\phi$, we get that $|E_G(X,Y)|\geq \Omega(\phi\cdot \vol(G))$. 
Therefore, we obtain the following corollary.
 
%given any conductance parameter $0<\phi<1$, our algorithm either computes a balanced cut $(A,B)$ with $|E_G(S, V\setminus S)|\le\alpha(r)\phi\cdot \vol(G)$, or correctly certifies that for any balanced cut $(A,B)$ in $G$, $|E_G(A,B)|\geq \Omega(\phi\cdot \vol(G))$. By combining this with a standard binary search over the values of the parameter $\phi$, we 
% obtain the following corollary:

\begin{cor}\label{cor: main for BC}
        There is an algorithm that, given an $n$-vertex $m$-edge graph $G$, a target value $\eta$ and a parameter $r\leq O(\log n)$, either returns a partition  $(A,B)$ of $V(G)$ with $\vol_G(A),\vol_G(B)\geq \vol(G)/4$ and $|E_G(A,B)|\leq \alpha(r)\cdot \eta$, for $\alpha(r)=(\log m)^{r^2}$, or it certifies that for any partition  $(A,B)$ of $V(G)$ with $\vol_G(A),\vol_G(B)\geq \vol(G)/3$, $|E_G(A,B)|> \eta$ must hold. The running time of the algorithm is $O\left ( m^{1+O(1/r)+o(1)}\cdot (\log m)^{O(r^2)}\right )$.
        \end{cor}

%\thatchaphol{I comment this out because it is quite repetitive}
%Note that in the $\BCut$ problem, we are required to either (1) output a cut $S$ that is balanced ($\vol_{G}(S),\vol_G(V\setminus S)\ge \vol(G)/3$) and has a low conductance ($\Phi_{G}(S)\le\phi_{U}$), or (2) certify that, once we remove, or ``prune'', a relatively small set $S$ of vertices from $G$, we obtain a graph of relatively high conductance ($\Phi(G -S)\ge\phi_{L}$). Moreover, in the latter case, there are relatively few edges connecting $S$ to the high-conductance graph ($E_G(S, V\setminus S)\le\phi_{U}\vol(G)$). Therefore, in a sense, an algorithm for the $\BCut$ problem provides stronger guarantees than those needed for solving the \BC problem: if the value of the  \BC is higher than the given threshold, then the certificate that we provide, in the form of a large high-conductance subgraph of $G$, can be exploited directly in various algorithms.
Algorithms for \BC often differ in the type of certificate that they provide when the value of the \BC is greater than the given threshold (that corresponds to the (Prune) case in \Cref{def:intro:BCut}). The original near-linear time algorithm of Spielman and Teng \cite{SpielmanT04} outputs a set $S$ of nodes of small volume, with the guarantee that for some subset $S' \subseteq S$, the graph $G- S'$ has high conductance. This guarantee, however, is not sufficient for several applications.
A version that was found to be more useful in several recent applications, such as e.g. \DC \cite{SaranurakW19,NanongkaiSW17,Wulff-Nilsen17,NanongkaiS17}, is somewhat similar to that in the definition of $\BCut$, but with a somewhat stronger guarantee in the (Prune) case\footnote{To be precise, that version requires that $|E_G(A, B)|\le \alpha\phi\cdot \vol_{G}(B)$, which is somewhat stronger than our requirement that  $|E_G(A, B)|\le \alpha\cdot \phi \cdot \vol(G)$. But for all applications we consider, our guarantee still suffices, possibly because the two guarantees are essentially the same when the cut $(A,B)$ is balanced.}.

%The ratio $\frac{\phi_U}{\phi_L}$ is sometimes referred to as the approximation factor, and for most applications $\frac{\phi_U}{\phi_L}=\poly(\frac{1}{\phi_U},\log{n})$ and $\phi_L=1/n^{o(1)}$ suffice. 
The approximation factor $\alpha$ of Spielman and Teng's algorithm \cite{SpielmanT04} depends on the parameter $\phi$, and its time complexity depends on both $\phi$ and $\alpha$. %The guarantees that they provide are sufficiently strong for most applications, including all applications discussed in this paper. 
Several subsequent papers have improved the approximation factor or the time complexity of their algorithm e.g. \cite{KhandekarRV09,AndersenCL07,OrecchiaV11,OrecchiaSV12,Madry10}; we do not discuss these results here since they are not directly related to this work.

\subsection{Applications}

An immediate consequence of our results is deterministic algorithms for the Sparsest Cut and the Lowest-Conductance Cut problems, summarized in the next theorem.

\begin{theorem}\label{thm: sparsest and low cond}
	There is a deterministic algorithm, that, given an $n$-vertex and $m$-edge graph $G$, and a parameter $r\leq O(\log n)$, computes a $(\log n)^{r^2}$-approximate solution for the Sparest Cut problem on $G$, in time $O\left (m^{1+O(1/r)+o(1)}\cdot (\log n)^{O(r^2)}\right )$. Similarly, there is a deterministic algorithm that achieves similar performance guarantees for the Lowest-Conductance Cut problem.
\end{theorem}

We note that the best current deterministic approximation algorithm for both Sparsest Cut and Lowest-Conductance Cut, due to Arora, Rao and Vazirani  \cite{AroraRV09}, achieves an $O(\sqrt{\log n})$-approximation. Unfortunately, the algorithm has a large (but polynomially bounded) running time, since it needs to solve an SDP deterministically.
There are faster $O(\log n)$-approximation deterministic algorithms for both the Sparsest Cut and the Lowest-Conductance problems, with running time $\Otil(m^2)$, that are based on the Multiplicative Weights Update framework \cite{Fleischer00,Karakostas08}. If we allow the approximation ratio to depend on $\phi$, where $\phi$ is the value of the optimal solution, then there are several algorithms that are based on spectral approach for both problems. The algorithm of  \cite{Alon86} computes  a cut with conductance at most $O(\phi^{1/2})$ in time $\Otil(n^\omega)$. Using Personalized PageRank algorithm \cite{AndersenCL07}, a cut of conductance at most $O(\phi^{1/2})$ can found in time $\Otil(mn)$. 
Recently, Gao et.~al \cite{GaoLNPSY19} provided an algorithm to compute a cut of conductance at most $\phi^{1/2}n^{o(1)}$, in time $O(m^{1.5+o(1)})$.\footnote{The last two algorithms, in fact, provide additional guarantees regarding the balance of the returned cut.}

%	 and the Low-Conductance Cut problems, in time $O\left (m^{1+O(1/r)+o(1)}\cdot (\log n)^{O(r^2)}\right )$, where $n$ and $m$ are the number of vertices and edges in the input graph $G$. 

Additionally, we obtain faster deterministic algorithms for a number of other cut and flow problems; the performance of our algorithms matches that of the best current randomized algorithms, to within factor $n^{o(1)}$. We summarize these bounds in \Cref{tabel:applications-static} and \Cref{tabel:applications-dyn}; see \ref{sec:app} and \ref{sec:app nophi} for a more detailed discussion. 
%
%
%
%, including expander decomposition, spectral sparsifiers, vertex- and edge-capacitated maximum flows on undirected graphs, minimum-cost bipartite perfect matching, and solving Laplacian systems. The performance of our algorithms match the existing randomized counterparts up to subpolynomial factors. We discuss these applications in details in \Cref{sec:app} (see \Cref{tabel:applications}  for a summary). 
%
We now turn to discuss the implications of our results to the \DC problem, which was one of the main motivations of this work.

\begin{table}
{\small
\begin{tabular}{|>{\centering}p{0.26\textwidth}|>{\centering}p{0.23\textwidth}|>{\centering}p{0.20\textwidth}|>{\centering}p{0.21\textwidth}|}
\hline 
Problem & Best previous running time: deterministic & Best previous running time: randomized & Our results: deterministic \tabularnewline
\hline 
\hline 

%took-out Max-flow (high accuracy) & $\tilde{O}(m\min\{m^{1/2},n^{2/3}\})$ \cite{GoldbergR98}
%& $\Otil(m\sqrt{n})$ \cite{LeeS14}\footnoteref{foot:sparseOnly} 
%{\footnotesize{}(Note: best for sparse graphs)} 
%& $\Ohat(m^{3/2})$: discussion after \ref{cor:laplacian}\tabularnewline
%\hline
$(1 + \epsilon)$-approximate undirected max-flow/min-cut
& $\tilde{O}(m\min\{m^{1/2},n^{2/3}\})$ \cite{GoldbergR98} (exact) & $\tilde{O}(m\epsilon^{-1})$ \cite{Sherman17,KelnerLOS14,Madry10}
    & $\Ohat(m\epsilon^{-2})$: \ref{cor:maxflow} \tabularnewline
\hline 
\hline 
%Generalized (min-cost) max flows \todo{see comments below (next page?)} & $\Omega(m^{2})$ \cite{Vaidya89a,GoldfarbJO97} & $\Otil(m^{3/2})$ \cite{DaitchS08} & $\Ohat(m^{3/2})$
%
%see discussion after \ref{cor:laplacian}\tabularnewline
%\hline 
$n^{o(1)}$-approximate sparsest cut  & $\Ohat(m^{1.5})$ \cite{GaoLNPSY19} & $\tilde{O}(m)$ \cite{KhandekarRV09,Sherman13} & $\Ohat(m)$: \ref{thm: sparsest and low cond} \tabularnewline
\hline 
$n^{o(1)}$-approximate lowest-conductance cut  & $\Ohat(m^{1.5})$ \cite{GaoLNPSY19} & $\tilde{O}(m)$ \cite{KhandekarRV09,Sherman13} & $\Ohat(m)$:  \ref{thm: sparsest and low cond}\tabularnewline
%\hline 
%$n^{o(1)}$-approximate vertex-capacitated sparsest cut (value $\alpha$) & $\Ohat(m^{1.5})$ \cite{GaoLNPSY19} & $\Ohat(n^2)$ \cite{ChuzhoyK19}\\ $O(m\sqrt{n})$~\cite{GaoLNPSY19}  & $\Ohat(\min\{m\sqrt{n},m \alpha^{-1}, n^2\})$
%\ref{cor:vexpansion_ck,cor:vertex expansion} \tabularnewline
\hline 
Expander decomposition (conductance $\phi$)
 & $\Ohat(m^{1.5})$ \cite{GaoLNPSY19} & $\tilde{O}(m/ \phi)$ \cite{SaranurakW19} \\ $\Ohat(m)$ \cite{NanongkaiS17, Wulff-Nilsen17} & $\Ohat(m)$ \ref{cor:expander decomp2}\tabularnewline
\hline 
Congestion approximator & $\Omega(m^{2})$  & $\tilde{O}(m)$ \cite{Madry10,Sherman13,KelnerLOS14} & $\Ohat(m)$ 
\ref{lem:det15sherman13} \tabularnewline
\hline 
%Dynamic connectivity \footnotesize{(worst-case update time)} & $O(\sqrt{n})$ \cite{Frederickson85,EppsteinGIN97}\\$O(\sqrt{n}\cdot\frac{\log\log n}{\sqrt{\log n}})$
%\cite{Kejlberg-Rasmussen16} & $O(\log^{4}n)$ \cite{KapronKM13,GibbKKT15} & $\Ohat(1)$ \\ \ref{cor:dynconn}\tabularnewline
%\hline 
%Dynamic minimum spanning forests  \footnotesize{(worst-case update time)}& $O(\sqrt{n})$ \cite{Frederickson85,EppsteinGIN97} & $\Ohat(1)$ \cite{NanongkaiSW17} & $\Ohat(1)$ \\ \ref{cor:dynconn}\tabularnewline
%\hline 
Spectral sparsifiers  & $O(mn^{3}\epsilon^{-2})$ \cite{BatsonSS12,SilvaHS16} & $\tilde{O}(m\epsilon^{-2})$ \cite{SpielmanT11-SecondJournal,Lee017} &
$n^{o(1)}$-approximation, time $\Ohat(m)$: \ref{cor:sparsifier} \tabularnewline
\hline 
Laplacian solvers & $\tilde{O}(m^{1.31}\log(1/\epsilon))$ \cite{SpielmanT03}  & $\tilde{O}(m\log(1/\epsilon))$ \cite{SpielmanT14-ThirdJournal} & $\Ohat(m\log(1/\epsilon))$ \\ \ref{cor:laplacian}\tabularnewline
\hline 
%took-out Single-source shortest paths with negative weights & $\tilde{O}(m\sqrt{n})$ \cite{GabowT89,Goldberg95}  & $\tilde{O}(m^{10/7})$ \cite{CohenMSV17}\footnote{\label{foot:sparseOnly} Note: To simplify the discussion, we only state best bounds for sparse graphs.}
%{\footnotesize{}(Note: We only state best bound for sparse graphs)} 
%& $\Ohat(m^{10/7})$ following \ref{cor:laplacian}\tabularnewline
%\hline 

%$(1+\epsilon)$-approx.~max flows & $\tilde{O}(m\min\{m^{1/2},n^{2/3}\})$ \cite{GoldbergR98} {\footnotesize{}Note:
%       exact } & $\tilde{O}(m\epsilon^{-1})$ \cite{Sherman17} & $\Ohat(m\epsilon^{-1})$ \ref{cor:maxflow}\tabularnewline
%\hline 
\end{tabular}
}
\caption{Applications of our results to static graph problems. As usual, $n$ and $m$ denote the number of nodes and edges of the input graph, respectively. We use $\tilde{O}$ and $\Ohat$ notation to hide $\polylog n$ and $n^{o(1)}$ factors respectively. For readability, we assume that  the weights and the capacities of edges are polynomial in the problem size. \label{tabel:applications-static}} %\thatchaphol{The best randomized max flow takes $m^{11/8}$ instead of $m^{10/7}$ now. Should we check it if we can match?}}\jl{cite new $m^{4/3}$ now by Liu-Sidford}\jnote{what is "max-flow high accuracy? Which line of the table do these new algorithms refer to, and do we have the corresponding papers in the biblio and references in the paper body?}
\end{table}

\begin{table}
        {\small
                \begin{tabular}{|>{\centering}p{0.25\textwidth}|>{\centering}p{0.25\textwidth}|>{\centering}p{0.23\textwidth}|>{\centering}p{0.17\textwidth}|}
                        \hline 
                        Dyanmic Problem & Best previous worst-case update time: deterministic & Best previous worst-case update time: randomized & Our results: deterministic\tabularnewline
                        \hline 
                        \hline 
                        Connectivity & $O(\sqrt{n})$ \cite{Frederickson85,EppsteinGIN97}\\$O(\sqrt{n}\cdot\frac{\log\log n}{\sqrt{\log n}})$
                        \cite{Kejlberg-Rasmussen16} & $O(\log^{4}n)$ \cite{KapronKM13,GibbKKT15} & $n^{o(1)}$ \ref{cor:dynconn}\tabularnewline
                        \hline 
                        Minimum Spanning Forest & $O(\sqrt{n})$ \cite{Frederickson85,EppsteinGIN97} & $n^{o(1)}$ \cite{NanongkaiSW17} & $n^{o(1)}$  \ref{cor:dynconn}\tabularnewline
                        \hline 
                \end{tabular}
        }
        \caption{Applications of our results to dynamic (undirected) graph problems. As before, $n$ and $m$ denote the number of vertices and edges of the input graph, respectively. %We use $\tilde{O}$ and $\Ohat$ notation to hide $\polylog n$ and $n^{o(1)}$ factors respectively. 
        	For readability, we assume that the weights and the capacities of edges/nodes are polynomial in problem size.}\label{tabel:applications-dyn}
\end{table}

%We discuss most applications in Section XXX, and only focus on the dynamic connectivity problem here. This is the original motivation of this paper. 

In the most basic version of the \DC problem, we are given a graph $G$ that undergoes edge deletions and insertions, and the goal is to maintain the information of whether $G$ is connected. The \DC problem and its generalizations -- dynamic Spanning Forest (\SF) and dynamic Minimum Spanning Forest (\MSF) -- have played a central role in the development of the area of dynamic graph algorithms for over three decades (see, e.g.,  \cite{NanongkaiS17,NanongkaiSW17} for further discussions).

An important measure of the performance of a dynamic algorithm is its {\em update time} -- the amount of time that is needed in order to process each update (an insertion or a deletion of an edge).
We distinguish between \emph{amortized update time}, that upper-bounds the average time that the algorithm spends on each update, and {\em worst-case update time}, that upper-bounds the largest amount of time that the algorithm ever spends on a single update.

%
%%%
%An update time that holds for every single update is called {\em worst-case}. This is to contrast with an {\em amortized update time} which holds ``on average''. 
%
The first non-trivial algorithm for the \DC problem dates back to Frederickson's work from 1985 \cite{Frederickson85}, that provided a deterministic algorithm with $O(\sqrt{m})$ worst-case update time.  
Combining this algorithm with the sparsification technique of Eppstein et al.~\cite{EppsteinGIN97} yields a deterministic algorithm for \DC with $O(\sqrt{n})$  worst-case update time.  Improving  and refining this bound has been an active research direction in the past three decades, but unfortunately, practically all follow-up results require either {\em randomization} or \emph{amortization}. We now provide a summary of these results.
%
%guarantee only amortized update time complexities or require the randomness: (i) 

\begin{itemize}[noitemsep]
\item 
{\bf (Amortized \& Randomized)} In their 1995 breakthrough paper, Henzinger and King \cite{HenzingerK99} greatly improve the $O(\sqrt{n})$ worst-case update bound with  a randomized Las Vegas algorithm, whose expected amortized update time is $\poly\log(n)$. 
%
%By allowing the update time to be amortized and the algorithm to be randomized, Henzinger and King \cite{HenzingerK99} presented in their 1995 breakthrough a randomized algorithm with polylogarithmic amortized update time. 
%
This result has been subsequently improved, and the current best randomized algorithms have amortized update time that almost matches existing lower bounds, to within $O((\log\log n)^2)$ factors; see, e.g., \cite{HuangHKP-SODA17,Thorup00,HenzingerT97,PatrascuD06}. 

\item
{\bf (Amortized \& Deterministic)} Henzinger and King's 1997 deterministic algorithm \cite{HenzingerK-ICALP97} achieves an amortized update time of $O(n^{1/3}\log n)$. This was later substantially improved to $O(\log^2 n)$ amortized update time by the deterministic algorithm of Holm, de Lichtenberg, and Thorup~\cite{HolmLT01}; this update time was in turn later improved to $O(\log^2(n)/\log\log n)$ by  Wulff-Nilsen \cite{Wulff-Nilsen13a}. 
%
%, and algorithms that work for harder problems such as two-edge connectivity and MSF (e.g. \cite{HolmLT01,HenzingerK01}).

\item 
{\bf (Worst-Case \& Randomized)} The first improvement over the $O(\sqrt{n})$ worst-case update bound  was due to Kapron, King and
Mountjoy \cite{KapronKM13}, who provided a randomized Monte Carlo algorithm  with worst-case update time $O(\log^{5}n)$. This bound was later improved to $O(\log^{4}n)$ by Gibb et al. \cite{GibbKKT15}. Subsequently, Nanongkai, Saranurak, and Wulff-Nilsen \cite{NanongkaiSW17,Wulff-Nilsen17,NanongkaiS17} presented a Las Vegas algorithm for the more general dynamic \MSF problem with $n^{o(1)}$ worst-case update time. 
\end{itemize}

A major open problem that was raised repeatedly  (see, e.g., \cite{KapronKM13,PatrascuT07,Kejlberg-Rasmussen16,King-Ency16,King-Ency08,HolmLT01,Wulff-Nilsen17}) is: {\em can we achieve an $O(n^{1/2-\epsilon})$ worst-case update time with a deterministic algorithm?} The only progress so far on this question is the deterministic algorithm of Kejlberg-Rasmussen et~al.~\cite{Kejlberg-Rasmussen16}, that slightly improves the $O(\sqrt{n})$ worst-case update time bound to $O(\sqrt{n(\log\log n)^{2}/\log n})$ using word-parallelism.
In this paper, we resolve this question in the affirmative, and provide a somewhat stronger result, that holds for the more general dynamic \MSF problem:

\begin{thm}
        \label{thm:intro:dynConn} There are  deterministic algorithms for \DC and $\mathsf{Dynamic}$
        \MSF, 
%       that, given an $n$-node graph, report the change to the maintained minimum Spanning Forest after each edge insertion or deletion in 
        with $n^{o(1)}$ worst-case update time.  
\end{thm}

%\yp{what about removing the rows related to graph partitioning in the table too? that way the unpublished stuff is featured less prominently}
%\jnote{I deleted generalized max flows. I am fine with Richard's suggestions to delete rows related to graph partitioning but will leave it to others to decide.}\danupon{I'm also neutral. Richard, feel free to remove these rows.}

%Observe that our update time holds even for the more-general dynamic MSF problem. 
In order to obtain this result, we use the algorithm of Nanongkai, Saranurak, and Wulff-Nilsen \cite{NanongkaiSW17} for dynamic \MSF. The only randomized component of their algorithm is the computation of an expander decomposition of a given graph. Since our results provide a fast deterministic algorithm for computing expander decomposition, we  achieve the same $n^{o(1)}$ worst-case update time as in \cite{NanongkaiSW17} via a deterministic algorithm.

\subsection{Techniques}

%\paragraph{Techniques.} 
%
%Our algorithm is the first recursive implementation of the cut-matching game framework, with the key insight being to make the ``matching'' player solve a harder problem similar to the multi-commodity flow problem. 
%
%Our result follows from the ability to implement the {\em cut-matching game framework} in a {\em recursive} manner. This framework was 
Our algorithm for the proof of \ref{thm:intro:main} is based on the {\em cut-matching game} framework that was introduced by Khandekar, Rao and Vazirani \cite{KhandekarRV09}, 
%, Rao, and Vazirani \cite{KhandekarRV09} (made explicit by  Khandekar~et~al.~\cite{KhandekarKOV2007cut}), 
and has been used in numerous algorithms for computing sparse cuts \cite{KhandekarRV09,NanongkaiS17,SaranurakW19,GaoLNPSY19} and beyond (e.g.~\cite{ChekuriC13,RackeST14,ChekuriC16,ChuzhoyL16}). 
Intuitively, the cut-matching game consists of two algorithms: one algorithm, called the \emph{cut player}, needs to compute a balanced cut of a given graph that has a small value, if such a cut exists. The second algorithm, called the \emph{matching player}, needs to solve (possibly approximately) a single-commodity maximum flow / minimum cut problem. A combination of these two algorithms is then used in order to compute a sparse cut in the input graph, or to certify that no such cut exists.
Unfortunately, all current algorithms for the cut player are randomized. 
Our main technical contribution is an efficient deterministic algorithm that implements the cut player. The algorithm itself is recursive, and proceeds by recursively running many cut-matching games in parallel, on much smaller graphs. This requires us to adapt the algorithm of the matching player, so that it solves a somewhat harder multi-commodity flow problem. We now provide more details on the cut-matching game and on our implementation of it.
%We start by describing an algorithm for a somewhat weaker variant of \Cref{thm:intro:main}, where $\phi_U,\phi_L\geq 1/n^{o(1)}$. We discuss its extension to the full range of parameters $\phi_U,\phi_L$ later.

%\thatchaphol{By this, we are able to implement}To do this, we need to implement the cut player very differently from the previous approaches that are based on solving a problem on $n$ vectors of dimension $n$, which seems to inherently require $\Omega(n^2)$ time. 
%

%(At a high level, this recursion strategy is similar to the recursive construction of short cycle decompositions~\cite{ChuGPSSW18,LiuSY19}.)

%\danupon{Perhaps this is too much? I guess we don't have to try to sell our techniques much since the result is strong anyway.}

%The cut-matching game is historically one of the two main methods for finding the balanced sparse cut. The first almost-linear time algorithm of Spielman and Teng is based on another method based on computing {\em PageRank} vectors.  

%In more details, the cut-matching game in an oversimplified form works as follows.\footnote{We present an oversimplified description of the cut-matching game for the sake of discussion. In reality, we will not get a per      
%       the witness graph $W$ will be  we will not get perfect matching ...\danupon{TO DO}}

\paragraph{Overview of the Cut-Matching Game.} 
We start with a high-level overview of a variant of the cut-matching game, due to Khandekar et~al.~\cite{KhandekarKOV2007cut}.
%\footnote{Readers who are not familiar with the cut-matching game may find it helpful to ignore the word ``near'' to follow the discussion. For simplicity, we sketch below the framework only in the case of constant-degree graphs.\danupon{Right?}}
We say that a graph $W$ is a $\psi$-\emph{expander} if it has no cut of sparsity less than $\psi$. We will informally say that $W$ is an \emph{expander} if it is a $\psi$-expander for some $\psi=1/n^{o(1)}$.
Given a graph $G=(V, E)$, the goal of the cut-matching game is to either find a balanced and sparse cut in $G$, or to embed an expander $W=(V, E')$  (called a {\em witness}) into $G$; note that $W$ and $G$ are defined over the same vertex set. The embedding of $W$ into $G$ needs to map every edge $e$ of $W$ to a path $P_e$ in $G$ connecting the endpoints of $e$. The \emph{congestion} of this embedding is the maximum number of  paths in $\{P_e\mid e\in E(W)\}$ that share a single edge of $G$. We require that the congestion of the resulting embedding is low. %(In fact the embedding that we use maps all but a few edges of $W$ to paths in $G$; we omit this detail to simplify the discussion.)
%For the sake of discussion, we oversimplify the notion of embedding. In reality our embedding maps all but a few edges in $W$ into paths in $G$.} 
%
Such an embedding serves as a certificate that there is no sparse balanced cut in $G$. This follows from the fact that, if $W$ is a $\psi$-expander, and it has a low-congestion embedding into another graph $G$, then $G$ itself is a $\psi'$-expander, where $\psi'$ depends on $\psi$ and on the congestion of the embedding.
The algorithm proceeds via an interaction between two algorithms, the cut player, and the matching player, and consists of $O(\log n)$ {\em rounds}. 
%
%This process is done by implementing two abstract data types (ADTs)\footnote{Recall that ADT only specifies the output to be expected from the point of view of the users. Unlike algorithms, it does not refer to a specific implementation.}
%%Recall that ADT is defined by its behavior (semantics) from the point of view of a user of the data in terms of possible values and operations on the data. 
%%Recall that    ADT only mentions what operations are to be performed but not how these operations will be implemented.}
%called {\em cut and matching players} which interact for {\em $O(\log n)$ rounds}. 
%

At the beginning of every round, we are given a graph $W$ whose vertex set is $V$, and its embedding into $G$; at the beginning of the first round, $W$ contains the set $V$ of vertices and no edges. In every round, the cut player either:
\begin{enumerate}[label=(C\arabic*),noitemsep]
        \item \label{item:intro:C1}``cuts $W$'', by finding a balanced sparse cut $S$ in $W$; or 
        \item\label{item:intro:C2} ``certifies $W$'' by announcing that $W$ is an expander.
\end{enumerate}
%(i) ``cuts $W$'' by finding a balanced sparse cut $S$ in $W$, or (ii) ``certifies $W$'' by announcing that $W$ is an expander.
%
%For the latter case, we find a desired embedding, so we can termininate the process and certify that $G$ admits no balanced sparse cut. For the earlier case, 
If $W$ is certified (\Cref{item:intro:C2}),  then we have constructed the desired embedding of an expander into $G$, so we can terminate the algorithm and certify that $G$ has no balanced sparse cut.
If a cut $S$  is found in $W$  (\Cref{item:intro:C1}), then we invoke the matching player, who either:
\begin{enumerate}[label=(M\arabic*),noitemsep]
        \item \label{item:intro:M1} ``matches $W$'', by adding to $W$ a large matching $M\subseteq S \times (V\setminus S)$ that can be embedded into $G$ with low congestion; or 
        \item \label{item:intro:M2} ``cuts $G$'', by finding a balanced sparse cut $T$ in $G$ (the cut $T$ is intuitively what prevents the matching player from embedding a large matching  $M\subseteq S \times (V\setminus S)$ into $G$).
\end{enumerate}
% (i) ``matches $W$'' by adding to $W$ a large matching $M$ between nodes in $S$ and $\ol{S}$ and embedding $M$ into $G$, or (ii) ``cuts $G$' by finding a sparse balanced cut $T$ in $G$;  
% 
%In the latter case, we terminate the process and output $T$ as a balanced sparse cut in $G$. In the earlier case, the game continues to the next round. It can be shown that after $\Theta(\log n)$ rounds $W$ becomes an expander, and so is $G$. 
%
If a sparse balanced cut $T$ is found in graph $G$ (\Cref{item:intro:M2}), then we return this cut and terminate the algorithm. Otherwise, the game continues to the next round. It was shown in \cite{KhandekarKOV2007cut} that the algorithm must terminate after $\Theta(\log n)$ rounds.
%becomes an expander, and so is $G$. 

%
%Roughly, given a graph $G=(V, E)$, the goal of this game is either find a sparse cut, or embed a graph $W=(V, E')$  (called {\em witness}) with high conductance into $G$; here, embedding means that each edge in $W$ corresponds to a path in $G$, and each edge in $G$ is contained in a few such paths. 
%%
%%In an oversimplified way, it works as follows. Given a graph $G$, there is a matching player who finds a set of edge-disjoint paths such that every node is an end-point of exactly one path. If such path does not exist, it can be argued that the matching player can find a sparse cut. If it exists
%%
%This is done by implementing two algorithms called {\em cut and matching players}. 
%%in the following process (oversimplification alert). 
%The process in an oversimplified form works as follows. 
%We start with the matching player, who finds a set $P$ of edge-disjoint paths such that every node is an end-point of exactly one path. If she does not find such paths, she must output a sparse cut in $G$. Otherwise, we add to $W$ a matching corresponding to these paths, i.e. we add an edge $(u,v)$ 

In the original cut-matching game by Khandekar, Rao and Vazirani~\cite{KhandekarRV09}, the matching player was implemented by an algorithm that computes a single-commodity maximum flow / minimum cut. 
The algorithm for the cut player was defined somewhat differently, in that in the case of \Cref{item:intro:C1}, the cut that it produced was not necessarily sparse, but it still had some useful properties, which guaranteed that the algorithm terminates after $O(\log^2n)$ iterations. 
In order to implement the cut player, the algorithm of \cite{KhandekarRV09} (implicitly) considers $n$ vectors of dimension $n$ each, that represent the probability distributions of random walks on the witness graph, starting from different vertices of $G$, and then uses a random projection of these vectors in order to construct the balanced cut. The algorithm exploits the properties of the witness graph in order to compute these projections efficiently, without explicitly constructing these vectors, which would be too time consuming.
%The cut player was implemented in near-linear time by manipulating $n$ vectors of dimension $O(\log n)$ each, that resulted from computing a \emph{random} projection of $n$ vectors of dimension $n$ each.  \jnote{is this actually correct? I don't remember any $\log n$-dimensional vectors}
%\yp{\url{https://people.eecs.berkeley.edu/~vazirani/pubs/partitioning.pdf} was projecting into a single direction, by taking each of the t matchings w.p. 1/2. So I don't see an intermedaite $\log{n}$ dimensional object either.}
%
Previous work (see, e.g., \cite{SaranurakW19,ChuzhoyK19}) implies that one can use algorithms for computing {\em maximal} flows instead of  maximum flows in order to implement the matching player in near-linear time deterministically, if the target parameters $1/\phi,\alpha \leq n^{o(1)}$.
%
%An almost-linear time implementation based on the cut-matching framework can be obtained (see e.g.~\cite{NanongkaiS17}) by replacing the exact single-commodity flow algorithms by an approximate (randomized) ones \cite{Sherman13,KelnerLOS14,Peng16}.
%%\footnote{Note that the result in \cite{NanongkaiS17} is not the first almost-linear time algorithm for the balanced sparse cut problem. }
%%
%While the randomness is required in the mentioned implementations, it can be observed that one can greedily compute a {\em maximal} single-commodity flow instead of a maximum one in order to implement the matching player. (This idea first appeared in \cite{ChuzhoyK19,SaranurakW19}; see \Cref{sec:match k one} for details.)
%\danupon{Should we also cite \cite{HenzingerRW17}? Also, \cite{ChuzhoyK19} was not mentioned in \Cref{sec:match k one} so this sentence migh tneed to be rewritten.}
%%\danupon{Right?}\thatchaphol{In \cite{SaranurakW19} as well which uses blocking flow or push-relabel.}; 
%
This still left open the question: {\em can the cut player be implemented via a deterministic and efficient algorithm?}

%\paragraph{A Deterministic Algorithm for the Cut Player.}

A natural strategy for derandomizing the algorithm of \cite{KhandekarRV09} for the cut player is to avoid the random projection of the vectors. In a previous work of a subset of the authors with Yingchareonthawornchai \cite{GaoLNPSY19}, this idea was used to develop a fast PageRank-based algorithm for the cut player, that can be viewed as a derandomization of the algorithm of Andersen, Chung and Lang for balanced sparse cut \cite{AndersenCL07}. Unfortunately, it appears that this technique cannot lead to an algorithm whose running time is below $\Theta(n^2)$: if we cannot use random projections, then we need to deal with $n$ vectors of dimension $n$ each when implementing the cut player, and so the running time of $\Omega(n^2)$ seems inevitable.
In this paper, we implement the cut player in a completely different way from the previously used approaches, by solving the balanced sparse cut problem recursively. 

We start by observing that, in order to implement the cut player via the approach of~\cite{KhandekarKOV2007cut}, it is sufficient to provide an algorithm for computing a balanced sparse cut on the witness graph $W$; in fact, it is not hard to see that it is sufficient to solve this problem approximately.
%\jnote{Question to Thatchaphol}\footnote{Originally Khandekar~et~al.~\cite{KhandekarKOV2007cut} required the cut player to solve balanced sparse cut exactly on $W$. As we observed in \cite{GaoLNPSY19}, solving the problem approximately on $W$ suffices for approximating the balanced sparse cut on $G$.} 
However, this leads us to a chicken-and-egg situation, where, in order to solve the \BC problem on the input graph $G$, we need to solve the \BC problem on the witness graph $W$. While graph $W$ is guaranteed to be quite sparse (with maximum vertex degree $O(\log n)$), it is not clear that solving the \BC problem on this graph is much easier. 
%It is thus tempting to run the same algorithm (that we run on $G$) on $W$. This, unfortunately, does not work since we do not know how to argue that the problem becomes easier on $W$, which is typically needed to guarantee that the recursion terminates. 

This motivates our recursive approach, in which, in order to solve the \BC problem on the witness graph $W$, we run a large number of cut-matching games in it simultaneously, each of which has a separate witness graph, containing significantly fewer vertices. It is then sufficient to solve the \BC problem on each of the resulting, much smaller, witness graphs. %The key to the analysis is to control the sizes of all graphs that we obtain at every level of the recursion, in order to bound the total running time. 
We prove the following theorem that provides a deterministic algorithm for the cut player via this recursive approach.

%---------------------------------------------

\begin{theorem}\label{thm: cut player}
	There is an universal constant $N_0$, and a deterministic algorithm, that we call \cutorcert, that, given an $n$-vertex graph $G=(V,E)$ with maximum vertex degree $O(\log n)$, and a parameter $r\geq 1$, such that $n^{1/r}\geq N_0$, returns one of the following:
	
	\begin{itemize}
		\item either a cut $(A,B)$ in $G$ with $|A|,|B|\geq n/4$ and $|E_G(A,B)|\leq n/100$; or
		\item a subset $S\subseteq V$ of at least $n/2$ vertices, such that $\Psi(G[S])\geq 1/\log^{O(r)}n$.
	\end{itemize}
	
	The running time of the algorithm is $O\left (n^{1+O(1/r)}\cdot (\log n)^{O(r^2)}\right )$.
\end{theorem}
%--------------------------------------------

%The above motivates an idea where we try to recurse on graphs with {\em fewer nodes}. 

We note that a somewhat similar recursive approach was used before,
e.g., in Madry's construction of $j$-trees~\cite{Madry10-jtree},
and in the recursive construction of short cycle decompositions~\cite{ChuGPSSW18,LiuSY19}.
In fact, \cite{GaoLNPSY19} use Madry's $j$-trees to solve \BC by running cut-matching games on graphs containing fewer and fewer nodes, obtaining an $(m^{1.5+o(1)})$-time algorithm. Unfortunately, improving this bound further does not seem viable via this approach, since the total number of edges contained in the graphs that belong to deeper recursive levels is very large. Specifically, assume that we are given an $n$-node graph $G$ with $m$ edges, together with a  parameter $k\geq 1$. We can then use the $j$-trees in order to reduce the problem of computing \BC on $G$ to the problem of computing \BC on $k$ graphs, each of which contains roughly $n/k$ nodes. Unfortunately,  each of these graphs may have $\Omega(m)$ edges. Therefore, the total number of edges in all resulting graphs may be as large as  $\Omega(mk)$, which is one of the major obstacles to obtaining faster algorithms for \BC using $j$-trees.

We now provide a more detailed description of the new recursive strategy that we use in order to prove \ref{thm: cut player}.

\paragraph{New Recursive Strategy.} 
%In this paper, we use a new strategy so that we can recurse on graphs with fewer nodes without significantly increasing the number of edges. In particular, we can recurse on $k$ graphs with $O(n/k)$ nodes such that the total number of edges is $O(m\polylog(n))$ edges (as opposed to $\Theta(mk)$ edges created by the approach above). By using large enough $k$ (e.g. $k=n^{o(1)}$), we can make the recursion stops before there are too many edges. 

%The key idea to play multiple, smaller, cut-matching games by replacing the matching player with our {\em multi-matching player}, as follows. 
%
% The difference is that $W$ is constructed as a union of expanders 
We partition the vertices of the input $n$-vertex graph $G$ into $k$  subsets $V_1, V_2, \ldots, V_k$ of roughly equal cardinality, for a large enough parameter $k$ (for example, $k=n^{o(1)}$). The algorithm consists of two stages.
In the first stage, we attempt to construct $k$ expander graphs $W_1,\ldots,W_k$, where $V(W_i)=V_i$ for all $1\leq i\leq k$, and embed them into the graph $G$ simultaneously. If we fail to do so, then we will compute a sparse balanced cut in $G$. In order to do so, we run $k$ cut-matching games in parallel. Specifically, we start with every graph $W_i$ containing the set $V_i$ of vertices and no edges, and then perform $O(\log n)$ iterations. In every iteration, we run the \cutorcert algorithm on each graph $W_1,\ldots,W_k$ in parallel. Assume that for all $1\leq i\leq k$, the algorithm returns a sparse balanced cut $(A_i,B_i)$ in $W_i$. We then use an algorithm of the matching player, that either computes, for each $1\leq i\leq k$, a matching $M_i$ between vertices of $A_i$ and $B_i$, and computes a low-congestion embedding of all  matchings $M_1,\ldots,M_k$ into graph $G$ simultaneously, or it returns a sparse balanced cut in $G$. In the former case, we augment each graph $W_i$ by adding the set $M_i$ of edges to it. In the latter case, we terminate the algorithm and return the sparse balanced cut in graph $G$ as the algorithm's output. If the algorithm never terminates with a sparse balanced cut, then we are guaranteed that, after $O(\log n)$ iterations, the graphs $W_1,\ldots,W_k$ are all expanders (more precisely, each of these graphs contains a large enough expander, but we ignore this technicality in this informal overview), and moreover, we obtain a low-congestion embedding of the disjoint union of these graphs into $G$. Note that, in order to execute this stage, we recursively apply algorithm \cutorcert to $k$ graphs, whose sizes are significantly smaller than the size of the graph $G$.

In the second stage, we attempt to construct a single expander graph $W^*$ on the set $\set{v_1,\ldots,v_k}$ of vertices, where for each $1\leq i\leq k$, we view vertex $v_i$ as representing the set $V_i$ of vertices of $G$. We also attempt to embed the graph $W^*$ into $G$, where every edge $e=(v_i,v_j)$ is embedded into $\Omega(n/k)$ paths connecting vertices of $V_i$ to vertices of $V_j$. In order to do so, we start with the graph $W^*$ containing the set $\set{v_1,\ldots,v_k}$ of vertices and no edges and then iterate. In every iteration, we run algorithm \cutorcert  on  the current graph $W^*$, obtaining a partition $(A,B)$ of its vertices. We then use an algorithm of the matching player in order to compute a matching $M$ between vertices of $A$ and vertices of $B$, and to embed every edge $(v_i,v_j)\in M$ of the matching into $\Omega(n/k)$ paths connecting vertices of $V_i$ to vertices of $V_j$ in graph $G$, with low congestion. If we do not succeed in computing the matching and the embedding, then the algorithm of the matching player returns a sparse balanced cut in graph $G$. We then terminate the algorithm and return this cut as the algorithm's output. Otherwise, we add the edges of $M$ to graph $W^*$ and continue to the next iteration. The algorithm terminates once graph $W^*$ is an expander, which must happen after $O(\log n)$ iterations. 

Lastly, we compose the expanders $W_1,\ldots,W_k$ and $W^*$ in order to obtain an expander graph $\hat W$ that embeds into $G$ with low congestion; the embedding is obtained by combining the embeddings of the graphs $W_1,\ldots,W_k$ and the embedding of graph $W^*$. This serves as a certificate that $G$ is an expander graph.

Note that the algorithm for the matching player that we need to use differs from the standard one in that it needs to compute $k$ different matchings between $k$ different pre-specified pairs of vertex subsets.
Specifically, the algorithm for the matching player is given $k$ pairs $(A_1,B_1),\ldots,(A_k,B_k)$ of subsets of vertices of $G$ of equal cardinality. Ideally, we would like the algorithm to either (i) compute, for all $1\leq i\leq k$, a perfect matching $M_i$ between vertices of $A_i$ and vertices of $B_i$, and embed all edges of $M_1\cup\cdots\cup M_k$ into $G$ simultaneously with low congestion; or (ii) compute a sparse balanced cut in $G$. In fact our algorithm for the matching player achieves a somewhat weaker objective: namely, the matchings $M_i$ are not necessarily perfect matchings, but they are sufficiently large. In order to overcome this difficulty, we introduce ``fake'' edges that augment each matching $M_i$ to a perfect matching. As a result, if the algorithm fails to compute a sparse balanced cut in $G$, then we are only guaranteed that $G\cup F$ is an expander, where $F$ is (a relatively small) set of fake edges. We then use a known ``expander trimming'' algorithm of \cite{SaranurakW19} in order to find a large subset $S\subseteq V(G)$ of vertices, such that $G[S]$ is an expander, and the cut $S$ is sufficiently sparse. We note that the notion of fake edges was used before in the context of the cut-matching game, e.g. in \cite{KhandekarRV09}.

The algorithm of the matching player builds on the idea of Chuzhoy and Khanna \cite{ChuzhoyK19} of computing maximal sets of short edge-disjoint paths,  which can be implemented efficiently via Even-Shiloach's algorithm for decremental single-source shortest paths~\cite{EvenS}.
Unfortunately, this approach requires slightly slower running time of $O\left(m^{1+O(1/r)}\cdot (\log m)^{O(r^2)}/\phi^2\right )$, introducing a quadratic dependence on $1/\phi$, where $\phi$ is the conductance parameter. The expander trimming algorithm of \cite{SaranurakW19} that is exploited by the cut player also unfortunately introduces a linear dependence on $1/\phi$. As a result, we obtain an algorithm for the \BCut problem that is sufficiently fast in the high-conductance regime, that is, where $\phi=1/\poly\log n$, but is too slow for the setting where the parameter $\phi$ is low. Luckily, the high-conductance regime is sufficient for many of our applications, and in particular it allows us to obtain efficient approximation algorithms for maximum flow. This algorithm can then in turn be used in order to implement the matching player, even in the low-conductance regime, removing the dependence of the algorithm's running time on $\phi$. Additional difficulty for the low-conductance regime is that we can no longer afford to use the expander trimming algorithm of \cite{SaranurakW19}. Instead, we provide an efficient deterministic bi-criteria approximation algorithm for the most-balanced sparest cut problem, and use this algorithm in order to solve the \BCut problem in the low-conductance regime. This part closely follows ideas of \cite{NanongkaiS17,Wulff-Nilsen17,ChuzhoyK19,ChangS19}.

\subsection{Paper Organization}
%\paragraph{Paper Outline.} 

%Our algorithm essentially follows the framework outlined above. To make all the pieces work, we need quite a few definitions and  reductions, which are presented in the following order.

%\jnote{the discussion in the next paragraph seems too detailed/technical for "organization" subsection, that normally only lists what appears where. Is it important to say here how degrees are defined, or that we work with unweighted graphs? This seem to belong to prelims}
We start with preliminaries in \Cref{sec:prelim}. %Also note that we work on unweighted graphs throughout the paper, except for applications on weighted graphs in \Cref{sec:app}. Our notion of node's degree is slightly unconventional in that a self-loop contributes a value of one to its end-node's degree (instead of two). This helps keeping our arguments simpler.\danupon{This is new. Please check.}
In \Cref{sec:match}, we define the problem to be solved by the new matching player, and provide an algorithm for solving it. %, that is based on an extension of Chuzhoy and Khanna's algorithm \cite{ChuzhoyK19}. 
We also provide a faster algorithm the case where $k=1$ (that is, the problem of the standard matching player), which we exploit later. % This algorithm plays a role in our final approximation factor. For convenience, we call the ``multi-matching player'' simply by the ``matching player''.
% (they can be distinguished by the parameter $k$). 
%
%In \Cref{sec:cut}, we discuss variants of the balanced sparse cut problem to be solved by the cut player. As noted earlier, there are many ways to define the balanced sparse cut problem (they differ in their outputs when a sparse cut is not found). We consider a few of these variants which we call \Cut, \BCut, \vBCut, and \CA.
%\yp{why mention these before we don't define them? Wouldn't it suffice to just say `we consider variants, and show that they interreduce'?}
%This is because we need to invoke the cut player in a few places, and choosing the right variants in these places makes the overall algorithm simpler. 
%
%In \Cref{sec:cut} and \Cref{sub:CA via vBCut} we define several variants of the \BCut\ problem and provide reductions between them. 
%
%
%Using known techniques, we show the following reductions between the variants: $\vBCut\rightarrow\BCut\rightarrow\Cut\rightarrow\CA$. 
%
%In \Cref{sub:CA via vBCut}, we ``close the reduction loop'' by showing the reduction $\CA\Rightarrow\vBCut$, i.e.~we solve $\CA$ by solving \vBCut on  graphs with fewer nodes. (We use ``$\Rightarrow$'' to emphasize that the reduction produces graphs with fewer nodes, as opposed to other reductions.) This is where we play many small cut-matching games in parallel as outlined above, where the cut player is implemented by solving \vBCut, and the output is a solution for \CA. 
%
We prove our main technical result, \ref{thm: cut player}, in \ref{sec: cut player}, obtaining the algorithm for the cut player.
In \Cref{sec:BCut_high_cond}, we obtain a proof of \Cref{thm:intro:main} with slightly weaker guarantees, where the running time depends linearly on $1/\phi^2$. 
 %put all the reductions in \Cref{sec:cut,sub:CA via vBCut} together. These reductions allow us to view our algorithm as a recursion where we solve any variant of the cut problem by solving the same problem on smaller inputs. We focus on $\BCut$ which can be used conveniently in applications. By following the chain of reductions $\BCut\rightarrow\Cut\rightarrow\CA\Rightarrow\vBCut \rightarrow\BCut$, we get a recursion for $\BCut$. 
%We analyze this algorithm and obtain the proof of \Cref{thm:intro:main} for a limited range of parameters $\alpha_U,\alpha_L$ in  a weaker result than our main result (\Cref{thm:intro:main}). 
%
%       \begin{cor}
%               \label{cor:intro:BCut linear}There is an algorithm of $\BCut$ that, for
%               any $(m,\phi_{U},\phi_{L})$ where 
%               \[
%               \phi_{U}\in[0,1]\text{ and }\phi_{L}\le\phi_{U}/2^{O(\log^{1/3}n\cdot(\log\log n)^{2/3})},
%               \]
%               has running time $T_{\BCut}(m,\phi_{U},\phi_{L})\le m\phi_{U}^{-2}2^{O(\log^{2/3}n\cdot(\log\log n)^{1/3})}.$
%       \end{cor}
%
%
In \Cref{sec:app}, we use our result from \Cref{sec:BCut_high_cond} to obtain algorithms for most of our applications. 
Finally, in \Cref{sec:BCut_nophi} we complete the proof of \Cref{thm:intro:main}, and provide some additional applications of our results for low-conductance regime, including the proof of \ref{thm: sparsest and low cond}. We conclude with open problems in \Cref{sec:conclusion}.

\section{Preliminaries}\label{sec:prelim}
All graphs considered in this paper are unweighted and undirected, and they may have parallel edges but no self-loops.
Given a graph $G=(V,E)$, for every vertex $v\in V$, we denote by $\deg_G(v)$ the degree of $v$ in $G$. For any set $S\subseteq V$ of vertices of $G$,
the \emph{volume} of $S$ is the sum of degrees of all nodes in $S$: $\vol_{G}(S)=\sum_{v\in S}\deg_{G}(v)$. We denote the total volume of the graph $G$ by  
$\volG=\vol_{G}(V)$. Notice that $\volG = 2|E|$. % (since we allow self-loops that contribute $1$ to vertex degrees).

%In this paper, all graphs are undirected unweighted multi-graphs that may have self loops. 
%Let $G = (V,E)$ be a graph. \jnote{this is standard graph-theoretic notation. Do we need to say this?} Let $V(G)=V$ and $E(G)=E$ be a set of nodes and edges of $G$ respectively.
%We emphasize that each self loop at a node $u$ contributes 1 to its degree, denoted by $\deg_G u$.
%
We use standard graph theoretic notation: for two subsets $A,B\subseteq V$ of vertices of $G$, we denote by $E_{G}(A,B)$ the set of all edges with one endpoint in $A$ and another in $B$. 
Assume now that we are given a subset $S$ of vertices of $G$. We denote by $G[S]$  the subgraph  of $G$ induced by $S$. We also denote $\overline S=V\setminus S$, and $G-S=G[\overline S]$.
%Given a set $S$ of vertices in $G$, we let $G\{S\}$ be the subgraph of $G$ induced by $S$ with all vertex degrees preserved: graph $G\{S\}$ is obtained from $G[S]$ by adding self-loops to the vertices of $S$ so that $\deg_{G\{S\}}u = \deg_{G}u$ for all $u \in S$. 
%We also denote by $G/S$ be the (multi)-graph obtained from $G$, after we contract the set $S$ of vertices into a single node, preserving all new self-loops. 

A \emph{cut} in $G$ is a partition $(A,B)$ of its vertices, where $A,B\neq \emptyset$. We sometimes also call a subset $S$ of vertices of $G$ with $S\neq \emptyset,V$ a cut, referring to the corresponding cut $(S,\overline S)$. The {\em size} of a cut $S$ is  $\delta_{G}(S)=|E_{G}(S,\overline{S})|$.  

\subsection{Conductance and Sparsity}
The two central cut-related notions that we use in this paper are conductance and sparsity.  Intuitively, both these notions measure  how much a given cut ``expands'', though they do it somewhat differently.
Formally, the \emph{conductance of a cut} $S$ is: $\Phi_{G}(S)=\frac{\delta_{G}(S)}{\min\{\vol_{G}(S),\vol_{G}(\overline{S})\}}$. Intuitively, if $\vol_G(S)\leq \volG/2$, then $\Phi_G(S)$ is the fraction of the edges incident to vertices of $S$ that have their other endpoint outside $S$. The \emph{conductance of a graph $G$}, that we denote by $\Phi(G)$, is the smallest conductance of any cut in $G$: $\Phi(G) = \min_{\emptyset\neq S \subsetneq V}\Phi_G(S)$.
The \emph{sparsity of a cut $S$} is: $\Psi_{G}(S)=\frac{\delta_{G}(S)}{\min\{|S|,|\overline{S}|\}}$,
and the \emph{expansion of a graph $G$} is
$\Psi(G) = \min_{\emptyset\neq S \subsetneq V}\Psi_G(S)$. 
%
%
%
%We sometimes informally call cuts of high conductance or sparsity as \emph{expanding}. and cuts with low conductance/sparsity are \emph{sparse}. \jnote{I think we should omit this. We should call them "sparse" or "low-conductance" cuts. Also I think expanders should be defined formally.}\jl{Is there a different term than sparse? Since ``sparse'' suggests sparsity in particular.} Graphs with no sparse cuts are \emph{expanders}. 
%
 The following claim establishes a basic connection between a conductance and a sparsity of a cut.

\begin{claim}%[Conductance vs.\ sparsity]
\label{prop:sparsity vs conductance}Let $G=(V,E)$ be a connected
graph with maximum vertex degree $\Delta$, and let $S\subsetneq V$ be any cut in $G$. Then: 

$$\frac{\Psi_{G}(S)}{\Delta} \le\Phi_{G}(S)\le\Psi_{G}(S).$$
\end{claim}

The proof immediately follows from the fact that, for every set $X$ of vertices of $G$, $|X|\leq \vol_G(X)\leq \Delta\cdot |X|$.

\subsection{Expanders}
We use the following definition of expanders.

\begin{defn}
	We say that a graph $G$ is a $\psi$-expander iff $\Psi(G)\geq \psi$. 
	\end{defn}

We will sometimes informally say that graph $G$ is an \emph{expander} if $\Psi(G)\geq 1/n^{o(1)}$. 
We use the following simple observation multiple times.

\begin{observation}\label{obs: exp plus matching is exp}
Let $G=(V,E)$ be an $n$-vertex graph that is a $\psi$-expander, and let $G'$ be another graph that is obtained from $G$ by adding to it a new set $V'$ of at most $n$ vertices, and a matching $M$, connecting every vertex of $V'$ to a distinct vertex of $G$. Then $G'$ is a $\psi/2$-expander.
\end{observation}

We also use the following theorem that provides a fast algorithm for an explicit construction of an expander, that is based on the results of Margulis \cite{Margulis} and Gabber and Galil \cite{GabberG81}.
\begin{theorem}
	%[Fast explicit expander construction]
	\label{thm:explicit expander}
	There is a constant $\alpha_0 > 0$ and a deterministic algorithm, that we call \constructexpander, that, given an integer $n>1$, in time $O(n)$ constructs a graph $H_n$ with $|V(H_n)|=n$, such that $H_n$ is an $\alpha_0$-expander, and every vertex in $H_n$ has degree at most $9$.
\end{theorem}

\begin{proof}
	We assume that $n\ge10$, as otherwise the graph $H_{n}$ with the required properties can be constructed in
	constant time. We use the expander construction of Margulis \cite{Margulis} and Gabber and Galil \cite{GabberG81}. For an integer $k>1$, let $H'_{k^{2}}$ be a graph whose vertex set is
	set $\mathbb{Z}_{k}\times\mathbb{Z}_{k}$ where $\mathbb{Z}_{k}=\mathbb{Z}/k\mathbb{Z}$.
	Each vertex $(x,y)\in\mathbb{Z}_{k}\times\mathbb{Z}_{k}$ has exactly eight adjacent edges, connecting it to the vertices $(x\pm2y,y),(x\pm(2y+1),y),(x,y\pm2x)$, and $(x,y\pm(2x+1))$.
	Gabber and Galil \cite{GabberG81} showed that $\Psi(H'_{k^{2}})=\Omega(1)$. 
	
	Given a parameter $n\geq 10$, we let $k$ be the unique integer with $(k-1)^{2}<n\le k^{2}$, and let $n'=n-(k-1)^2$. Clearly, $n'\leq k^2-(k-1)^2\leq 2k<(k-1)^2$. In order to obtain the graph $H_n$, we start with the graph $H_{(k-1)^2}$, whose vertex set we denote by $V'$, and then add a set $V''$ of $n'$ isolated vertices to this graph. Lastly, we add an arbitrary matching, connecting every vertex of $V''$ to a distinct vertex of $V'$, obtaining the final graph $H_n$. It is immediate to verify that $|V(H_n)|=n$, and that every vertex in $H$ has degree at most $9$. Moreover, from \ref{obs: exp plus matching is exp}, graph $H_n$ is an $\Omega(1)$-expander.
\end{proof}

\subsection{The Cut-Matching Game}

\label{subsec:KKOV}

The \emph{cut-matching	game} was introduced by Khandekar, Rao, and Vazirani \cite{KhandekarRV09} as part of their fast randomized algorithm for the Sparsest Cut and Balanced Cut problems. We use a variation of this game, due to  Khandekar et al. \cite{KhandekarKOV2007cut}, that we slightly modify to fit our framework. 
The game involves two players - the \emph{cut player}, who wants to construct an expander fast, and the \emph{matching player}, who wants to delay the construction of the expander. Initially, the game starts with a graph $H$ that contains an even number $n$ of vertices an no edges. The game is played in iterations, where in every iteration $i$, some set $M_i$ of edges is added to the current graph $H$. The $i$th iteration is played as follows. The cut player computes a partition $(A_i,B_i)$ of $V(H)$ with $|A_i|,|B_i|\geq n/4$ and $|E_H(A_i,B_i)|\leq n/100$. Assume without loss of generality that $|A_i|\leq |B_i|$. The matching player computes any partition $(A_i',B_i')$ of $V(H)$ with $|A'_i|=|B'_i|$, such that $A_i\subseteq A_i'$, and then computes an arbitrary perfect matching $M_i$ between $A_i'$ and $B_i'$. The edges of $M_i$ are then added to the graph $H$. The algorithm terminates when graph $H$ no longer contains a partition $(A,B)$ of $V(H)$ with $|A|,|B|\geq n/4$ and $|E_H(A,B)|\leq n/100$. Intuitively, once the algorithm terminates, it is easy to see that $H$ contains a large subgraph that is an expander. Alternatively, it is easy to turn $H$ into an expander by adding one last set of $O(n)$ edges to it. We note that the graph $H$ is a multi-graph, that is, it may contain parallel edges. The following theorem follows from the result of \cite{KhandekarKOV2007cut}
(since we slightly modify their setting, we include the proof in Appendix for completeness).

\begin{thm}
	\label{thm:KKOV-new} There is a constant $\cCMG$, such that the algorithm described above terminates after at most $\cCMG\log n$ iterations.
\end{thm}

We will use this cut-matching game together with algorithm \cutorcert from \Cref{thm: cut player}, that will be used in order to implement the cut player. The matching player will be implemented by a different algorithm, that we discuss in the following section. Note that, as long as the algorithm from \Cref{thm: cut player} produces a cut $(A,B)$ of $H$ with the required properties, we can use the output of this algorithm as the response of the cut player. \Cref{thm:KKOV-new} guarantees that, after at most $O(\log n)$ iterations of the game, the algorithm from  \Cref{thm: cut player} will return a subset $S\subseteq V(H)$ of at least $n/2$ vertices, such that graph $H[S]$ is an expander. Once this happens, we will terminate the cut-matching game.

\subsection{Expander Pruning}

We use the following theorem from \cite{SaranurakW19}.

\begin{theorem}[Restatement of Theorem 1.3 from~\cite{SaranurakW19}]\label{thm: expander pruning}
	There is a deterministic algorithm, that, given a graph $G=(V,E)$ of conductance $\Phi(G)=\phi$, for some $0<\phi\leq 1$, and a collection $E'\subseteq E$ of $ {k\leq \phi |E|/10}$ edges of $G$, computes a subgraph $G'\subseteq G\setminus E'$, that has conductance $\Phi(G')\geq \phi/6$. Moreover, if we denote $A=V(G')$ and $B=V(G)\setminus A$, then $|E_G(A,B)| \leq 4k$, and $\vol_G(B)\leq 8k/\phi$.
	The total running time of the algorithm is $ {\tilde O(|E|/\phi)}$.
\end{theorem}

We note that~\cite{SaranurakW19} provide a significantly stronger result, where the edges of $E'$ arrive in an online fashion and the graph $G'$ is maintained after each edge arrival. Additionally, the running time of the algorithm is $\tilde O(k/\phi^2)$ if the algorithm is given an access to the adjacency list of $G$. However, the weaker statement above is cleaner and it is sufficient for our purposes.

\subsection{Embeddings of Graphs and Expansion}
Next, we define embeddings of graphs, that will be later used to certify graph expansion.

\begin{defn}
	Let $G$, $H$ be two graphs with $V(G)= V(H)$. An \emph{embedding} of $H$ into $G$ is a collection $\pset=\set{P(e)\mid e\in E(H)}$ of paths in $G$, such that for each edge $e\in E(H)$, path $P(e)$ connects the endpoints of $e$ in $G$. We say that the embedding causes \emph{congestion} $\cong$ iff every edge $e'\in E(G)$ participates in at most $\cong$ paths in $\pset$. %We say that $\pset$ is an $(\ell,\cong)$-embedding iff it causes congestion at most $\cong$, and every path in $\pset$ has length at most $\ell$.
	\end{defn}

Next we show that, if $G$ and $H$ are any two graphs with $|V(G)|=|V(H)|$, and $H$ is a $\psi$-expander that embeds into $G$ with a small congestion, then $G$ is also an expander, for an appropriately chosen expansion parameter.
We note that this observation was used in a number of previous algorithms in order to certify that a given graph is an expander; see, e.g. \cite{LeightonR99,AroraRV09,KhandekarRV09,KhandekarKOV2007cut,AroraHK10,Sherman09}. %However, until our work, given a graph $G$, we do not know how to embed an expander $H$ of the same size into $G$ deterministically in almost-linear time. 

%Assume now that we are given two graphs $G$ and $H$ with $V(H)\subseteq V(G)$, and an embedding $\pset=\set{P(e)\mid e\in E(H)}$ of $H$ into $G$. Let $V(\pset)$ denote the set of all vertices of $G$ that participate in the paths in $\pset$, and let $G^{\pset}$ be the sub-graph of $G$ induced by the vertices of $V(\pset)$. The following key lemma shows that, if $H$ is an expander, and $\pset$ is an $(\ell,\cong)$-embedding of $H$ into $G$, with parameters $\ell,\cong$ that are not too large, then $G^{\pset}$ is also an expander.

\begin{lem}\label{lem: embedding expander gives expander}
		Let $G$, $H$ be two graphs with $V(G)= V(H)$, such that $H$ is a $\psi$-expander, for some $0<\psi<1$. Assume that there exists an embedding $\pset=\set{P(e)\mid e\in E(H)}$ of $H$ into $G$ with congestion at most $\cong$, for some $\cong\geq 1$. Then $G$ is a $\psi'$-expander, for  $\psi'=\psi/\cong$.
	\end{lem}

\begin{proof}
%For convenience, we denote the graph $G^{\pset}$ by $G'$, and we refer to the vertices of $V(H)$ as \emph{terminals}, denoting $V(H)=T$. The remaining vertices of $G'$ are called \emph{non-terminal vertices}; note that each non-terminal vertex of $G'$ must be an inner vertex of at least one and at most $\cong$ paths in $\pset$.
	Consider any partition $(A,B)$ of $V(G)$, and assume that $|A|\leq |B|$. Consider the corresponding cut $(A,B)$ in $H$, and let $E'=E_H(A,B)$. Since $H$ is a $\psi$-expander, $|E'|\geq \psi |A|$. Note that for every edge $e\in E'$, its corresponding path $P(e)$ in $G$ must contain an edge of $E_G(A,B)$. Since the paths in $\pset$ cause congestion at most $\cong$, we get that $|E_G(A,B)|\geq \frac{|E_H(A,B)|}{\cong}\geq \frac{\psi |A|}{\eta}$.
\end{proof}

\subsection{Embeddings with Fake Edges and Expansion}

\iffalse
\begin{theorem}[Restatement of Theorem 1.3 in~\cite{SaranurakW19}]\label{thm: expander pruning}
	There is a deterministic algorithm, that, given an access to the adjacency list of a graph $G=(V,E)$ with $|E|=m$, a parameter $0<\phi\leq 1$, and a sequence $\sigma=(e_1,e_2,\ldots,e_k)$ of $ {k\leq \phi m/10}$ online edge deletions, maintains a vertex set $P\subseteq V$ with the following properties. Let $G_i$ be the graph $G$ after the edges $e_1,\ldots,e_i$ have been deleted from it; let $P_0=\emptyset$ be the set $P$ at the beginning of the algorithm, and for all $0<i\leq k$, let $P_i$ be the set $P$ after the deletion of $e_1,\ldots,e_i$. Then, for all $1\leq i\leq k$:
	
	\begin{itemize}
		\item $P_{i-1}\subseteq P_i$;
		\item $ {\vol_G(P_i)\leq 8i/\phi}$;
		\item $|E(P_i,V\setminus P_i)|\leq 4i$; and
		\item if  $G$ is a $\phi$-expander, then $G_i[V\setminus P_i]$ is a $\phi/6$-expander.
	\end{itemize}
	
	The total running time of the algorithm is $ {O(k\log m/\phi^2)}$.
\end{theorem}
\fi

In general, when using the cut-matching game, one can usually either embed an expander into a given graph $G$, or compute a sparse cut $S$ in $G$. Unfortunately, it is possible that $|S|$ is quite small in the latter case. Since each execution of the cut-matching game algorithm takes at least $\Omega(|E(G)|)$ time, we cannot afford to iteratively remove such small sparse cuts from $G$, if our goal is to either embed a large expander or to compute a balanced sparse cut in $G$ in almost-linear time. In order to overcome this difficulty, we use \emph{fake edges} (that were also used in \cite{KhandekarRV09}), together with the expander pruning algorithm from \ref{thm: expander pruning}.

Specifically, suppose we are given any graph $G=(V,E)$, and let $F$ be a collection of edges whose endpoints lie in $V$, but the edges of $F$ do not necessarily belong to $G$. We denote by $G+F$ the graph obtained by adding the edges of $F$ to $G$. If an edge $e$ lies both in $E$ and $F$, then we add a new parallel copy of this edge. We note that $F$ is allowed to be a multi-set, in which case multiple parallel copies of an edge may be added to $G$. %We extend our definition of graph embeddings to include fake edges.

%\begin{defn}
%	Let $G$, $H$ be two graphs with $V(H)\subseteq V(G)$, and let $F$ be a set of fake edges for $G$. An embedding of $H$ into $G$ with a set $F$ of fake edges is simply an embedding $\pset=\set{P_e\mid e\in E(H)}$ of $H$ into $G+F$. As before, we say that $\pset$ is an $(\ell,\cong)$-embedding iff it causes congestion at most $\cong$, and every path in $\pset$ has length at most $\ell$.
%\end{defn}

We show that, if $H$ is an expander graph, and we embed it into a graph $G+F$ with a small collection $F$ of fake edges, then we can efficiently compute a large subgraph of $G$ that is an expander.

\begin{lem}\label{lem: embedding expander w fake edges gives expander}
	Let $G$ be an $n$-vertex graph, and let $H$ be another graph with $V(H)= V(G)$, with maximum vertex degree $\Delta_H$, such that $H$ is a $\psi$-expander, for some $0<\psi<1$. Let $F$ be any set of $k$ fake edges for $G$, and let $\Delta_G$ be the maximum vertex degree in $G+F$. Assume that there exists an embedding $\pset=\set{P(e)\mid e\in E(H)}$ of $H$ into $G+F$, that causes congestion at most $\cong$, for some $\cong\geq 1$. Assume further that $k\leq \frac{\psi n}{32\Delta_G\cong}$. Then there is a subgraph $G'\subseteq G$ that is a $\psi'$-expander, for $\psi'\geq \frac{\psi}{6\Delta_G\cdot\cong}$, such that, if we denote by $A=V(G')$ and $B=V(G)\setminus A$, then  $|A|\geq n-\frac{4k\cong}{\psi}$ and $|E_G(A,B)|\leq 4k$. Moreover, there is a deterministic algorithm, that we call \extractexpander, that, given $G,H,\pset$ and $F$, computes such a graph $G'$ in time $\tilde O(|E(G)|\Delta_G\cdot\cong/\psi)$.
\end{lem}

\begin{proof}
	For convenience, we denote $\hat G=G+F$. From \ref{lem: embedding expander gives expander}, graph $\hat G$ is a $\hat \psi$-expander, for  $\hat \psi=\psi/\cong$. %Note that for every vertex $v$ in $\hat G'$, the degree of $v$ in $\hat G'$ is bounded by $\Delta+\cong$, and $|V(\hat G')|\geq |V(H)|$. 
	Moreover, from \Cref{prop:sparsity vs conductance}:
	
	 \[\Phi(\hat G)\geq \frac{\Psi(\hat G) }{\Delta_G}\geq \frac{\psi}{\Delta_G\cdot \cong}.\]

	In the remainder of the proof, we apply \ref{thm: expander pruning} to graph $\hat G$ and the set $F$ of edges. 
	Recall that the set $F$ of fake edges has cardinality $k\leq  \frac{\psi n}{32\Delta_G\cdot \cong}\leq \frac{n\cdot \Phi(\hat G)}{10}\leq \frac{|E(\hat G)|\cdot \Phi(\hat G)}{10}$. 
	Therefore, we can use \ref{thm: expander pruning} to obtain a subgraph $G'\subseteq (\hat G\setminus F)\subseteq G$, that has conductance at least $\frac{\Phi(\hat G)}{6}\geq \frac{\psi}{6\Delta_G\cdot \cong}$. Denoting $A=V(G')$ and $B=V(\hat G)\setminus V(G')=V(G)\setminus V(G')$, \ref{thm: expander pruning} guarantees that $|E_G(A,B)|\leq |E_{\hat G}(A,B)|\leq 4k$. 
	From \Cref{prop:sparsity vs conductance}, $\Psi(G')\geq \Phi(G')$, and so graph $G'$ is a $\psi'$-expander, for $\psi'=\frac{\psi}{6\Delta_G\cong}$. The running time of the algorithm is $\tilde O(|E(\hat  G)|/\Phi(\hat G))=\tilde O(|E(G)|\Delta_G\cong/\psi)$.
	It remains to show that $|A|$ is sufficiently large.
	
%	We denote $A=V(G')$ and $B=V(\hat G)\setminus A$, obtaining a partition $(A,B)$ of $V(\hat G)$. 
Recall that \ref{thm: expander pruning} guarantees that $|E_{\hat G}(A,B)|\leq 4k$, while $\vol_{\hat G}(B)\leq \frac{8k}{\Phi(\hat G)}\leq \frac{8k\Delta_G\cong}{\psi}$.
	In particular, $|B|\leq \frac{8k\Delta_G\cong}{\psi}\leq \frac{n}{2}$, since $k\leq  \frac{\psi n}{32\Delta_G\cong}$.
	Since graph $\hat G$ is a $\hat \psi$-expander, and $|E_{\hat G}(A,B)|\leq 4k$, we conclude that $|B|\leq \frac{|E_{\hat G}(A,B)|}{\hat \psi}\leq  \frac{4k}{\hat \psi}\leq  \frac{4k\cong}{\psi}$, and so $|A|\geq n-\frac{4k\cong}{\psi}$.
\end{proof}

%------------------------------------------

%!TEX root = main_det_cut.tex

\section{Route or Cut: Algorithm for the Matching Player}
\label{sec:match}

The goal of this section is to design an algorithm that will be used by the matching player.
We use the following definition for routing pairs of vertex subsets.

\begin{defn}

Assume that we are given a graph $G=(V,E)$, and disjoint subsets $A_1,B_1,A_2,B_2,\ldots,A_k,B_k$ of its vertices, that we refer to as \emph{terminals}. Assume further that for each $1\leq i\leq k$, $|A_i|\leq |B_i|$; we denote $|A_i|=n_i$. A \emph{partial routing} of the sets $A_1,B_1,\ldots, A_k,B_k$ consists of:

\begin{itemize}
	\item A set $M=\bigcup_{i=1}^kM_i\subseteq V\times V$ of pairs of vertices, where for each $1\leq i\leq k$, $M_i$ is a matching between vertices of $A_i$ and vertices of $B_i$ (we emphasize that the pairs $(u,v)\in M_i$ do not necessarily correspond to edges of $G$); and

\item For every pair $(u,v)\in M$ of vertices, a path $P(u,v)$ connecting $u$ to $v$ in $G$.
\end{itemize}

We denote the resulting routing by $\pset=\set{P(u,v)\mid (u,v)\in M}$ (note that the matching $M$ is implicitly defined by $\pset$). We say that the routing $\pset$ causes congestion $\cong$, if every edge in $G$ belongs to at most $\cong$ paths in $\pset$. The \emph{value} of the routing is $\sum_{i=1}^k|M_i|$.

\end{defn}
	
	We are now ready to state the main result of this section, which is an algorithm that will be used by the Matching Player. We note that the theorem is a generalization of a similar result that was proved in \cite{ChuzhoyK19}, for the special case where $k=1$.

	\begin{thm}\label{thm: matching player}
		There is a deterministic algorithm, that, given an $n$-vertex graph $G=(V,E)$ with maximum vertex degree $\Delta$, disjoint subsets $A_1,B_1,\ldots,A_k,B_k$ of its vertices, where for all $1\leq i\leq k$, $|A_i|\leq |B_i|$ and $|A_i|=n_i$, and integers $z\geq 0$, $\ell\geq 32\Delta\log n$, computes one of the following:
		
		\begin{itemize}
			\item either a partial routing of the sets $A_1,B_1,\ldots,A_k,B_k$, of value at least $\sum_in_i-z$, that causes congestion at most $\ell^2$; or
		
			\item a cut $(X,Y)$ in $G$, with $|X|,|Y|\geq z/2$, and $\Psi_G(X,Y)\leq 72\Delta\log n/\ell$.
		\end{itemize}
	
The running time of the algorithm is $\tilde O(\ell^3k|E(G)|+\ell^2kn)$.
\end{thm}

(We note that the parameter $\ell$ in the above theorem bounds the lengths of the paths in $\pset$, that is, we will ensure that every path in $\pset$ contains at most $\ell$ edges; however, since our algorithm does not rely on this fact, this is immaterial).

\begin{proof}
	The proof of the theorem immediately follows from the following lemma.

	\begin{lem}\label{lem: matching player 1 it}
		There is a deterministic algorithm, that, given an $n$-vertex graph $G=(V,E)$ with maximum vertex degree $\Delta$, disjoint subsets $A'_1,B'_1,\ldots,A'_k,B'_k$ of its vertices, where for all $1\leq i\leq k$, $|A'_i|\leq |B'_i|$, and $|A'_i|=n'_i$, and an integer $\ell\geq 32\Delta\log n$, computes one of the following:
		
		\begin{itemize}
			\item either a partial routing of the sets $A'_1,B'_1,\ldots,A'_k,B'_k$ in $G$, of value at least $\left(\sum_{i=1}^kn'_i\right )\cdot\frac{8\log n}{\ell^2}$ and congestion $1$; or
			
			\item a cut $(X,Y)$ in $G$, with $|X|,|Y|\geq \left(\sum_{i=1}^kn'_i\right )/2$, and $\Psi_G(X,Y)\leq 72\Delta\log n/\ell$.
		\end{itemize}
		
		The running time of the algorithm is $\tilde O(k\ell |E(G)|+kn)$.
	\end{lem}
	
	Before we prove the lemma, we complete the proof of \ref{thm: matching player} using it. Throughout the algorithm, we maintain the matchings $M_1,\ldots,M_k$, where $M_i$ is a matching between vertices of $A_i$ and vertices of $B_i$, and a routing $\pset=\set{P(u,v)\mid (u,v)\in \bigcup_iM_i}$. Initially, we set $M_i=\emptyset$ for all $i$, and $\pset=\emptyset$. We then iterate. In every iteration, for each $1\leq i\leq k$, we let $A_i'\subseteq A_i$ and $B'_i\subseteq B_i$ be the subsets of vertices that do not participate in the matching $M_i$, and we denote $n'_i=|A'_i|$; since $|A_i|\leq |B_i|$, we are guaranteed that $|A'_i|\leq |B'_i|$. We also denote $N'=\sum_in'_i$. If $N'\leq z$, then we terminate the algorithm, and return the current matchings $M_1,\ldots,M_k$, together with their routing $\pset$. Otherwise, we apply \ref{lem: matching player 1 it} to graph $G$ and vertex sets $A'_1,B'_1,\ldots,A'_k,B'_k$. If the outcome is a cut  $(X,Y)$ in $G$, with $|X|,|Y|\geq N'/2$, and $\Psi_G(X,Y)\leq 72\Delta\log n/\ell$, then we terminate the algorithm, and return the cut $(X,Y)$. Notice that, since $N'>z$ holds, we are guaranteed that $|X|,|Y|\geq z/2$, as required. Therefore, we assume from now on that, whenever \ref{lem: matching player 1 it} is called, it returns a partial routing $\left ((M'_1,\ldots,M'_k),\pset'\right )$ of the vertex sets $A'_1,B'_1,\ldots,A'_k,B'_k$, of value at least $\frac{8N'\log n}{\ell^2}$, that causes congestion  $1$. 
	We then add the paths in $\pset'$ to $\pset$, and for each $1\leq i\leq k$, we add the matching $M'_i$ to $M_i$, and continue to the next iteration.
	
	The key in the analysis of the algorithm is to bound the number of iterations. For all $j\geq 1$, let $N'_j$ denote the parameter $N'$ at the beginning of iteration $j$. Then, since  \ref{lem: matching player 1 it}  returns a routing of value at least $\frac{8N'_j\log n}{\ell^2}$, we get that $N'_{j+1}\leq N_j(1-8\log n/\ell^2)$. Therefore, after $\ell^2$ iterations, parameter $N'_j$ is guaranteed to fall below $z$, and the algorithm will terminate. Notice that the congestion of the final routing $\pset$ is bounded by the number of iterations, $\ell^2$. Moreover, since the running time of each iteration is $\tilde O(k\ell |E(G)|+kn)$, the total running time of the algorithm is $\tilde O(k\ell^3 |E(G)|+kn\ell^2)$. In order to complete the proof of \ref{thm: matching player}, it is now enough to prove \ref{lem: matching player 1 it}.
	
	\begin{proofof}{\ref{lem: matching player 1 it}}
	Our algorithm is very similar to that employed in \cite{ChuzhoyK19}, and consists of two phases. In the first phase, we employ a simple greedy algorithm that attempts to compute a partial routing of sets $A_1',B'_1,\ldots,A'_k,B'_k$. If the resulting routing contains enough paths then we terminate the algorithm and return this routing. Otherwise, we proceed to the second phase, where we compute the desired cut.
	
	\paragraph{Phase 1: Route.}
	We use a simple greedy algorithm. Initially, we set, for all $1\leq i\leq k$, $M_i=\emptyset$, and we set $\pset=\emptyset$. The algorithm then iterates, as long as there is a path $P$ in $G$ of length at most $\ell$, that, for some $1\leq i\leq k$, connects some vertex $v\in A_i'$ to some vertex $u \in B'_i$. The algorithm computes any such path $P$, adds $(u,v)$ to $M_i$, and adds the path $P$ to $\pset$, denoting $P=P(u,v)$. We then delete every edge of $P$ from $G$, and we delete $u$ from $A'_i$ and $v$ from $B'_i$, and then continue to the next iteration. The algorithm terminates when, for each $1\leq i\leq k$, every path in the remaining graph $G$ connecting a vertex of $A'_i$ to a vertex of $B'_i$ has length greater than $\ell$ (or $A'_i=\emptyset$). It is easy to verify that, for each $1\leq i\leq k$, the final set $M_i$ is a matching between vertices of $A'_i$ and vertices of $B'_i$, and that $\pset$ is a collection of edge-disjoint paths, of length at most $\ell$ each, containing, for every pair $(u,v)\in \bigcup_iM_i$, a path $P(u,v)$ connecting $u$ to $v$ in $G$. If $\sum_i|M_i|\geq \left(\sum_{i=1}^kn'_i\right )\frac{8\log n}{\ell^2}$, then we terminate the algorithm, obtaining the desired partial routing. Otherwise, we continue to the second phase, where a cut $(X,Y)$ will be computed.
	
	We implement the algorithm for the first phase by using Even-Shiloach trees.

	\begin{lem}[\cite{EvenS,Dinitz}]
		There is a deterministic data structure, called ES-tree, that, given an unweighted undirected $n$-vertex graph $G$ undergoing edge deletions, a root node $s$, and a depth parameter $\ell$, maintains, for every vertex $v\in V(G)$ a value $\delta(s,v)$ such that $\delta(s,v) = \dist_G (s,v)$ if $\dist_G (s,v) \le \ell$ and $\delta(s,v) = \infty$ otherwise  (here, $\dist_G(s,v)$ is the distance between $s$ and $v$ in the current graph $G$). The data structure supports shortest-paths queries: given a vertex $v$, return a shortest path connecting $s$ to $v$ in $G$, if $\dist_G(s,v)\leq \ell$, and return $\infty$ otherwise. The total update time of the data structure is $\tilde O(|E(G)|\ell+n)$, and time needed to process each query is $O({|P|})$, where $P$ is the path returned in response to the query.
	\end{lem}

	We construct $k$ graphs $G_1,\ldots,G_k$, where graph $G_i$ is obtained from a copy of $G$, by adding a source vertex $s_i$ that connects to every vertex in $A'_i$ with an edge, and a destination vertex $t_i$, that connects to every vertex in $B'_i$ with an edge. 
	For each $1\leq i\leq k$, we then maintain an ES-tree in graph $G_i$, from source $s_i$, up to depth $\ell+2$. Note that the total update time needed in order to maintain all these ES-trees under edge deletions is $\tilde O(\ell k|E(G)|+kn)$. Our algorithm processes the graphs $G_i$ one-by-one. When graph $G_i$ is processed, we perform a number of iterations, as long as $\dist_{G_i}(s_i,t_i)\leq \ell+2$. In each such iteration, we perform a shortest-path query in the corresponding ES-tree for vertex $t_i$, obtaining a path $P$, of length at most $\ell+2$, connecting $s_i$ to $t_i$. By discarding the first and the last edge on this path, we obtain a path $P'$ of length at most $\ell$, connecting some vertex $v\in A_i'$ to some vertex $u \in B_i'$. We delete all edges on path $P'$ from all copies $G_1,\ldots,G_k$ of the graph $G$, and we delete $v$ and $u$ from $G_i$, updating all corresponding ES-trees. Note that the total time to respond to all queries is $O(|E(G)|)$, as whenever a path $P$ is returned, all its edges are deleted from all graphs $G_i$. Therefore, the total running time of the algorithm is $\tilde O(k\ell |E(G)|+kn)$.
	\end{proofof}
	
\paragraph{Phase 2: Cut.}
We use the following standard algorithm that follows the ball-growing paradigm.

\begin{claim}\label{claim: cut through ball growing}
	There is a deterministic algorithm, that, given an unweighted $n'$-vertex graph $H'$ with maximum vertex degree at most $\Delta$, and two sets $S,T$ of its vertices, such that every path connecting a vertex of $S$ to a vertex of $T$ in $H'$ has length greater than $\ell$, for some parameter $\ell>1$ computes, in time $O(|E(H')|)$, a cut $Z$ in $H'$, such that:
	\begin{itemize}
		\item $|Z|\leq n'/2$;
		\item either $S\subseteq Z$ or $T\subseteq Z$ hold; and
		\item $|E_{H'}(Z,V(H')\setminus Z)|<\frac{8\Delta\log n'}{\ell}\cdot |Z|$.
	\end{itemize}
\end{claim}

\begin{proof}
	Let $S_0=S$, and for all $j>0$, let $S_j$ contain all vertices of $S_{j-1}$, and all neighbors of vertices of $S_{j-1}$ in graph $H'$. We also define $T_0=T$, and for all $j>0$, we let $T_j$ contain all vertices of $T_{j-1}$, and all neighbors of vertices of $T_{j-1}$ in graph $H'$. We need the following standard observation:

\begin{observation}\label{obs: sparse cut for far away sets}
	There is an index $0\leq j<\ceil{\ell/4}$, such that either (i) $|S_{j+1}|<n'/2$ and $|E_{H'}(S_j,V(H')\setminus S_j)|<\frac{8\Delta\log n'}{\ell}\cdot |S_j|$; or (ii) $|T_{j+1}|<n'/2$ and $|E_{H'}(T_j,V(H')\setminus S_j)|<\frac{8\Delta\log n'}{\ell}\cdot |T_j|$.
\end{observation}

\begin{proof}
	Assume for contradiction that the claim is false. Let $j'$ be the smallest index, such that $|S_{j'}|>n'/2$ or $|T_{j'}|>n'/2$. Assume w.l.o.g. that $|T_{j'}|>n'/2$. 
	
	Assume first that $j'<\ell/2$. Then for all $1\leq j\leq \ceil{\ell/4}$, $|S_j|<n'/2$ must hold (as otherwise, there is a path connecting a vertex of $S$ to a vertex of $T$, of length at most $\ell$). However, from our assumption, for all $0\leq j<\ceil{\ell/4}$, $|E_{H'}(S_j,V(H')\setminus S_j)|>\frac{8\Delta\log n'}{\ell}\cdot |S_j|$. Since the maximum vertex degree in $H'$ is bounded by $\Delta$, we get that $|S_{j+1}\setminus S_j|\geq \frac{8\log n'}{ \ell}\cdot |S_j|$, and so
	$|S_{j+1}|\geq |S_j|\left (1+ \frac{8\log n'}{ \ell}\right )$. Overall, we get that
	$|S_{\ceil{\ell/4}}|\geq |S_0|\cdot \left (1+ \frac{8\log n'}{ \ell}\right )^{\ceil{\ell/4}}>\frac{n'}{2}$, a contradiction.
	
	Assume now that $j'\geq \ell/2$. Then we get that for all   $1\leq j\leq \ceil{\ell/4}$, $|T_j|<n'/2$ must hold. Applying the same reasoning as above to sets $T_j$, we conclude that $|T_{\ceil{\ell/4}}|\geq n'/2$, a contradiction.
\end{proof}

The algorithm performs two BFS searches in $H'$ simultaneously, one starting from $S$ and another starting from $T$, until an index $j$ with the properties guaranteed by \ref{obs: sparse cut for far away sets} is found. If $|S_{j+1}|<n'/2$ and $|E_{H'}(S_j,V(H')\setminus S_j)|<\frac{8\Delta\log n'}{\ell}\cdot |S_j|$, then we return $Z=S_j$; otherwise, and otherwise we return $Z=T_j$.
\end{proof}

We are now ready to describe the algorithm for Phase 2.
	For convenience, we denote $N=\sum_{i=1}^kn'_i$.
	Recall that Phase 2 is only executed if the routing $\pset$ computed in Phase 1 contains fewer than $\frac{8N\log n}{\ell^2}$ paths. Let $E'$ be the set of all edges lying on the paths in $\pset$, so $|E'|\leq \frac{8N\log n}{\ell}$ (as the length of every path in $\pset$ is at most $\ell$), and let $H=G\setminus E'$. We also denote, for all $1\leq i\leq k$, by $A''_i\subseteq A'_i$ the subset of all vertices of the original set $A'_i$ that do not participate in the matching $M_i$, and we define $B''_i\subseteq B'_i$ similarly. Notice that for all $1\leq i\leq k$, if $A''_i,B''_i\neq \emptyset$, then the length of the shortest path, connecting a vertex of $A''_i$ to a vertex of $B''_i$ is greater than $\ell$.
	
	Our algorithm is iterative. We maintain a subgraph $H'$ of $H$, that is initially set to be $H$. In every iteration $i$, we compute a subset $U_i\subseteq V(H')$ of vertices of $H'$, such that $|U_i|\leq |V(H')|/2$, and 
	$|E_{H'}(U_i,V(H')\setminus U_i)|<\frac{8\Delta\log n}{\ell}\cdot |U_i|$. We then delete, from graph $H'$, all vertices of $U_i$, and continue to the next iteration. %Our algorithm will guarantee that, in the end, for all $1\leq j\leq k$, either $A''_j$ is contained in $\bigcup_iU_i$, or $B''_j$ is contained in $\bigcup_iU_i$ (or both). 
	Throughout the algorithm, we may update the sets $A''_j$ and $B''_j$, by removing some vertices from them.
	
	The algorithm is executed as long as there is some index $1\leq j\leq k$, with $A''_j,B''_j\neq \emptyset$, and as long as $|\bigcup_iU_i|\leq n/4$; if either of these conditions do not hold, the algorithm is terminated. We now describe the $i$th iteration of the algorithm, and we let $1\leq j\leq k$ be an index for which $A''_j,B''_j\neq \emptyset$. We apply the algorithm from \ref{claim: cut through ball growing} to the current graph $H'$, and the sets $S=A''_j$, $T=B''_j$ of vertices; recall that every path connecting a vertex of $A''_j$ to a vertex of $B''_j$ in $H'$ has length greater than $\ell$. Let $Z$ be the cut returned by the algorithm. We set $U_i=Z$. We also denote by $E_i=E_{H'}(Z,V(H')\setminus Z)$. Recall that we are guaranteed that $|E_i|\leq \frac{8\Delta\log n}{\ell}\cdot |U_i|$. Moreover, either $A''_j\subseteq U_i$, or $B''_j\subseteq U_i$. We update the current graph $H'$, by deleting the vertices of $U_i$ from it. For all $1\leq j'\leq k$, we delete from $A''_{j'}$ and from $B''_{j'}$ all vertices that lie in the set $U_i$.
	
	Let $q$ be the number of iterations in the algorithm; it is easy to see that $q\leq k$. Therefore, the running time of the algorithm in Phase 2 so far is $O(k\cdot |E(H)|)=O(k\cdot |E(G)|)$. Let $U=\bigcup_{i=1}^rU_i$, and let $\hat E=\bigcup_{i=1}^rE_i$.
	
	If the algorithm terminated because $|U|\geq n/4$, then we are guaranteed that $|U|\geq N/2$, as $N\leq n/2$ must hold. Otherwise, we are guaranteed that for all $1\leq j\leq k$, either $A''_j=\emptyset$ (and so $A'_j\subseteq U$), or $B''_j=\emptyset$ (and so $B'_j\subseteq U$). In the latter case, we get that:
	
	 \[|U|\geq \sum_{j=1}^kn'_j-|\pset|\geq N-\frac{8N\log n}{\ell^2}\geq N/2,\]
	 
	 since we have assumed that $\ell\geq 32\Delta \log n$. 
	 Moreover, it is immediate to verify that $|\hat E|\leq \frac{8\Delta\log n}{\ell}\cdot |U|$. 
	
	Consider now the original graph $H$. We define a cut $(X,Y)$ in $H$ by setting $X=U$ and $Y=V(H)\setminus U$. Since $|E(G)\setminus E(H)|=|E'|\leq \frac{8N\log n}{\ell}\leq \frac{16|U|\log n}{\ell}$, we get that $|E_G(X,Y)|\leq |\hat E|+|E'|\leq \frac{24\Delta\log n}{\ell}\cdot |X|$.
	
	Next, we claim that $|X|\leq 3n/4$. Indeed, we are guaranteed that $\sum_{i=1}^{q-1}|U_i|\leq n/4$, and so $U_q\leq \frac {n-\sum_{i=1}^{q-1}|U_i|}2$. We then get that altogether, $|X|=\sum_{i=1}^q|U_i|\leq \frac{n}{2}+\frac{\sum_{i=1}^{q-1}|U_i|}2\leq \frac{3n}{4}$. In particular, $|Y|\geq n/4$ and so $|Y|\geq |X|/3$. Therefore, $|E_G(X,Y)|\leq \frac{24\Delta\log n}{\ell}\cdot |X|\leq \frac{72\Delta\log n}{\ell}\cdot\min\set{|X|,|Y|}$, and so $\Psi_G(X,Y)\leq  \frac{72\Delta\log n}{\ell}$.
	As observed already, $|X|\geq N/2=\sum_in'_i/2$, and $|Y|\geq n/4\geq \sum_in'_i/2$, as $\sum_in'_i\leq n/2$ must hold.
\end{proof}

    The following corollary follows immediately from \ref{thm: matching player}, by setting the parameter 	$\ell=144\Delta\log n/\psi$.
    
    \begin{cor}\label{cor: matching player}
    	There is a deterministic algorithm, that we call \routeorcut, that, given an $n$-vertex graph $G=(V,E)$ with maximum vertex degree $\Delta$, disjoint subsets $A_1,B_1,\ldots,A_k,B_k$ of its vertices, where for all $1\leq i\leq k$, $|A_i|\leq |B_i|$ and $|A_i|=n_i$, an integer $z\geq 0$, and a parameter $0<\psi<1/2$, computes one of the following:
    	
    	\begin{itemize}
    		\item either a partial routing of the sets $A_1,B_1,\ldots,A_k,B_k$, of value at least $\sum_in_i-z$, that causes congestion at most $O(\Delta^2\log^2n/\psi^2)$; or
    		
    		\item a cut $(X,Y)$ in $G$, with $|X|,|Y|\geq z/2$, and $\Psi_G(X,Y)\leq \psi$.
    	\end{itemize}
    	
    	The running time of the algorithm is $\tilde O(\Delta^3k|E(G)|/\psi^3+k\Delta^2 n/\psi^2)$.
    \end{cor}

\subsection*{An Improved Algorithm for $k=1$}
For the special case where $k=1$, we provide a somewhat faster algorithm, summarized in the following theorem. We note that this algorithm is not 
essential for the proof of our main result (\ref{thm:intro:main}), but we can use it to provide a self-contained proof of the theorem with a somewhat slower running time, which we believe is of independent interest.

\begin{thm}\label{thm:push match}
	There is a deterministic algorithm, that we call \routeorcutp, that, given a connected $n$-vertex $m$-edge graph $G=(V,E)$ with maximum vertex degree $\Delta$, two disjoint subsets $A_1,B_1$ of its vertices, where $|A_1|\leq |B_1|$ and $|A_1|=n_1$, an integer $z\geq 0$, and a parameter $0<\psi<1/2$, computes one of the following:
	
	\begin{itemize}
		\item either a partial routing of the sets $A_1,B_1$, of value at least $n_1-z$, that causes congestion at most $4\Delta/\psi$; or
		
		\item a cut $(X,Y)$ in $G$, with $|X|,|Y|\geq z/\Delta$, and $\Psi_G(X,Y)\leq \psi$.
	\end{itemize}
	
	The running time of the algorithm is $O\left (\frac{m \Delta \log m}{\psi}\right)$.
	\end{thm}

\begin{proof}
\ref{thm:push match} is an easy application of either the \emph{bounded-height} variant of the push-relabel-based algorithm of Henzinger,
Rao and Wang \cite{HenzingerRW17} for max-flow,   or the \emph{bounded-height} variant of
the blocking-flow-based algorithms by Orrecchia and Zhu \cite{OrecchiaZ14}.\footnote{Both algorithms are designed to have \emph{local} running time, that is, they may not read the whole graph. However, we do not need to use 
this  property here.}

We start by introducing some basic notation. Suppose we are given an unweighted
undirected graph $G=(V,E)$.
We let $S:V\rightarrow \mathbb{Z}_{\ge0}$
denote a \emph{source function} and $T:V\rightarrow\mathbb{Z}_{\ge0}$
denote a \emph{sink function}. For a vertex $v\in V$, we sometimes call $T(v)$ its \emph{sink capacity}. Intuitively, initially, for every vertex $v\in V$, we have $S(v)$ units of mass (substance that needs to be routed) placed on vertex $v$. Additionally, every vertex $v\in V$ may absorb up to $T(v)$ units of mass. Our goal is to route the initial mass across the graph (using standard single-commodity flow) so that all mass is absorbed. We use a flow function $f: V\times V\rightarrow \mathbb {R}$, that must satisfy: (i) for all $u,v\in V$, $f(u,v)=-f(v,u)$; and (ii) if $(u,v)\not\in E$, then $f(u,v)=0$. Whenever $f(u,v)>0$, we interpret it as $f(u,v)$ units of mass are sent via the edge $(u,v)$ from $u$ to $v$, while $f(u,v)<0$ means that the same amount of mass is sent in the opposite direction. 

 We require that $\sum_{v\in V}S(v)\le\sum_{u}T(u)$, that is, the total amount of mass that needs to be routed is bounded by the total sink capacities of the vertices.
Given a flow $f:V\times V\rightarrow \mathbb{R}$, the \emph{congestion}
of $f$ is $\max_{e\in E}|f(e)|$. We say that $f$ is a \emph{preflow}
if, for every vertex $v\in V$, $\sum_{u\in V}f(v,u)\le S(v)$; in other words, the net amount of mass routed away from any node $v$ is bounded by the amount of the source mass $S(v)$. For every vertex $v\in V$, we also denote by $f(v)=S(v)+\sum_{u\in V}f(u,v)$ the amount of mass that remains at $v$ after the routing $f$. We define
the \emph{absorbed mass} of a node $v$ as $\ab_{f}(v)=\min\set{f(v),T(v)}$,
and the \emph{excess of $v$} as $\ex_{f}(v)=f(v)-\ab_f(v)$,
measuring the amount of flow that remains at $v$ and cannot be absorbed by it.
Note that, if $\ex_{f}(v)=0$
for every vertex $v$, then all the mass is successfully routed to the sinks.
Let $\ex_{f}(V)=\sum_{v}\ex_{f}(v)$ denote the total amount
of mass that is not absorbed by the sinks.

The following lemma easily follows from Theorem 3.3 in \cite{NanongkaiSW17} (or Theorem 3.1
in \cite{HenzingerRW17}).
\begin{lem}
\label{lem:unitflow} There is a deterministic algorithm, that, given an $m$-edge graph $G=(V,E)$, a source function $S:V\rightarrow \mathbb{Z}_{\ge0}$, a sink function $T:V\rightarrow \mathbb{Z}_{\ge0}$, and a parameter $0<\phi\leq 1$, such that $\sum_{v\in V}S(v)\le\sum_{v\in V}T(v)$, and for every vertex $v\in V$,  $S(v)\le \deg_{G}(v)$ and $T(v)\le\deg_{G}(v)$, computes, in time $O\left (\frac{m \log m}{\phi}\right)$, 
an integral preflow $f$ of congestion at most $4/\phi$. Moreover, if the total excess $\ex_{f}(V)>0$, then the algorithm also
computes a cut $(S,\overline{S})$ with $\Phi_{G}(S)<\phi$ and  $\vol_{G}(S),\vol_{G}(\overline{S})\ge \ex_{f}(V)$.
\end{lem}

We are now ready to complete the proof of  \ref{thm:push match}. 
For convenience, we denote $A_1$ by $A$, $B_1$ by $B$, and $n_1$ by $N$.
For the input graph $G=(V,E)$, we define a source function as follows: for all $v\in A$, $S(v)=1$, and for all other vertices, $S(v)=0$. Similarly, we define the sink function to be $T(v)=1$ if $v\in B$, and $T(v)=0$ otherwise.

We then apply the algorithm from \ref{lem:unitflow} to graph $G$, source function $S$, sink function $T$ and parameter $\phi=\psi/\Delta$. Let $f$ be the resulting preflow
with congestion at most $4/\phi\leq 4\Delta/\psi$. The running time of the algorithm is $O\left (\frac{m \log m}{\phi}\right)=O\left (\frac{m \Delta \log m}{\psi}\right)$

We now consider two cases. The first case happens when $\ex_{f}(V)\ge z$. In this case, we obtain a cut $(X,Y)$ with $\Phi_{G}(X,Y)<\phi$ and  $\vol_{G}(X),\vol_{G}(Y)\ge \ex_{f}(V)\geq z$. Since the maximum vertex degree in $G$ is bounded by $\Delta$, we get that $|X|,|Y| \geq z/\Delta$. Moreover, from \ref{prop:sparsity vs conductance}, $\Psi_G(X,Y)\leq \Delta\Phi_G(X,Y)\leq \Delta \phi\leq \psi$.

Consider now the second case, where $\ex_{f}(V)<z$. Let $B'$ be a multi-set of vertices, where for each vertex $v\in V$, we add $\ex_f(v)$ copies of $v$ into $B'$ (since $f$ is integral, so is $\ex_f(v)$ for all $v\in V$). Then $|B'|\leq z$, and $f$ defines a valid integral flow from $A$ to $B\cup B'$, with congestion at most $4\Delta/\psi$, such that all but at most $z$ flow units terminate at distinct vertices of $B$. It now remains to compute a decomposition of $f$ into flow-paths, and then discard the flow-paths that terminate at vertices of $B'$. This can be done by using, for example, the link-cut tree \cite{SleatorT83}, or simply a standard Depth-First Search. 
For the latter, construct a graph $G'$, obtained from $G$ by creating $|f(e)|$ parallel copies of every edge $e\in E(G)$, that are directed along the direction of the flow $f$ on $e$; recall that $|f(e)|\leq 4\Delta/\psi$. We also add a source $s$ that connects to every vertex of $A$ with a directed edge. We then perform a DFS search of the resulting graph $G'$, starting from $s$. If the DFS search leaves some vertex $v$ without reaching any vertex of $B\cup B'$, then we delete $v$ from the graph $G'$. If the search reaches a vertex $v\in B\cup B'$, then we retrace the current path from $s$ to $v$, adding it to the path-decomposition that we are constructing, and deleting all edges on this path from $G'$. We then restart 
the DFS search. It is easy to verify that every edge is traversed at most twice throughout this procedure, and so the total running time is $O|E(G')|=O(|E(G)|\cdot \Delta/\psi)$. Let $\P$ be the final collection of paths that we obtain. Then every vertex of $A$ has exactly one path in $\P$ originating from it, and all but at most $z$ paths in $\P$ terminate at distinct vertices of $B$. We discard from $\P$ all paths that do not terminate at vertices of $B$, obtaining the desired final collection of paths. The total running time of the algorithm is $O\left (\frac{m \Delta \log m}{\psi}\right)$.
\end{proof}

\section{Deterministic Cut-Matching Game: Proof of \ref{thm: cut player}}
\label{sec: cut player}

The goal of this section is to prove \ref{thm: cut player}.
We do so using the following theorem, that can be thought of as a restatement of \ref{thm: cut player} in a way that will be more convenient to work with in our inductive proof. Recall that $\cCMG$ is the constant from \ref{thm:KKOV-new}.

\begin{theorem}\label{thm: cut player rec}
	There are universal constants $c_0$, $N_0$ and a deterministic algorithm, that, given an $n$-vertex graph $G=(V,E)$ and parameters $N,q$ with $N>N_0$ an integral power of $2$, and $q\ge 1$ an integer, such that $n\leq N^q$, and the maximum vertex degree in $G$ is at most $\cCMG\log n$, computes one of the following:
		
		\begin{itemize}
			\item either a cut $(A,B)$ in $G$ with $|A|,|B|\geq n/4$ and $|E_G(A,B)|\leq n/100$; or
			\item a subset $S\subseteq V$ of at least $n/2$ vertices, such that $\Psi(G[S])\geq 1/\left (q\log N\right )^{8q}$.
		\end{itemize}
		
		The running time of the algorithm is $O\left (N^{q+1}\cdot(q\log N)^{c_0q^2}\right )$.
	\end{theorem}

We first show that \ref{thm: cut player} follows from \ref{thm: cut player rec}. 
The parameter $N_0$ in \ref{thm: cut player} remains the same as that in \ref{thm: cut player rec}.
Assume that we are given an $n$-vertex graph and a parameter $r$, such that $n^{1/r}\geq N_0$. We set $q=r$, and we let $N$ be the smallest integral power of $2$ such that $N\geq n^{1/q}$; observe that $(N/2)^q\leq n\leq N^q$ and $N\geq N_0$ hold. Moreover, since $q\log (N/2)\leq \log n$, if $N_0$ is a large enough constant, then $q\log N\leq 2\log n$.

We apply the algorithm from \ref{thm: cut player rec} to graph $G$ with the parameter $q$. If the outcome is a cut $(A,B)$ 
with $|A|,|B|\geq n/4$ and $|E(A,B)|\leq n/100$, then we return this cut as the outcome of the algorithm. Otherwise, we obtain a subset  $S\subseteq V$ of at least $n/2$ vertices, such that $\Psi(G[S])\geq 1/\left (q\log N\right )^{8q}\geq 1/\left(2\log n\right)^{8q}\geq \Omega\left(1/(\log n)^{O(r)}\right )$, as required.
Lastly, the running time of the algorithm is $O\left (N^{q+1}\cdot(q\log N)^{c_0q^2}\right )=O\left (n^{1+O(1/r)}\cdot (\log n)^{O(r^2)}\right )$.

The remainder of this section is dedicated to proving \ref{thm: cut player rec}.
 The proof is by induction on the parameter $q$. We start with the base case where $q=1$ and then show the step for $q>1$.

\subsection{Base Case: $q=1$}
The algorithm uses the following key theorem.

\begin{thm}\label{thm: r=1 iterations}
	There is a deterministic algorithm that, given as input a graph $G'=(V',E')$ with $|V'|=n'$ and maximum vertex degree $\Delta=O(\log n')$, in time $\tilde O((n')^2)$ returns one of the following:
	
	\begin{itemize}
		\item either a subset $S\subseteq V'$ of at least $2n'/3$ vertices such that $G'[S]$ is an $\Omega(1/\log^5n')$-expander; or
		\item a cut $(X,Y)$ in $G'$ with $|X|,|Y|\geq \Omega(n'/\log^5n')$ and $\Psi_{G'}(X,Y)\leq 1/100$.
	\end{itemize}
\end{thm}

We prove \ref{thm: r=1 iterations} below, after we complete the proof of \ref{thm: cut player} for the case where $q=1$ using it. 
Our algorithm performs a number of iterations. We maintain a subgraph $G'\subseteq G$; at the beginning of the algorithm, $G'=G$. In the $i$th iteration, we compute a subset $S_i\subseteq V(G')$ of vertices, and then update the graph $G'$ by deleting the vertices of $S_i$ from it. The iterations are performed as long as $|\bigcup_iS_i|< n/4$.

In order to execute the $i$th iteration, we consider the current graph $G'$, denoting $|V(G')|=n'$. Note that, since we assume that $|\bigcup_{i'<i}S_{i'}|<n/4$, we get that $n'\geq 3n/4$. We then apply \ref{thm: r=1 iterations} to graph $G'$. If the outcome is a subset $S\subseteq V'$ of at least $2n'/3$ vertices such that $G'[S]$ is an $\Omega(1/\log^5n')$-expander, then we terminate the algorithm and return $S$; in this case we say that the iteration terminated with an expander. Notice that, since $n'\geq 3n/4$, and $|S|\geq 2n'/3$, we are guaranteed that $|S|\geq n/2$.
Moreover, assuming that $N_0$ is a large enough constant, the expansion of $G[S]$ is at least $\Omega(1/\log^5n')\geq 1/\log^8n\geq 1/\log^8N$ as required.
 Otherwise, we obtain a cut $(X,Y)$ in $G'$ with $|X|,|Y|\geq \Omega(n'/\log^5n')$ and $\Psi_{G'}(X,Y)\leq 1/100$; in this case we say that the iteration terminated with a cut. 
 Assume w.l.o.g. that $|X|\leq |Y|$. We then set $S_i=X$,  update the graph $G'$ by removing the vertices of $S_i$ from it, and continue to the next iteration. If the algorithm does not terminate with an $1/\log^8N$-expander, then it terminates once $|\bigcup_{i'}S_{i'}|\geq n/4$ holds. Let $i$ denote the number of iterations in this case. Since we are guaranteed that $|\bigcup_{i'<i}S_{i'}|<n/4$, while $|S_i|\leq n'/2\leq 3n/8$, we get that $n/4\leq |\bigcup_{i'=1}^iS_{i'}|\leq 5n/8$. Let $A=\bigcup_{i'=1}^iS_{i'}$, and let $B=V(G)\setminus A$. From the above discussion,  we are guaranteed that $|A|,|B|\geq n/4$, and moreover, since the cut $S_{i'}$ that we obtain in every iteration $i'$ has sparsity at most $1/100$ in its current graph $G'$, it is easy to verify that $|E_G(A,B)|\leq |A|/100\leq n/100$. We then return the cut $(A,B)$ as the algorithm's outcome. Since for all $1\leq i'\leq i$, $|S_{i'}|\geq \Omega(n/\log^5n)$, the number of iterations is bounded by $O(\log^6n)$, and so the total running time of the algorithm is $\tilde O(n^2)=O(N^2\log^{c_0}n)$, if $c_0$ is a large enough constant.
In the remainder of this subsection we focus on the proof of \ref{thm: r=1 iterations}.

\subsection*{Proof of \ref{thm: r=1 iterations}}
As our first step, we use Algorithm \constructexpander from \ref{thm:explicit expander} to construct a $\psi^*$-expander $H=H_{n'}$ on $n'$ vertices, with $\psi^*=\Psi(H)=\Omega(1)$, such that maximum vertex degree in $H$ is at most $9$. 
We identify the vertices of $H$ with the vertices of $G'$, so $V(H)=V'$.
The running time of this step is $O(n')$. Using a simple greedy algorithm, and the fact that the maximum vertex degree in $H$ is at most $9$, we can partition the set $E(H)$ of edges into $17$ matchings, $M_1,\ldots,M_{17}$. We then perform up to $17$ iterations; in each iteration $i$, we will either embed the edges of $M_i$ into $G'$, after possibly adding a small number of fake edges to it, or we will compute the desired cut $(A,B)$ in $G'$.

The $i$th iteration is executed as follows. We denote $M_i=\set{e_1,\ldots,e_{k_i}}$, where the edges are indexed in an arbitrary order. For each $1\leq j\leq k_i$, denote $e_j=(u_j,v_j)$. We define two corresponding sets $A_j,B_j$ of vertices of $G'$, where $A_j=\set{u_j}$ and $B_j=\set{v_j}$. We then apply Algorithm \routeorcut from \ref{cor: matching player} to graph $G'$, the sets $A_1,B_2,\ldots,A_{k_i},B_{k_i}$ of its vertices, integer $z=\ceil{\frac{n'}{c\log^5n'}}$ for some large enough constant $c$, and parameter $\psi=1/100$. Recall that the running time of the algorithm is $\tilde O(k_i|E(G')|\Delta^3/\psi^3+k_in'/\Delta^2)=\tilde O((n')^2)$.
We now consider two cases. If the algorithm returns a cut $(X,Y)$, with $\Psi_{G'}(X,Y)\leq \psi$, then we terminate the algorithm and return this cut; in this case, $|X|,|Y|\geq z/2\geq \Omega(n'/\log^5n')$ must hold. Otherwise, the algorithm computes a partial routing $\pset$ of the sets $A_1,B_1,\ldots,A_{k_i},B_{k_i}$, of value at least $k_i-z$, that causes congestion at most $O(\Delta^2\log^2n'/\psi^2)=O(\log^4n')$. Let $M'_i\subseteq M_i$ be the subset of edges that are routed in $\pset$, so for every edge $e_j\in M'_i$ there is a path $P(e_j)\in \pset$ connecting its endpoints. Let $M''_i\subseteq M_i$ denote the set of the remaining edges, so $|M''_i|\leq z$. We let $F_i=M''_i$ be a set of fake edges in graph $G'$, that we use in order to route the edges of $M''_i$. For each edge $e_j\in M''_i$, we let $P(e_j)$ be the path consisting of the new fake copy of $e_j$ in $G'$. Let $\pset_i=\pset\cup \set{P(e_j)\mid e_j\in M''_i}$. We now obtained an embedding of the edges of $M_i$ into $G'+F_i$, with congestion $O(\log^4n')$.

If the algorithm never terminates with the cut  $(X,Y)$ with  $\Psi_{G'}(X,Y)\leq \psi$, then, after $17$ iterations, we obtain an embedding $\pset^*=\bigcup_{i=1}^{17}\pset_i$ of $H$ into $G+F$, where $F=\bigcup_{i=1}^{17}F_i$ is a set of at most $17z$ fake edges; the congestion of the embedding is $\eta=O(\log^4n')$.
Moreover, if we denote by $\Delta_G$ the maximum vertex degree in $G+F$, then $\Delta_G\leq 17+\Delta\leq 17\Delta$.
 Next, we apply Algorithm \extractexpander  from \ref{lem: embedding expander w fake edges gives expander} to graphs $G'$, $H$, the set $F$ of fake edges, and the embedding $\pset^*$ of $H$ into $G$. 
Since $\frac{\psi^*n'}{32\Delta_G\eta}\geq \frac{\psi^*n'}{O(\Delta\log^4n')}\geq \frac{n'}{O(\log^5n')}$, by letting the constant $c$ used in the definition of $z$ be large enough, we ensure that $|F|\leq 17z\leq  \frac{\psi^*n'}{32\Delta_G\eta}$, as required.
The algorithm from \ref{lem: embedding expander w fake edges gives expander} then computes a subgraph $G''\subseteq G'$ that is a $\psi'$-expander, where $\psi'\geq \frac{\psi^*}{6\Delta_G\eta}=\Omega\left (\frac{1}{\log^5n'}\right)$, with:

\[V(G'')\geq n'-\frac{4\cdot 17z\eta}{\psi^*}=n'-O(z\log^4n') \]

By letting $c$ be a large enough constant, we can ensure that $|V(G'')|\geq 2n'/3$. The running time of Algorithm \extractexpander from \ref{lem: embedding expander w fake edges gives expander} is $\tilde O(|E(G')|\Delta_G\eta/\psi^*)=\tilde O(n')$, and so the total running time of the algorithm is $\tilde O((n')^2)$.

%---------------------------------------------
%---------------------------------------------
%---------------------------------------------
\subsection{Step: $q>1$}
%---------------------------------------------
%---------------------------------------------
%---------------------------------------------
Suppose we are given an integer $q>1$. We assume that \ref{thm: cut player rec} holds for $q-1$: that is, there is a deterministic algorithm, that we denote by $\aset(q-1)$, that, given an $n$-vertex graph $G$ with maximum vertex degree at most $\cCMG \log n$ and $n\leq N^{q-1}$, for some $N>N_0$, either returns a cut $(A,B)$ in $G$ with $|A|,|B|\geq n/4$ and $|E(A,B)|\leq n/100$, or it computes a subset $S\subseteq V(G)$ of at least $n/2$ vertices, such that $\Psi(G[S])\geq \psi_{q-1}$, where $\psi_{q-1}=1/\left ((q-1)\log N)^{8(q-1)}\right )$. We denote the running time of this algorithm by $T(q-1)=O\left (N^{q}\cdot((q-1)\log N)^{c_0(q-1)^2}\right )$. Throughout the proof, we also denote $\psi_{q}=1/\left (q\log N)^{8q}\right )$

We now prove that the theorem holds for the given value of $q$, by invoking Algorithm $\aset(q-1)$ a number of times. 
The following theorem is central to proving the induction step.

\begin{thm}\label{thm: r>1 iterations}
	There is a deterministic algorithm that, given as input an $n'$-vertex graph $G'=(V',E')$  and integers $N,q$ with $N>N_0$ an integral power of $2$ and $q>1$, such that $N^{q-1}/2\leq n'\leq N^q$, and maximum vertex degree of $G'$ is $\Delta=O(\log n')$, returns one of the following:
	
	\begin{itemize}
		\item either a subset $S\subseteq V'$ of at least $2n'/3$ vertices such that $G'[S]$ is a $\psi_q$-expander; or
		\item a cut $(X,Y)$ in $G'$ with $|X|,|Y|\geq \Omega\left (\frac{\psi_{q-1}\cdot n'}{\log^8n'}\right )$ and $\Psi_{G'}(X,Y)\leq 1/100$.
	\end{itemize}

The running time of the algorithm is  $ O\left (N^{q+1}\cdot (q\log N)^{8q+O(1)} \right )+O\left (N\cdot \log n'\right )\cdot T(q-1)$.
\end{thm}

We prove \ref{thm: r>1 iterations} below, after we complete the proof of \ref{thm: cut player} for the current value of $q$ using it. 

Note that we can assume that $n>N^{q-1}$, since otherwise we can use algorithm $\aset(q-1)$, to either compute a cut $(A,B)$ in $G$ with $|A|,|B|\geq n/4$ and $|E(A,B)|\leq n/100$, or to compute a subset $S\subseteq V(G)$ of at least $n/2$ vertices, such that $\Psi(G[S])\geq \psi_{q-1}\geq \psi_q$, in time $T(q-1)=O\left (N^{q}\cdot((q-1)\log N)^{c_0(q-1)^2}\right )$.

Our algorithm performs a number of iterations. We maintain a subgraph $G'\subseteq G$; at the beginning of the algorithm, $G'=G$. In the $i$th iteration, we compute a subset $S_i\subseteq V(G')$ of vertices, and then update the graph $G'$ by deleting the vertices of $S_i$ from it. The iterations are performed as long as $|\bigcup_iS_i|< n/4$.

In order to execute the $i$th iteration, we consider the current graph $G'$, denoting $|V(G')|=n'$. Note that, since we assume that $|\bigcup_{i'<i}S_{i'}|<n/4$, we get that $n'\geq 3n/4$, and in particular $N^{q-1}/2\leq n'\leq N^q$. We then apply \ref{thm: r>1 iterations} to graph $G'$. If the outcome is a subset $S\subseteq V'$ of at least $2n'/3$ vertices such that $G'[S]$ is a $\psi_q$-expander, then we terminate the algorithm and return $S$. Notice that, since $n'\geq 3n/4$, and $|S|\geq 2n'/3$, we are guaranteed that $|S|\geq n/2$. Otherwise, we obtain a cut $(X,Y)$ in $G'$ with $|X|,|Y|\geq \Omega\left (\frac{\psi_{q-1}\cdot n'}{\log^8n'}\right )$ and $\Psi_{G'}(X,Y)\leq 1/100$. Assume w.l.o.g. that $|X|\leq |Y|$. We then set $S_i=X$,  update the graph $G'$ by removing the vertices of $S_i$ from it, and continue to the next iteration. If the algorithm does not terminate with a $\psi_q$-expander, then it terminates once $|\bigcup_{i'}S_{i'}|\geq n/4$ holds. Let $i$ denote the number of iterations in this case. Since we are guaranteed that $|\bigcup_{i'<i}S_{i'}|<n/4$, while $|S_i|\leq n'/2\leq 3n'/8$, we get that $n/4\leq |\bigcup_{i'=1}^iS_{i'}|\leq 5n/8$. Let $A=\bigcup_{i'=1}^iS_{i'}$, and let $B=V(G)\setminus A$. From the above discussion,  we are guaranteed that $|A|,|B|\geq n/4$, and moreover, since the cut $S_{i'}$ that we obtain in every iteration $i'$ has sparsity at most $1/100$ in its current graph $G'$, it is easy to verify that $|E_G(A,B)|\leq |A|/100\leq n/100$. We then return the cut $(A,B)$ as the algorithm's outcome.

Notice that the number of iterations in the algorithm is bounded by:

\[O(\log^9n/\psi_{q-1})=O\left((q\log N)^{8(q-1)}\cdot \log^9n\right )\leq O\left(\left (q\log N\right )^{8q+1}\right ),\]

since $n\leq N^q$. Therefore, the total running time of the algorithm is at most:

\[ O\left (N^{q+1} \cdot (q\log N)^{16q+O(1)}\right )+O\left (N(q\log N)^{8q+2}\right )\cdot T(q-1).\]

From the induction hypothesis, $T(q-1)=O\left (N^{q}\cdot((q-1)\log N)^{c_0(q-1)^2}\right )$. Assuming that $q\geq 1$, %for a sufficiently large constant $q_0$
 and that $c_0$ is a large enough constant, we get that the running time is $T(q)=O\left (N^{q+1}\cdot(q\log N)^{c_0q^2}\right )$, as required.

In the remainder of this subsection we focus on the proof of \ref{thm: r>1 iterations}.

\subsection*{Proof of \ref{thm: r>1 iterations}}
One of the main technical tools in the proof of the theorem is composition of expanders that we discuss next.

\subsubsection{Composing Expanders}
Suppose we are given a collection $\set{G_1,\ldots,G_h}$ of disjoint graphs, where for all $1\leq i\leq h$, the set $V(G_i)$ of vertices, that is denoted by $V_i$, has cardinality at least $N$. Let $H$ be another graph, whose vertex set is $\set{v_1,\ldots,v_h}$. 
An \emph{$N$-composition} of $H$ with $G_1,\ldots,G_h$ is another graph $G$, whose vertex set is $\bigcup_{i=1}^hV_i$, and whose edge set consists of two subsets: set $E^1=\bigcup_{i=1}^hE(G_i)$, and another set $E^2$ of edges, defined as follows: for each edge $e=(v_i,v_j)\in E(H)$, let $M(e)$ be an arbitrary matching of cardinality $N$ between vertices of $V_i$ and vertices of $V_j$. Then $E^2=\bigcup_{e\in E(H)}M(e)$. 
The following theorem shows that, if each of the graphs $G_1,\ldots,G_h$ is a $\psi$-expander, and graph $H$ is a $\psi'$-expander, then the resulting graph $G$ is also an expander for an appropriately chosen expansion parameter.

\begin{thm}\label{thm: expander composition}
Let $G_1,\ldots,G_h$ be a collection of $h>1$ graphs, such that for each $1\leq i\leq h$, $N\leq |V(G_i)|\leq \gamma N$, and $G_i$ is a $\psi$-expander, for some $N\geq 1$, $\gamma\geq 1$, and $0<\psi\leq 1$. Let $H$ be another graph with vertex set $\set{v_1,\ldots,v_h}$, such that $H$ is a $\psi'$-expander, and let $\Delta$ be maximum vertex degree in $H$. Lastly, let $G$ be a graph that is an $N$-composition of $H$ with $G_1,\ldots,G_h$. Then graph $G$ is a $\psi''$-expander, for $\psi''=\psi\psi'/(16\Delta\gamma^2)$.
\end{thm}

\begin{proof}
	For convenience, for all $1\leq i\leq h$, we denote $V(G_i)$ by $V_i$.
	Let $(A,B)$ be any partition of $V(G)$. It is sufficient to prove that $|E_G(A,B)|\geq \psi''\cdot \min\set{|A|,|B|}$. 
%	Assume without loss of generality that $|A|\leq |B|$.
	
	Consider any graph $G_i$, for $1\leq i\leq h$. We say that $G_i$ is of type 1 if $|V_i\cap A|> \left(1-\frac{1}{2\gamma}\right )|V_i|$, and we say that it is of type 2 if $|V_i\cap B|> \left(1-\frac{1}{2\gamma}\right ) |V_i|$. Notice that a graph $G_i$ cannot belong to both types simultaneously, and it is possible that it does not belong to either type. Let $N_1$ be the number of type-1 graphs $G_i$, and let $N_2$ be the number of type-2 graphs. Assume w.l.o.g. that $N_1\leq N_2$. Let $S\subseteq V(H)$ contain all vertices $v_i$, such that $G_i$ is a type-1 graph, so $|S|=N_1$. Since graph $H$ is a $\psi'$-expander, $|E_H(S,V(H)\setminus S)|\geq \psi'|S|$.
	
	We partition the set $A$ of vertices into two subsets: set $A'$ contains all vertices that lie in type-1 graphs $G_i$, and set $A''$ contains all remaining vertices. Recall that graph $G$ contains, for every edge $e=(v_i,v_j)\in  E_H(S,V(H)\setminus S)$, a collection $M(e)$ of $N$ edges, connecting vertices of $V_i$ to vertices of $V_j$.
	Consider any such edge $e=(v_i,v_j)$, with $v_i\in S$. Since $|V_i\cap A|\geq \left(1-\frac{1}{2\gamma}\right )|V_i|$, and $|V_i|\leq \gamma N$, $|V_i\cap B|\leq \frac{|V_i|}{2\gamma}\leq \frac N 2$. Therefore, at least $N/2$ edges of $M(e)$ have one endpoint in $A'$; the other endpoint of each such edge must lie in $A''\cup B$. We conclude that $|E_G(A',A''\cup B)|\geq \frac{N\cdot |E_H(S,V(H)\setminus S)|}2\geq \frac{\psi'N|S|}2$. Since every graph $G_i$ contains between $N$ and $\gamma N$ vertices, we get that $|A'|\leq \gamma N|S|$, and so $|E_G(A',A''\cup B)|\geq \frac{\psi'|A'|}{2\gamma}$.
	Since maximum vertex degree in $H$ is $\Delta$, every vertex in $A''$ may be an endpoint of at most $\Delta$ such edges.
	
	We now consider two cases. First, if $|A''|\leq \psi'|A'|/(4\Delta\gamma)$, then $|E_G(A',A'')|\leq \Delta|A''|\leq  \psi'|A'|/(4\gamma)$. Therefore, $|E_G(A,B)|\geq |E_G(A',B)|\geq \psi'|A'|/(4\gamma)\geq \psi' |A|/(8\gamma)\geq \psi''|A|$.
	
	Lastly, assume that $|A''|> \psi'|A'|/(4\Delta\gamma)$, so $|A''|\geq \psi'|A|/(8\Delta\gamma)$.
	Consider any graph $G_i$ that is not a type-1 graph, so $|V_i\cap A|\leq \left (1-\frac{1}{2\gamma}\right)|V_i|$. If $|V_i\cap A|\leq |V_i|/2$, then there are at least $\psi|V_i\cap A|$ edges of $G_i$ in $E_G(A,B)$. Otherwise, there are at least $\psi|V_i\cap B|$ edges of $G_i$ in $E_G(A,B)$. Since $|V_i\cap B|\geq |V_i|/2\gamma\geq |V_i\cap A|/(2\gamma)$, the number of edges that $G_i$ contributes to  $E_G(A,B)$ is at least $\psi |V_i\cap B|\geq \psi |V_i\cap A|/\gamma$. We conclude that $|E_G(A,B)|\geq \psi|A''|/(2\gamma) \geq \psi \psi'|A|/(16\Delta\gamma^2)\geq \psi''|A|$.	
\end{proof}

\subsubsection{Proof Overview}
We now provide an overview of the proof of \ref{thm: r>1 iterations}, and set up some notation.

In order to simplify the notation, we denote the input graph by $G=(V,E)$, and we denote $|V|=n$ and $|E|=m$; recall that $|E|=O(n\log n)$. Let 
$\tN'=N^{q-1}/2$, and let $\tN=\floor{n/\tN'}$, so $\tN\leq 2N$. Since $N$ is an integral power of $2$, $\tN'$ is an even integer. Moreover, from our assumption that $n\geq N^{q-1}/2$, we get that $\tN\geq 1$.

We partition the set $V$ of vertices into $\tN+1$ subsets $V_1,\ldots, V_{\tN},V_{\tN+1}$, where sets $V_1,\ldots,V_{\tN}$ have cardinality exactly $\tN'$ each, and the last set, that we denote by $Z=V_{\tN+1}$ has cardinality less than $\tN'$. We call the vertices in $Z$ the \emph{extra vertices}.

The algorithm consists of three steps. In the first step, we construct expanders $H_1,\ldots,H_{\tN}$, where for all $1\leq i\leq \tN$, $V(H_i)=V_i$, that we attempt to embed into $G$. We will either succeed in embedding these expanders with a small congestion and a relatively small number of fake edges, or we will compute the desired cut $(X,Y)$ in $G$. In the second step, we construct an expander $H'$ whose vertex set is $v_1,\ldots,v_{\tN}$, where we think of vertex $v_i$ as representing the set $V_i$ of vertices of $G$. We will attempt to embed graph $H'$ into $G$, with a small number of fake edges and low congestion, where every edge $e=(v_i,v_j)$ of $H'$ is embedded into $\tN'$ paths connecting vertices of $V_i$ to vertices of $V_j$ in $G$. If our algorithm fails to find such an embedding, then we will again produce the desired cut  $(X,Y)$ in $G$. If, over the course of the first two steps, the algorithm does not terminate with a cut  $(X,Y)$ in $G$, then we consider an expander $H^*$, obtained by computing 
a $\tilde N'$-composition of $H_1,\ldots,H_N$ and of $H'$, and then adding the vertices of $Z$, together with a matching connecting every vertex of $Z$ to some vertex of $V_1\cup\cdots\cup V_{\tN}$ to the resulting graph. The algorithm from the first two steps has then computed an embedding of $H^*$ into $G$, with a relatively small number of fake edges. In our last step, we compute a large subset $S$ of vertices of $G$ such that $G[S]$ is a $\psi_q$-expander, using Algorithm \extractexpander from \ref{lem: embedding expander w fake edges gives expander}. We now proceed to describe each of the three steps in turn.
Throughout the algorithm, we use a parameter $z=\frac{\psi_{q-1}n}{c\log^8n}$, where $c$ is a large enough constant, whose value will be set later.

\subsubsection{Step 1: Embedding Many Small Expanders}
The goal of this step is to construct a collection $\hset=\set{H_1,\ldots,H_{\tN}}$ of expanders, where for $1\leq i\leq \tN$, $V(H_i)=V_i$, and to compute an low-congestion embedding of all these expanders into $G+F$, where $F$ is a small set of fake edges for $G$. In other words, if we let $H$ be the graph obtained by taking a disjoint union of the graphs $H_1,\ldots,H_{\tN}$, and the set $Z$ of isolated vertices, then we will attempt to compute an embedding of $H$ into $G$. We will either find such an embedding, that uses relatively few fake edges, or we will return a cut $(X,Y)$ of $G$ with the required properties. We summarize this step in the following lemma.

\begin{lem}\label{lem: step 1}
	There is a deterministic algorithm that either computes a cut $(X,Y)$ in $G$ with $|X|,|Y|\geq  \Omega\left (\frac{\psi_{q-1}\cdot n}{\log^8n}\right )$ and $\Psi_G(X,Y)\leq 1/100$; or it constructs a collection $\hset=\set{H_1,\ldots,H_{\tN}}$ of $\hat \psi$-expanders, where  for $1\leq i\leq N$, $V(H_i)=V_i$, and $\hat \psi=\psi_{q-1}/2$, together with a set $F$ of at most $O(z\log n)$ fake edges, and an embedding $\pset$ of the graph $H=\left (\bigcup_iH_i\right )\cup Z$ into $G+F$, with congestion $O(\log^5n)$, such that every vertex of $G$ is incident to at most $O(\log n)$ edges of $F$. The running time of the algorithm is $ O\left (N^{q+1} \cdot \poly\log n\right )+O\left (N\cdot \log n\right )\cdot T(q-1)$.
	\end{lem}

\begin{proof}

The construction of the graphs $H_1,\ldots,H_{\tN}$, and of their embedding into $G$ is done gradually, by running $\tN$ instances of the cut-matching game, in parallel. Initially, for each $1\leq i\leq \tN$, we let the graph $H_i$ contain the set $V_i$ of vertices and no edges. Throughout the algorithm, we denote by $\hset=\set{H_1,\ldots,H_{\tN}}$ the current collection of the expanders we are constructing. We partition $\hset$ into two subsets: set $\hset'$ of \emph{active} graphs, and set $\hset''$ of \emph{inactive} graphs. Initially, every graph $H_i$ is active, so $\hset'=\hset$ and $\hset''=\emptyset$. Throughout the algorithm, for every inactive graph $H_i$, we will maintain a subset $S_i\subseteq V_i$ of at least $\tN'/2$ vertices, such that graph $H_i[S_i]$ is a $\psi_{q-1}$-expander. Throughout the algorithm, we also let $H$ denote the graph obtained by taking the disjoint union of all graphs in $\hset$ with a set $Z$ of isolated vertices. We will maintain an embedding $\pset$ of $H$ into $G$ throughout the algorithm. We will ensure that, throughout the algorithm, for all $1\leq i\leq \tN$, the maximum vertex degree in each graph $H_i$ is at most $\cCMG\log \tN'$.

At the beginning of the algorithm, for each $1\leq i\leq \tN$, graph $H_i$ contains the set $V_i$ of vertices and no edges, so graph $H$ consists of the set $V$ of vertices and no edges. The initial embedding is $\pset=\emptyset$, and every graph $H_i$ is active. 

As long as $\hset'\neq \emptyset$, we perform iterations, where the $j$th iteration is executed as follows. We apply algorithm $\aset(q-1)$ to every graph $H_i\in \hset'$ separately. Observe that each such graph contains $\tN'\leq N^{q-1}$ vertices, and has maximum vertex degree at most $\cCMG\log \tN'$. For each such graph $H_i$, if the outcome is a subset $S_i\subseteq V_i$ of vertices, such that $|S_i|\geq \tN'/2$ and $H_i[S_i]$ is a $\psi_{q-1}$-expander, then we add $H_i$ to the set $\hset''$ of inactive graphs, and store the set $S_i$ of vertices with it. Let $\hat \hset\subseteq \hset'$ be the collection of all remaining active graphs, so for each graph $H_i\in \hat \hset$, the algorithm has computed a cut $(A_i,B_i)$ with $|A_i|,|B_i|\geq \tN'/4$, and $|E_{H_i}(A_i,B_i)|\leq \tN'/100$. We assume without loss of generality that $|A_i|\leq |B_i|$. Let $(A_i',B_i')$ be any partition of $V_i$ with $|A_i'|=|B_i'|$, such that $A_i\subseteq A_i'$. We treat the partition $(A_i',B_i')$ as the move of the cut player in the cut-matching game corresponding to the graph $H_i$.

For convenience, we assume w.l.o.g. that $\hat \hset=\set{H_1,\ldots,H_k}$.
In order to implement the response of the matching player,
we apply Algorithm \routeorcut from \ref{cor: matching player} to graph $G$, the sets $A'_1,B'_1,\ldots,A'_k,B'_k$ of vertices, and parameters $\psi=1/100$ and $z$ (recall that we have defined $z=\frac{\psi_{q-1}n}{c\log^8n}$ for some large enough constant $c$). 
We now consider two cases. If Algorithm \routeorcut from \ref{cor: matching player} returns a cut $(X,Y)$ of $G$ with $|X|,|Y|\geq z/2$ and $\Psi_G(X,Y)\leq \psi$, then we say that the current iteration terminates with a cut. In this case, we terminate the algorithm, and return $(X,Y)$ as its outcome; it is immediate to verify that this cut has the required properties. In the second case, we obtain a partial routing $(M'=\bigcup_{i=1}^kM'_i,\pset')$ of the sets $A'_1,B'_1,\ldots,A'_k,B'_k$ of vertices, where $|M'|\geq k\tN'/2-z$ (recall that for all $i$, $|A'_i|=|V_i|/2=\tN'/2$). 
The congestion of the embbedding is at most $O(\Delta^2\log^2n/\psi^2)=O(\log^4n)$.
We then say that the current iteration has terminated with a routing. 

Consider now some index $1\leq i\leq k$, and let $A''_i\subseteq A'_i$ and $B''_i\subseteq B'_i$ be the subsets of vertices that do not participate in the matching $M'_i$. Let $M''_i$ be an arbitrary perfect matching between the vertices of $A''_i$ and the vertices of $B''_i$, and let $F_i$ be a set of fake edges $F_i=\set{(u,v)\mid (u,v)\in M''_i}$. For every pair $(u,v)\in M''_i$, we embed the pair $(u,v)$ into the corresponding fake edge $(u,v)\in F_i$. Let $M^j_i=M'_i\cup M''_i$. We add the edges of $M^j_i$ to graph $H_i$. 

Denote $M^j=\bigcup_{i=1}^kM^j_i$, and let $F^j=\bigcup_{i=1}^kF_i$ be the resulting set of fake edges; recall that $|F^j|\leq z$. Let
 $\pset^j$ be the embedding of all edges in $M^j$ that is obtained from the partial routing $\pset'$, by adding the embeddings of all fake edges to it. 
Observe that we have now obtained an embedding $\pset^j$ of all edges of $M^j$ into $G+F^j$, with congestion $O(\log^4n)$. We add the paths of $\pset^j$ to the embedding $\pset$ of the current graph $H$, and continue to the next iteration.

Our algorithm can therefore be viewed as running $\tN$ parallel copies of the cut-matching game. From \ref{thm:KKOV-new}, the number of iterations is bounded by $\cCMG\log \tN'$, and so for every graph $H_i$, its maximum vertex degree is always bounded by $\cCMG\log \tN'$. The algorithm terminates once all graphs $H_i$ become inactive. Recall that for each such graph $H_i$, we are given a subset $S_i$ of its vertices, such that $|S_i|\geq |V_i|/2$, and $H_i[S_i]$ is a $\psi_{q-1}$-expander. We perform one last iteration, whose goal is to turn each graph $H_i$ into an expander, by adding a new set of edges to it, while simultaneously embedding these edges into the graph $G$ together with a small number of fake edges, or find a cut $(X,Y)$ as required. Let $r-1$ denote the index of the last iteration before every graph $H_i$ becomes inactive.

\paragraph{Last Iteration.}
For each $1\leq i\leq \tN$, we let $B_i=S_i$ and $A_i=V_i\setminus S_i$, so that $|A_i|\leq |B_i|$ holds. 
We apply Algorithm \routeorcut from \ref{cor: matching player} to graph $G$, the sets $A_1,B_1,\ldots,A_{\tN},B_{\tN}$ of vertices, and parameters $\psi=1/100$ and $z$. The remainder of the iteration is executed exactly as before. If Algorithm \routeorcut from \ref{cor: matching player} returns a cut $(X,Y)$ of $G$ with $|X|,|Y|\geq z/2$ and $\Psi_G(X,Y)\leq \psi$, then we terminate the algorithm and return this cut. Otherwise, we obtain a partial routing $(M'=\bigcup_{i=1}^{\tN}M'_i,\pset')$ of the sets $A_1,B_1,\ldots,A_{\tN},B_{\tN}$ of vertices, where $|M'|\geq \sum_{i=1}^{\tN}|A_i|-z$, whose congestion is at most $O(\log^4n)$ as before.

Consider some index $1\leq i\leq \tN$, and let $A'_i\subseteq A_i$ and $B'_i\subseteq B_i$ be the subsets of vertices that do not participate in the matching $M'_i$. Let $M''_i$ be an arbitrary matching, in which every vertex of $A'_i$ is matched to some vertex of $B'_i$, and let $F_i$ be a set of fake edges corresponding to this matching $M''_i$, defined as before. For every pair $e=(u,v)\in M''_i$, we embed the fake edge $e$ into the path $P(e)=(e)$. Let $M^r_i=M'_i\cup M''_i$. We add the edges of $M^r_i$  to graph $H_i$. 

Denote $M^r=\bigcup_{i=1}^{\tN}M^r_i$, and let $F^r=\bigcup_{i=1}^{\tN}F_i$ be the resulting set of fake edges; as before $|F^r|\leq z$. Let
$\pset^r$ be the embedding of all edges in $M^r$ that is obtained from the partial routing $\pset'$, by adding the embeddings of all fake edges to it. As before, we have obtained an embedding $\pset^r$ of all edges of $M^r$ into $G+F^r$, with congestion $O(\log^4n)$. We add the paths of $\pset^r$ to the embedding $\pset$ of the current graph $H$.
Note that, from \ref{obs: exp plus matching is exp} we are
now guaranteed that every graph $H_i\in \hset$ is a $\psi_{q-1}/2$-expander. 

Recall that the congestion incurred by each path set $\pset^j$ of edges is $O(\log^4n)$, and, since the number of iterations is $O(\log n)$, the embedding $\pset$ causes congestion $O(\log^5n)$. The total number of fake edges in $F=\bigcup_{j=1}^rF^j$ is $O(z\log n)$.
Since each set $F^j$ of fake edges is a matching, every vertex of $G$ is incident to $O(\log n)$ fake edges.

We now analyze the running time of the algorithm. As observed before, the algorithm has $O(\log n)$ iterations. In every iteration, we apply algorithm $\aset(q-1)$ 
to $\tN=O(N)$ graphs. Additionally, we use Algorithm \routeorcut from \ref{cor: matching player}, whose running time is $\tilde O(k|E(G)|/\psi^3+kn/\psi^2)=\tilde O(k n)$, where $k$ is the number of vertex subsets. Since $k\leq |\hset|=\tN\leq N$, this running time is bounded by $\tilde O(Nn)=\tilde O(N^{q+1})$.
Therefore, the total running time of the algorithm is $\tilde O\left(N^{q+1}\right )+O\left(N\log n\right )\cdot T(q-1)$.

\end{proof}

%---------------------------------
%---------------------------------
%---------------------------------
\subsubsection{Step 2: Embedding One Large Expander}
%---------------------------------
%---------------------------------
%---------------------------------

We use Algorithm \constructexpander from  \ref{thm:explicit expander}, in order to construct, in time $O(\tilde N)$, a $\psi^*$-expander $H'=H_{\tN}$ on $\tN$ vertices, with $\psi^*=\Psi(H')=\Omega(1)$, such that maximum vertex degree in $H'$ is at most $9$. 
For convenience, we denote $V(H')=\set{v_1,\ldots,v_{\tN}}$. The main part of this step is summarized in the following lemma.

\begin{lem}\label{lem: step 2}
	There is a deterministic algorithm, that either computes a cut $(X,Y)$ in $G$ with $|X|,|Y|\geq  \Omega\left (\frac{\psi_{q-1}\cdot n}{\log^8n}\right )$ and $\Psi_G(X,Y)\leq 1/100$; or it computes  a collection $F'$ of at most $17z$ fake edges in $G$, and, for every edge $e=(v_i,v_j)\in E(H')$ a set $\pset(e)$ of $\tN'$ paths in $G+F'$, such that every path in $\pset(e)$ connects a vertex of $V_i$ to a vertex of $V_j$, and the endpoints of the paths in $\pset(e)$ are disjoint. Moreover, every vertex of $G$ is incident to at most $17$ fake edges in $F'$, and every edge of $G\cup F'$ participates in at most $O(\log^4n)$ paths in $\bigcup_{e\in E(H')}\pset(e)$.	 
	 The running time of the algorithm is $\tilde O\left (N^{q+1}\right )$.
\end{lem}

\begin{proof}
 Using a standard greedy algorithm, and the fact that the maximum vertex degree in $H'$ is at most $9$, we can partition the set $E(H')$ of edges into $17$ matchings, $M_1,\ldots,M_{17}$. We then perform up to $17$ iterations; in each iteration $i$, we will either compute a small set $F^i$ of fake edges for $G$, and the sets $\pset(e)$ of paths for all edges $e\in M_i$, in graph $G+F^i$, or we will compute the cut $(X,Y)$ in $G$ with the required properties.
 
In order to execute the $i$th iteration, we consider the set $M_i$ of edges of $H'$, and denote, for convenience, $M_i=\set{e_1,\ldots,e_{k_i}}$. For each $1\leq j\leq k_i$, if $e_j=(v_z,v_{z'})$, then we define $A_j=V_z$ and $B_j=V_{z'}$. Observe that $|A_i|=|B_j|=\tN'$, and the resulting vertex sets $A_1,B_1,\ldots,A_{k_i},B_{k_i}$ are all disjoint.

We apply Algorithm \routeorcut from \ref{cor: matching player} to graph $G$, the sets $A_1,B_1,\ldots,A_{k_i},B_{k_i}$ of vertices, and parameters $\psi=1/100$ and $z$ (as defined before, $z=\frac{\psi_{q-1}n}{c\log^8n}$). If Algorithm \routeorcut returns a cut $(X,Y)$ of $G$ with $|X|,|Y|\geq z/2$ and $\Psi_G(X,Y)\leq \psi$, then we terminate the algorithm and return this cut; it is easy to verify that cut $(X,Y)$ has all required properties. In this case we say that the iteration terminates with a cut. Otherwise, we obtain a partial routing $(\hat M_i=\bigcup_{e\in M_i}\hat M(e),\hat \pset_i)$ of the sets $A_1,B_1,\ldots,A_{k_i},B_{k_i}$ of vertices, where $|\hat M_i|\geq \sum_{j=1}^{k_i}|A_j|-z$, whose congestion is at most $O(\Delta^2\log^2n/\psi^2)=O(\log^4n)$. In this case we say that the iteration terminates with a routing. Consider now some edge $e_j\in M_i$. Let $A'_j\subseteq A_j$, $B'_j\subseteq B_j$ be the subsets of vertices that do not participate in the matching $\hat M(e_j)$. Let $\hat M'(e_j)$ be an arbitrary perfect matching between $A'_j$ and $B'_j$, and let $F^i_j$ be the corresponding set of fake edges for graph $G$ (so for every edge $e\in \hat M'(e_j)$, we add an edge with the same endpoints to $F^i_j$). Finally, set $\hat M''(e_j)=\hat M_j\cup \hat M'_j$. Let $F^i=\bigcup_{j=1}^{k_i}F^i_j$; recall that $|F^i|\leq z$. Let $\pset'_i$ be the set of paths routing the edges of $F^i$, where for each edge $e\in F^i$, the corresponding path $P(e)\in \pset'_i$ consists of the edge $e$. Lastly, let $\hat \pset_i''=\hat \pset_i\cup \hat \pset_i'$. Note that $\hat \pset_i''$ is the routing of all edges in $\hat M''_i\cup F^i$ in graph $G+F^i$, that causes edge-congestion at most $O(\log^4n)$.

If any iteration of the algorithm terminated with a cut, then we terminate the algorithm and return the corresponding cut. We assume from now on that every iteration of the algorithm terminated with a routing. Setting $F'=\bigcup_{i=1}^{17} F^i$, we obtain the desired routing of the edges of $H'$ in graph $G+F'$, with congestion $O(\log^4n)$. Since, for every $1\leq i\leq 17$, the edges of $F^i$ form a matching, every vertex of $G$ is incident to at most $17$ such edges.

Recall that the running time of Algorithm \routeorcut from \ref{cor: matching player} is $\tilde O(k|E(G)|/\psi^3+kn/\psi^2)=\tilde O(kn)=\tilde O(kN^q)$, where $k$ is the number of pairs of sets that we need to route. Since $k\leq |V(H')|\leq \tN\leq O(N)$, and the number of iterations is at most $17$, we get that the running time of the algorithm is $\tilde O\left (N^{q+1}\right )$.
\end{proof}

Finally, we need the following claim, in order to connect the set $Z$ of extra vertices to the remaining vertices of $G$.

\begin{claim}\label{claim: step 2 extra}
	There is a deterministic algorithm, that either computes a cut $(X,Y)$ in $G$ with $|X|,|Y|\geq  \Omega\left (\frac{\psi_{q-1}\cdot n}{\log^8n}\right )$ and $\Psi_G(X,Y)\leq 1/100$; or it computes  a matching $M$ connecting every vertex of $Z$ to a distinct vertex of $V(G)\setminus Z$, a collection $F''$ of at most $z$ fake edges in $G$, and a set $\pset''=\set{P(e)\mid e\in M}$ of paths in $G+F''$, such that, for each edge $e=(u,v)\in M$, path $P(e)$ connects $u$ to $v$. Moreover, every vertex of $G$ is incident to at most one fake edge in $F''$, and every edge of $G\cup F''$ participates in at most $O(\log^4n)$ paths in $\pset''$.
	The running time of the algorithm is $\tilde O\left (N^{q}\right )$.
\end{claim}

\begin{proof}
	We apply Algorithm \routeorcut from \ref{cor: matching player} to graph $G$, the sets $A_1=Z$, $B_1=V(G)\setminus Z$ of vertices, parameter $\psi=1/100$, and parameter $z$. If the outcome of Algorithm \routeorcut is a cut $(X,Y)$ of $G$ with $|X|,|Y|\geq z/2$ and $\Psi_G(X,Y)\leq \psi$, then we return this cut; it is immediate to verify that cut  $(X,Y)$ has the required properties. Otherwise, we obtain a  routing $(M',\pset')$ of the sets $A_1,B_1$, with $|M'|\geq |Z|-z$. 
	The congestion of the routing is at most $O(\Delta^2\log^2n/\psi^2)=O(\log^4n)$. We let $Z'\subseteq Z$ be the set of all vertices of $Z$ that do not participate in the matching $M'$, and we let $M''$ be an arbitrary matching that matches every vertex of $Z'$ to a distinct vertex of $V(G)\setminus Z$, such that $M=M'\cup M''$ is a matching; such a set $M''$ exists since $Z$ contains at most half the vertices of $G$. We let $F''$ be a set of fake edges for $G$ corresponding to the edges of $M''$, so every edge $e=(u,v)\in M''$ is also added to $F''$. We let $P(e)$ be the path that only consists of the edge $e$, and we treat $P(e)$ as the embedding of $e$.
	
	We now obtained a set $F''$ of at most $z$ fake edges, and every vertex of $G$ is incident to at most one such fake edge. We also obtained an embedding $\pset''=\pset'\cup\set{P(e)\mid e\in F''}$ of $M$ into $G$ with congestion $O(\log^4n)$. The running time of Algorithm \routeorcut is $\tilde O(\Delta^3|E(G)|/\psi^3+n/\psi^2)=\tilde O(n)=\tilde O(N^q)$.	
\end{proof}

If the algorithm from \ref{lem: step 2} or the algorithm from \ref{claim: step 2 extra} produce a cut $(X,Y)$ in $G$ with  $|X|,|Y|\geq  \Omega\left (\frac{\psi_{q-1}\cdot n}{\log^8n}\right )$ and $\Psi_G(X,Y)\leq 1/100$, then we terminate the algorithm and return this cut. Otherwise, consider the following graph $H^*$: we start by letting $H^*$ be a disjoint union of the graphs $H_1,\ldots,H_{\tN}$ constructed in the first step. Additionally, for every edge $e=(v_i,v_j)\in E(H')$, for every path $P\in \pset(e)$, whose endpoints are $x\in V_i$, $y\in V_j$, we add the edge $(x,y)$ to $E(H^*)$. It is immediate to verify that graph $H^*$ is an $N'$-composition of $H_1,\ldots,H_{\tN}$, and graph $H'$. Recall that for all $1\leq i\leq N$, graph $H_i$ is a $ \psi_{q-1}/2$-expander, while graph $H'$ is a $\psi^*$-expander, for some $\psi^*=\Omega(1)$. The maximum vertex degree in $H'$ is bounded by $9$. Therefore, from \ref{thm: expander composition}, graph $H^*$ is a $\psi'$-expander, for $\psi'=\psi_{q-1}\psi^*/O(\log n)=\Omega(\psi_{q-1}/\log n)$. Note that the maximum vertex degree in $H^*$ is $O(\log n)$. Lastly, we add to graph $H^*$ the set $Z$ of extra vertices as isolated vertices, and the matching $M$ that was computed in \ref{claim: step 2 extra}. 
Recall that, from \ref{obs: exp plus matching is exp}, graph $H^*$ is a $\psi'/2$-expander, where $\psi'=\Omega(\psi_{q-1}/\log n)$; to simplify the notation, we say that $H^*$ is a $\psi'$-expander, adjusting the value of $\psi'$ accordingly.
Let $F^*=F\cup F'\cup F''$ be the union of the sets of fake edges computed by the algorithms from \ref{lem: step 1}, \ref{lem: step 2}, and \ref{claim: step 2 extra}. Recall that $|F^*|=O(z\log n)$, where $z=\frac{\psi_{q-1}n}{c\log^8n}$ for some large enough constant $c$.

 We denote by $\Delta_G$ the maximum vertex degree of $G+F^*$. Since the set $F^*$ of fake edges consists of $O(\log n)$ matchings, $\Delta_G=O(\log n)$.

By combining the outcomes of the algorithms from \ref{lem: step 1}, \ref{lem: step 2}, and \ref{claim: step 2 extra},  we obtain an embedding of $H^*$ into $G+F^*$ with congestion at most $O(\log^5n)$. The maximum vertex degree in $H^*$, that we denote by $\Delta_{H^*}$, is $O(\log n)$. The maximum vertex degree in $G+F^*$, that we denote by $\Delta_G$, is $O(\log n)$. Note that the running time of the algorithm so far is 
$ O\left (N^{q+1} \cdot \poly\log n\right )+O\left (N\cdot \log n\right )\cdot T(q-1)$.

\subsubsection{Step 3: Obtaining the Final Expander}
In this step, we apply Algorithm \extractexpander from \ref{lem: embedding expander w fake edges gives expander} to graphs $G$ and $H^*$, the set $F^*$ of fake edges, and the embedding of $H^*$ into $G+F^*$ with congestion at most $\cong=O(\log^5n)$. 
We need first to verify that $|F^*|\leq \frac{\psi' n}{32\Delta_G\cong}$. Recall that  $\psi'=\Omega(\psi_{q-1}/\log n)$, $\Delta_G=O(\log n)$, and $\cong=O(\log^5n)$. Therefore, $\frac{\psi' n}{32\Delta_G\cong}\geq \Omega\left (\frac{\psi_{q-1} n}{\log^7n}\right )$, while $|F^*|\leq O(\log n)\cdot \frac{\psi_{q-1}n}{c\log^8n}$. Setting the constant $c$ to be large enough, we can ensure that the inequality indeed holds.

Recall that Algorithm \extractexpander from  \ref{lem: embedding expander w fake edges gives expander} computes a subgraph $G'\subseteq G$, that is a $\psi''$-expander, for $\psi''\geq \frac{\psi'}{6\Delta_G\cdot\cong}=\Omega\left (\frac{\psi_{q-1}}{\log^7n}\right)$, as $\psi'=\Omega(\psi_{q-1}/\log n)$. Recall also that $\psi_{q-1}=1/((q-1)\log N)^{8(q-1)}$, and $n\leq N^q$. Therefore:

\[\psi''\geq \Omega\left (\frac{1}{((q-1)\log N)^{8(q-1)} \cdot (q\log N)^7 }\right )\geq \frac{1}{(q\log N)^{8q}}=\psi_q.
\]

%assuming that $q>q_0$ for some large enough constant $q_0$.

%Assuming that \fbox{$q^7<\log N/c'$}, for some large enough constant $c'$, 
%We get that $\psi''\geq 1/(q^8q\log^{8q}N)=\psi_q$, as required.

Note that the number of vertices in $G'$ is at least: $n-\frac{4|F^*|\cong}{\psi'}$. Since
$ \frac{4|F^*|\cong}{\psi'}\leq O\left(\frac{z\log^7n}{\psi_{q-1}}\right )$ and $z=\frac{\psi_{q-1}n}{c\log^8n}$, letting $c$ be a large enough constant, we can ensure that $|V(G')|\geq 2n/3$, as required.

The running time of Algorithm \extractexpander is $\tilde O(|E(G)|\Delta_G\cdot\cong/\psi')=\tilde O(n/\psi_{q-1})=\tilde O(N^q\cdot (q\log N)^{8q})$.

By combining all three steps together, we obtain total running time:
 $ O\left (N^{q+1} \cdot (q\log N)^{8q+O(1)} \right )+O\left (N\cdot \log n\right )\cdot T(q-1)$, as required.

\section{A Slower Algorithm for \BCut}

\label{sec:BCut_high_cond}

In this section we prove the following:

\begin{theorem}\label{thm: main slower alg}
	There is a universal constant $c$, and a deterministic algorithm, that, given an $n$-vertex $m$-edge graph $G=(V,E)$, a parameter $0<\phi<1$, and another parameter $r\leq c\log m$, returns a cut $(A,B)$ in $G$ with $|E_G(A,B)|\leq \phi\cdot \vol(G)$, such that:
	\begin{itemize}
		\item either $\vol_G(A),\vol_G(B)\geq \vol(G)/3$; or
		\item $\vol_G(A)\geq \frac{7}{12}\cdot \vol(G)$, and the graph $G[A]$ has conductance $\phi'\geq \phi/\log^{O(r)}m$.
	\end{itemize}
	
	The running time of the algorithm is $ O\left(m^{1+O(1/r)}\cdot (\log m)^{O(r^2)}/\phi^2\right )$. 
\end{theorem}

%We note that by setting the parameter $r=\Theta\left((\log n/\log\log n)^{2/3}\right )$, we obtain an $\alpha$-approximate algorithm for $\BCut$, for $\alpha=2^{O((\log m)^{1/3}(\log\log m)^{2/3})}$, and running time $O(m^{1+o(1)}/\phi^2)$. By letting $r$ be a large enough constant, we can obtain a $(\log m)^{O(1/\epsilon)}$- approximation in time $O(m^{1+\epsilon}/\phi^2)$, for any constant $0<\epsilon<1$.

From the definition of the \BCut problem from \Cref{def:intro:BCut}, this implies a slower version of \ref{thm:intro:main} when the conductance parameter $\phi$ is low:

\begin{cor}\label{thm:BCut phi}
	There is a deterministic algorithm, that, given a graph $G$ with $m$ edges,
	and parameters $\phi\in (0,1]$, $1\leq r\leq O(\log m)$, and $\alpha =(\log m)^{O(r)}$,
	computes an $\alpha$-approximate solution to instance $(G,\phi)$ of $\BCut$ in time $O\left ( m^{1+O(1/r)}\cdot (\log m)^{O(r^2)} / \phi^2 \right )$.%\yg{removed $m^{o(1)}$}
\end{cor}	

While the above algorithm can significantly slower than the one from \ref{thm:intro:main} when the conductance parameter $\phi$ is low, many of our applications only need to solve the $\BCut$ problem for relatively high values of $\phi$, and so the algorithm from \ref{thm: main slower alg} is sufficient for them. In particular, we will use this algorithm in order to obtain fast deterministic approximation algorithms for max $s$-$t$ flow, which will then in turn be used in order to obtain the full proof of \ref{thm:intro:main}. The remainder of this section is dedicated to the proof of \ref{thm: main slower alg}.

%The proof consists of two ingredients. The first ingredient is the following extension of \ref{thm: cut player} to higher sparsity regime.

Two key ingredients in the proof are an extension of \ref{thm: cut player} to higher sparsity regime, and a degree reduction procedure, that are discussed in the next two subsections, respectively.

\subsection{Extension of \ref{thm: cut player} to Smaller Sparsity} 

In this subsection we prove the following lemma.

\begin{lem}\label{lem: iterations-final-cut}
	There is a deterministic algorithm, that, given an $n$-vertex graph $G=(V,E)$, with maximum vertex degree $\Delta$, parameters $0<\psi<1$, $z\geq 0$ and $r\geq 1$, such that $n^{1/r}\geq N_0$ (where $N_0$ is the constant from \ref{thm: cut player}), returns one of the following:
	
	\begin{itemize}
		\item either a cut $(X,Y)$ in $G$ with $|X|,|Y|\geq z/\Delta$ and $\Psi_{G}(X,Y)\leq \psi$; or
		\item a graph $H$ with $V(H)=V(G)$, that is a $\psi_r(n)$-expander (for $\psi_r(n)=1/(\log n)^{O(r)}$), together with a set $F$ of at most $O(z\log n)$ fake edges for $G$, and an embedding of $H$ into $G+F$ with congestion at most $O(\Delta \log n/\psi)$, such that every vertex of $G$ is incident to at most $O(\log n)$ edges of $F$.
	\end{itemize}
	
	The running time of the algorithm is  $\tilde O\left(n^{1+O(1/r)}\cdot (\log n)^{O(r^2)}+n\Delta^2/\psi\right )$.
\end{lem}

\begin{proof} %{\ref{lem: iterations-final-cut}}
	If the number of vertices in graph $G$ is odd, then we add an additional new vertex $v_0$, and we connect it to an arbitrary vertex of $G$ with a fake edge. For simplicity, the new number of vertices is still denoted by $n$.
	
	Our algorithm runs the cut-matching game, as follows. We start with a graph $H$, whose vertex set is $V$, and whose edge set is empty, and then perform iterations. Throughout the algorithm, we will ensure that the maximum vertex degree in $H$ is $O(\log n)$.
	
	Iteration $i$ is executed as follows. We apply  Algorithm \cutorcert from \ref{thm: cut player} to graph $H$. We now consider two cases. In the first case, the outcome is a cut $(A_i,B_i)$ in $H$, with $|A_i|,|B_i|\geq n/4$ and $|E_H(A_i,B_i)|\leq n/100$. Let $(A_i',B_i')$ be any partition of $V$ with $A_i\subseteq A_i'$, $B_i\subseteq B_i'$, and $|A_i'|=|B_i'|$. We apply Algorithm  \routeorcutp from \ref{thm:push match} to graph $G$, with the vertex sets $A_i',B_i'$, and parameters $z$ and $\psi$. If the outcome is a cut $(X,Y)$ in $G$ with $|X|,|Y|\geq z/\Delta$ and $\Psi_G(X,Y)\leq \psi$, then we terminate the algorithm and return this cut as its outcome. Otherwise, we obtain a partial routing $(M_i,\pset_i)$ of the sets $A_i',B_i'$, of value at least $|A_i'|-z$, that causes congestion at most $4\Delta/\psi$. Let $A''_i\subseteq A_i'$, $B''_i\subseteq B_i'$ be subsets of vertices that do not participate in the matching $M_i$. Let $M'_i$ be an arbitrary perfect matching between $A''_i$ and $B''_i$, and let $F_i$ be a set of fake edges corresponding to the matching $M'_i$ (so every edge in the matching becomes a fake edge). For every edge $e\in F_i$, we also let $P(e)$ be a path consisting of only the fake edge $e$. Let $M''_i=M_i\cup M'_i$, and let $\pset'_i=\pset_i\cup \set{P(e)\mid e\in F_i}$. Then $M''_i$ is a perfect matching between $A'_i$ and $B'_i$, and $\pset'_i$ is a routing of this matching in $G\cup F_i$, with congestion at most $4\Delta/\psi$. We add the edges of $M''_i$ to $H$, and continue to the next iteration.
	
	Consider now the second case, where the outcome of Algorithm \cutorcert from  \ref{thm: cut player} is a subset $S\subseteq V$ of at least $n/2$ vertices, such that $\Psi(G[S])\geq \psi_r(n)$.
	Let $i^*$ be the index of the current iteration. We then let $B_{i^*}=S$ and $A_{i^*}=V\setminus S$; note that $|A_{i^*}|\leq |B_{i^*}|$ must hold. We again employ Algorithm \routeorcutp from \ref{thm:push match}, with the vertex sets $A_{i^*},B_{i^*}$, and parameters $z$ and $\psi$. If the outcome is a cut $(X,Y)$ in $G$ with $|X|,|Y|\geq z/\Delta$ and $\Psi_G(X,Y)\leq \psi$, then we terminate the algorithm and return this cut as its outcome. Otherwise, we obtain a partial routing $(M_{i^*},\pset_{i^*})$ of the sets $A_{i^*},B_{i^*}$, of value at least $|A_{i^*}|-z$, that causes congestion at most $4\Delta/\psi$. As before, we let $A'_{i^*}\subseteq A_{i^*}$, $B'_{i^*}\subseteq B_{i^*}$ be subsets of vertices that do not participate in the matching $M_{i^*}$. Let $M'_{i^*}$ be an arbitrary matching, that matches every vertex of $A'_{i^*}$ to some vertex of $B'_{i^*}$, and let $F_{i^*}$ be a set of fake edges corresponding to the matching $M'_{i^*}$. As before, for every edge $e\in F_{i^*}$, we let $P(e)$ be a path consisting of only the fake edge $e$. Let $M''_{i^*}=M_{i^*}\cup M'_{i^*}$, and let $\pset'_{i^*}=\pset_{i^*}\cup \set{P(e)\mid e\in F_{i^*}}$. Then $M''_{i^*}$ matches every vertex of $A_{i^*}$ to a distinct vertex of $B_{i^*}$, and $\pset'_{i^*}$ is a routing of this matching in $G\cup F_{i^*}$, with congestion at most $4\Delta/\psi$. We add the edges of $M''_{i^*}$ to $H$, and terminate the algorithm.
	
	Observe that, if the algorithm never terminates with a cut $(X,Y)$ with $|X|,|Y|\geq z/\Delta$ and $\Psi_{G'}(X,Y)\leq \psi$, then, from \ref{obs: exp plus matching is exp}, the final graph $H$ is a $\psi_r(n)/2$-expander. Moreover, if we let $F=\bigcup_{i=1}^{i^*}F_i$, together with an additional fake edge incident to $v_0$ if the initial number of vertices in $G$ was odd, and $\pset=\bigcup_{i=1}^{i^*}\pset'_i$, then $\pset$ is an embedding of $H$ into $G+F$. From \ref{thm:KKOV-new}, the number of iterations in the algorithm is bounded by $O(\log n)$. Since, for all $i$, edge set $F_i$ is a matching, every vertex of $G$ is incident to $O(\log n)$ edges of $F$. Since every set $F_i$ contains at most $z$ edges, $|F|=O(z\log n)$. Lastly, since every set $\pset'_i$ of paths causes congestion $O(\Delta/\psi)$, the paths in $\pset$ cause congestion $O(\Delta \log n/\psi)$. 
	It now remains to bound the running time of the algorithm.
	
	The algorithm performs $O(\log n)$ iterations. Each iteration requires running the Algorithm  \cutorcert from \ref{thm: cut player}, which takes time $O\left(n^{1+O(1/r)}\cdot (\log n)^{O(r^2)}\right )$, and Algorithm \routeorcutp from \ref{thm:push match}, that takes time $\tilde O\left (n\Delta^2/\psi\right )$. Therefore, the total running time of the algorithm is: $\tilde O\left(n^{1+O(1/r)}\cdot (\log n)^{O(r^2)}+n\Delta^2/\psi\right )$.
\end{proof}

%------------------------------------------
%------------------------------------------
%------------------------------------------
%------------------------------------------\\
%------------------------------------------
\subsection{Degree Reduction}
%------------------------------------------
%------------------------------------------
%------------------------------------------
%------------------------------------------
\label{subsec:constant degree}

Assume that we are given a graph $G=(V,E)$ with $|V|=n$ and $|E|=m$, that we view as an input to the $\BCut$ problem. In this subsection we show a deterministic algorithm, that we call \reducedegree, that has running time $O(m)$, and transforms $G$ into a bounded-degree graph $\hat G$. We also provide an algorithm that transforms any sparse balanced cut in a subgraph of $\hat G$ into a ``nice'' cut, that corresponds to a sparse balanced cut in a subgraph of $G$.

We first describe Algorithm \reducedegree for constructing the graph $\hat G$. For convenience, we denote $V=\set{v_1,\ldots,v_n}$. For every vertex $v_i\in V$, we let $\deg(v_i)$ denote the degree of $v_i$ in $G$, and we let $\set{e_1(v_i),\ldots,e_{\deg(v_i)}(v_i)}$ be the set of edges incident to $v$, indexed in an arbitrary order. For every vertex $v_i\in V$, we use Algorithm \constructexpander from
\ref{thm:explicit expander} to construct a graph $H_i$ on a set $V_i$ of $\deg(v_i)$ vertices, that is an $\alpha_0$-expander, for some constant $\alpha_0$, such that the maximum vertex degree in $H_i$ is at most $9$. Recall that the running time of the algorithm for constructing $H_i$ is $O(\deg(v_i))$. We denote the vertices of $H_i$ by $V_i=\set{u_1(v_i),\ldots,u_{\deg(v_i)}(v_i)}$.

In order to obtain the final graph $\hat G$, we start with a disjoint union of all graphs in $\set{H_i\mid v_i\in V}$. All edges lying in such graphs $H_i$ are called \emph{type-1 edges}. Additionally, we add to $\hat G$ a collection of type-2 edges, defined as follows. Consider any edge $e=(v,v')\in E$, and assume that $e=e_j(v)=e_{j'}(v')$ (that is, $e$ is the $j$th edge incident to $v$ and it is the $j'$th edge incident to $v'$). We then let $\hat e$ be the edge $(u_j(v),u_{j'}(v))$. For every edge $e\in E$, we add the corresponding new edge $\hat e$ to graph $\hat G$ as a type-2 edge. This concludes the construction of the graph $\hat G$, that we denote by $\hat G=(\hat V,\hat E)$. Note that the maximum vertex degree in $\hat G$ is at most $10$, and $|\hat V|=2m$. Moreover, the running time of the algorithm for constructing the graph $\hat G$ is $O(m)$.

We say that a subset $S\subseteq \hat V$ of vertices is \emph{canonical} iff for every vertex $v_i\in V$, either $V_i\subseteq S$, or $V_i\cap S=\emptyset$. Similarly, we say that a cut $(X,Y)$ in a subgraph of $\hat G$ is canonical iff each of $X,Y$ is a canonical subset of $\hat V$.
The following lemma allows us to convert an arbitrary sparse balanced cuts in a subgraph of $\hat G$ into a canonical one.

\begin{lem}\label{lem: degree reduction balanced cut case}
	Let $\alpha_0 >0$ be the constant from \Cref{thm:explicit expander}.
	There is a deterministic algorithm, that we call \makecanonical, that, given a subgraph $\hat G'\subseteq \hat G$, where $V(\hat G')$ is a canonical vertex set, and a cut $( A, B)$ in $\hat G'$, computes, in time $O(m)$, a canonical cut $(A',B')$ in $\hat G'$, such that $|A'|\geq |A|/2$, $|B'|\geq |B|/2$, and moreover, if $|E_{\hat G}(A,B)|\leq \psi \min\set{|A|,|B|}$, for $\psi\leq \alpha_0/2$, then  $|E_{\hat G}(A',B')|\leq O(|E_{\hat G}(A,B)|)$.
\end{lem}

\begin{proof}	
	We start with the cut $(\hat A,\hat B)=(A,B)$ in graph $\hat G'$ and then gradually modify it, by processing the vertices of $V(G)$ one-by-one. When a vertex $v_i$ is processed, if $V_i\cap V(\hat G')\neq \emptyset$, we move all vertices of $V_i$ to either $\hat A$ or $\hat B$. Once every vertex of $V(G)$ is processed, we obtain the final cut $(A',B')$, that will serve as the output of the algorithm.

	Consider an iteration when some vertex $v_i\in V(G)$ is processed, and assume that $V_i\subseteq V(\hat G')$. Denote $A_i=A\cap V_i$ and $B_i=B\cap V_i$. If $|A_i|\geq |B_i|$, then we move all vertices of $B_i$ to $\hat A$, and otherwise we move all vertices of $A_i$ to $\hat B$. Assume w.l.o.g. that the latter happened (the other case is symmetric). Note that the only new edges that are added to the cut $E_{\hat G}(\hat A,\hat B)$ are type-2 edges that are incident to the vertices of $A_i$. The number of such edges is bounded by $|A_i|$. The edges of $E_{H_i}(A_i,B_i)$ belonged to the cut $E_{\hat G}(\hat A,\hat B)$ before the current iteration, but they do not belong to the cut at the end of the iteration. Since $H_i$ is an $\alpha_0$-expander, we get that $|A_i|\leq |E_{H_i}(A_i,B_i)|/\alpha_0$. Therefore, the increase in $|E_{\hat G}(\hat A,\hat B)|$, due to the current iteration is bounded by $|E_{H_i}(A_i,B_i)|/\alpha_0$. We \emph{charge} the edges of $E_{H_i}(A_i,B_i)$ for this increase; note that these edges will never be charged again. The algorithm  terminates once all vertices of $V(G)$ are processed. 
	Let $(A',B')$ denote the final cut $(\hat A,\hat B)$.
	From the above discussion, we are guaranteed that $|E_{\hat G}(A', B')|\leq |E_{\hat G}(A,B)|+\sum_{v_i\in V(G)}|E_{H_i}(A_i,B_i)|/\alpha_0\leq O(|E_{\hat G}(A,B)|)$. 
	
	Next, we claim that $|A'|\geq |A|/2$ and that $|B'|\geq |B|/2$. We prove this for $|A'|$; the proof for $|B'|$ is symmetric. Indeed, assume otherwise. Let $V'\subseteq V$ be the set of all vertices $v_i$, such that, when the algorithm processed $v_i$, the vertices of $A_i$ were moved from $\hat A$ to $\hat B$, and let $n_i=|A_i|$. Then $\sum_{v_i\in V'}n_i> | A|/2$ must hold. Notice however that for a vertex $v_i\in V'$, $|E_{H_i}(A_i,B_i)|\geq \alpha_0|A_i|=\alpha_0 n_i$ must hold. Therefore, graph $H_i$ contributed at least $\alpha_0 n_i$ edges to the original cut $E_{\hat G}( A, B)$. Since we are guaranteed that $|E_{\hat G}( A, B)|\leq \psi \cdot | A|$, we get that $\sum_{v_i\in V'}\alpha_0 n_i\leq \psi \cdot | A|$, and so $\sum_{v_i\in V'}n_i\leq \psi \cdot | A|/\alpha_0\leq | A|/2$, since we have assumed that $\psi\leq \alpha_0/2$. But this contradicts the fact that we established before, that $\sum_{v\in V'}n_i> | A|/2$.	
%	
%	In order to obtain the final cut $(A',B')$ in $G$, we add every vertex $v\in V(G)$ to $A'$ if $V(H(v))\subseteq \hat A$, and we add it to $B'$ otherwise. It is immediate to verify that $\vol_G(A')=|\hat A|=\Omega(\beta |\hat V|)=\Omega(\beta \cdot \vol(G))$, and similarly $\vol_G(B')=\Omega(\beta \cdot \vol (G))$. Moreover, $|E_G(A',B')|=|E_{\hat G}(\hat A,\hat B)|=O(|E_{\hat G}(A,B)|)$. The running time of the algorithm is clearly bounded by $O(m)$.
\end{proof}

\subsection{Completing the Proof of \ref{thm: main slower alg}}
We prove the following theorem, from which \ref{thm: main slower alg} immediately follows.

\begin{theorem}\label{thm: slower alg player higher sparsity}
	There is a universal constant $N'_0$, and a deterministic algorithm, that, given an $n$-vertex $m$-edge graph $G=(V,E)$, a parameter $0<\phi<1$, and another parameter $r\geq 1$, such that $m^{1/r}\geq N'_0$, returns a cut $(A,B)$ in $G$ with $|E_G(A,B)|\leq \phi\cdot \vol(G)$, such that:
	\begin{itemize}
		\item either $\vol_G(A),\vol_G(B)\geq \vol(G)/3$; or
		\item $\vol_G(A)\geq \frac{7}{12}\cdot \vol(G)$, and the graph $G[A]$ has conductance $\phi'\geq \phi/\log^{O(r)}m$.
	\end{itemize}
	
	The running time of the algorithm is $ O\left(m^{1+O(1/r)}\cdot (\log m)^{O(r^2)}/\phi^2\right )$. 
\end{theorem}

In order to complete the proof of \ref{thm: main slower alg}, we let $c$ be a large enough constant, so that $m^{1/(c\log m)}\geq N_0'$ holds. %If the input parameter $r<r_0$, then we simply set $r=r_0$. 
We then apply the algorithm from \ref{thm: slower alg player higher sparsity} to the input graph $G$ and the  parameter $r$.
In the remainder of this section we focus on the proof of \ref{thm: slower alg player higher sparsity}.

\begin{proofof}{\ref{thm: slower alg player higher sparsity}}
	We denote by $\psi_r(n)=1/\log^{O(r)}n$ the parameter from \ref{thm: cut player} (that is, when Algorithm \cutorcert from \ref{thm: cut player} returns a set  $S$ of at least $n/2$ vertices, then $\Psi(G[S])\geq \psi_r(n)$ holds).
	Throughout the proof, we use two parameters: $\psi= \phi/\hat c$, and $z=\frac{\phi m}{\hat c(\log m)^{\hat c r}}$, where $\hat c$ is a large constant to be set later.
 We also set $N'_0=4N_0$, where $N_0$ is the universal constants from \ref{thm: cut player}.

We start by using Algorithm \reducedegree described in \ref{subsec:constant degree}, in order to construct, in time $O(m)$, a graph $\hG$ whose maximum vertex degree is bounded by $10$, and $|V(\hG)|=2m$. Denote $V(G)=\set{v_1,\ldots,v_n}$. Recall that graph $\hG$ is constructed from graph $G$ by replacing each vertex $v_i$ with an $\alpha_0$-expander $H_i$ on $\deg_G(v_i)$ vertices, where $\alpha_0=\Theta(1)$. For convenience, we denote the set of vertices of $H_i$ by $V_i$. Therefore, $V(\hG)$ is a union of the sets $V_1,\ldots,V_n$ of vertices. Consider now some subset $S$ of vertices of $\hG$. Recall that we say that $S$ is a \emph{canonical} vertex set iff for every $1\leq i\leq n$, either $V_i\subseteq S$ or $V_i\cap S=\emptyset$ holds. 

%Our algorithm will iteratively invoke Lemma \ref{lem: iterations-final-cut} on subgraphs of $\hat G$.

	\iffalse
	
	\begin{lem}\label{lem: iterations-final-cut}
		There is a deterministic algorithm, that, given a connected $n'$-vertex graph $G'=(V',E')$, with maximum vertex degree $\Delta$, a parameter $0<\psi<1$, and another parameter $r\geq r_0$, such that $(n')^{1/r}\geq N_0$ (where $N_0,r_0$ are the constants from \ref{thm: cut player}), returns one of the following:
		
		\begin{itemize}
			\item either a cut $(X,Y)$ in $G'$ with $|X|,|Y|\geq z/\Delta$ and $\Psi_{G'}(X,Y)\leq \psi$; or
			\item a graph $H$ with $V(H)=V(G')$, that is a $\psi_r(n')$-expander, together with a set $F$ of at most $O(z\log n)$ fake edges for $G$, and an embedding of $H$ into $G'+F$ with congestion at most $O(\Delta \log n/\psi)$, such that every vertex of $G'$ is incident to at most $O(\log n)$ edges of $F$.
		\end{itemize}
		
		The running time of the algorithm is  $\tilde O\left(n^{1+O(1/r)}\cdot (\log n)^{r^2}+n\Delta^2/\psi\right )$.
	\end{lem}
	
	\fi
		
	The algorithm performs a number of iterations. We maintain a subgraph $\hG'\subseteq \hG$; at the beginning of the algorithm, $\hG'=\hG$. In the $i$th iteration, we compute a canonical subset $S_i\subseteq V(\hG')$ of vertices, and then update the graph $\hG'$, by deleting the vertices of $S_i$ from it. The iterations are performed as long as $|\bigcup_iS_i|< |V(\hG)|/3$.
	
	In order to execute the $i$th iteration, we consider the current graph $\hG'$, denoting $|V(\hG')|=n'$. Note that, since we assume that $|\bigcup_{i'<i}S_{i'}|<|V(\hG)|/3$, we get that $n'\geq 2|V(\hG)|/3$. From our choice of parameter $N_0'$, we are guaranteed that $(n')^{1/r}\geq N_0$.
	We can now apply \ref{lem: iterations-final-cut} to graph $\hG'$, with the parameters $r$, $\psi$ and $z$. Recall that the maximum vertex degree in $\hG'$ is $\Delta\leq 10$. Assume first that the outcome is a cut $(X,Y)$ in $\hG'$ with $|X|,|Y|\geq z/\Delta\geq z/10$  and $\Psi_{\hG'}(X,Y)\leq \psi$. We say that the iteration terminates with a cut in this case.
%	Let $\beta=\min\set{|X|,|Y|}/|V(\hG')|$, so that $|X|,|Y|\geq \beta |V(\hG')|$. 
%	We are then guaranteed that $|E_{\hG'}(X,Y)|\leq \psi \beta |V(\hG')|$. 
By setting $\hat c$ to be a large enough constant, we can ensure that $\psi \leq \alpha_0/2$ where $\alpha_0$ is the constant from \Cref{thm:explicit expander}.
	We use the algorithm \makecanonical from \ref{lem: degree reduction balanced cut case} to compute,  in time $O(m)$, a canonical partition $(X',Y')$ of $V(\hG')$, such that $|X'|,|Y'|\geq \Omega(z)$, and $|E_{\hG'}(X',Y')|\leq O(|E_{\hG'}(X,Y)|)$. 
	 Assume w.l.o.g. that $|X'|\leq |Y'|$. We are then guaranteed that $|X'|\geq \Omega(z)$, and that for some constant $\mu$, $|E_{\hG'}(X',Y')|\leq \mu \psi |X'|$, or equivalently, $\Psi_{\hG'}(X',Y')\leq \mu \psi$. We set $S_i=X'$, delete the vertices of $S_i$ from $\hat G'$, and continue to the next iteration. Observe that set $V(\hat G')$ of vertices remains canonical.
	Otherwise, the outcome of  \ref{lem: iterations-final-cut} is  a graph $H$ with $V(H)=V(\hG')$, that is a $\psi_r(n')$-expander, together with a set $F$ of at most $O(z\log n)$ fake edges for $\hat G'$, and an embedding of $H$ into $\hG'+F$ with congestion at most $O(\log m/\psi)$, such that every vertex of $\hat G'$ is incident to at most $O(\log m)$ edges of $F$. In this case we say that the iteration terminates with an expander. If an iteration terminates with an expander, then the whole algorithm terminates.
	
Let $i$ denote the index of the last iteration of the algorithm that terminated with a cut. Recall that one of the following two cases must hold: 
\begin{itemize}
	\item (Case 1): the algorithm had exactly $i$ iterations, every iteration terminated with a cut, and $|\bigcup_{i'\leq i}S_{i'}|\geq |V(\hG)|/3$; 
	\item (Case 2): the algorithm had $(i+1)$ iterations, the first $i$th iterations terminated with cuts, and the last iteration terminated with an expander. 
\end{itemize}

In either case, let $S=\bigcup_{i'=1}^iS_{i'}$. Then $S$ is a canonical vertex set for $\hG$, and moreover, it is easy to verify that:

\begin{equation}
|E_{\hG}(S,\overline S)|\leq \mu \psi |S|\leq \mu \psi |V(\hG)|. \label{eq: few edges from S}
\end{equation}

Assume first that Case 1 happened. Consider the partition $(A',B')$ of $V(\hG)$, where $A'=S$ and $B'=V(\hG)\setminus S$. 
Recall that $|\bigcup_{i'<i}S_{i'}|< |V(\hG')|/3$ held (or we would not have executed the $i$th iteration). Let $\hG_i$ denote the graph $\hG'$ that served as input to the $i$th iteration, and let $n_i=|V(\hG_i)|$. Then $n_i\geq 2|V(\hG)|/3$. Let $(X_i,Y_i)$ be the cut that was returned by \ref{lem: iterations-final-cut}, and let $(X'_i,Y_i')$ be the canonical cut that we obtained in $\hG'$, so that $S_i=X'_i$. Recall that $|X'_i|\leq |Y'_i|$. It follows that $|Y'_i|\geq |V(\hG)|/3$, and $|\bigcup_{i'\leq i}S_{i'}|\geq |V(\hG)|/3$. Since $A'=\bigcup_{i'\leq i}S_{i'}$ and $B'=Y'_i$, we get that $|A'|,|B'|\geq |V(\hG)|/3$. From \ref{eq: few edges from S}, $|E_{\hG}(A',B')|\leq  \psi\mu |V(\hG)|$.

Lastly, we obtain a cut $(A,B)$ of $V(G)$ as follows. For every vertex $v_i\in V(G)$, if $V_i\subseteq A'$, then we add $v_i$ to $A$, and otherwise we add it to $B$. Since, for every $1\leq i\leq n$, $|V_i|=\deg_G(v_i)$, it is easy to verify that $\vol(A)=|A'|\geq |V(\hG)|/3=\vol(G)/3$, and similarly $\vol(B)\geq \vol(G)/3$. It is also immediate to verify that $|E_G(A,B)|=|E_{\hG}(A',B')|\leq \mu \psi |V(\hG)|=\mu \psi\cdot \vol(G)$. Since $\psi=\phi/\hat c$, by letting $\hat c$ be a large enough constant, we can ensure that $|E_G(A,B)|\leq \phi\cdot \vol(G)$. We return the cut $(A,B)$ as the outcome of the algorithm.

Assume now that Case 2 happened.
Let $\hG_{i+1}$ denote the graph $\hG'$ that served as input to the last iteration. Recall that in this last iteration, the algorithm from \ref{lem: iterations-final-cut} returned  a graph $H$ with $V(H)=V(\hG_{i+1})$, that is a $\psi_r(n')$-expander, where $n'=|V(\hG_{i+1}|\geq 2|V(\hG)|/3$, together with a set $F$ of at most $O(z\log n)$ fake edges for $\hG_{i+1}$, and an embedding of $H$ into $\hG_{i+1}+F$ with congestion at most $O(\log m/\psi)$, such that every vertex of $\hG_{i+1}$ is incident to at most $O(\log m)$ edges of $F$. 
Let $\hG''$ be the graph obtained from $\hG_{i+1}$, by adding the edges of $F$ to it.  Then graph $H$ embeds into $\hG''$ with congestion at most $O(\log m/\psi)$, and so, from \ref{lem: embedding expander gives expander}, graph $\hG''$ is a $\psi'$-expander, for $\psi'=\Omega(\psi_r(n')\cdot \psi/\log m)=\Omega\left (\phi/(\log m)^{O(r)}\right )$.

Recall that all vertex sets $S_1,\ldots,S_i$ are canonical; therefore, the set $V(\hG'')$ of vertices is also canonical. Let $G''$ be the graph obtained from $\hG''$ as follows. For every vertex $v_j\in V(G)$, if $V_j\subseteq V(\hG'')$, then we contract the vertices of $V_j$ into a single vertex $v_j$, and remove all self loops. 
Let $A'=V(G'')$. It is easy to verify that $G''$ can be obtained from $G[A']$, by adding at most $O(z\log m)$ edges to it -- the edges corresponding to the fake edges in $F$. Moreover, $\vol(A')=|V(\hG)|-|S|\geq 2|V(\hG)|/3\geq 2\vol(G)/3$. It is also easy to verify that $G''$ has conductance at least $\psi'$. Indeed, consider any cut $(X,Y)$ in $G''$. This cut naturally defines a cut $(X',Y')$ in $\hG''$: for every vertex $v_i\in A'$, if $v_i\in X$, then we add all vertices of $V_i$ to $X'$, and otherwise we add them to $Y'$. Then $|X'|=\vol_G(X)\geq \vol_{G''}(X)$, $|Y'|=\vol_G(Y)\geq \vol_{G''}(Y)$, and $|E_{\hG''}(X',Y')|=|E_{G''}(X,Y)|$. Since graph $\hG''$ is a $\psi'$-expander, we get that $|E_{G''}(X,Y)|\geq |E_{\hG''}(X',Y')|\geq \psi' \min\set{|X'|,|Y'|}\geq \psi'\min\set{\vol_{G''}(X),\vol_{G''}(Y)}$.

In our last step, we get rid of the fake edges in $G''$ by applying \ref{thm: expander pruning} to it, with conductance parameter $\psi'$, and the set $F$ of fake edges; (recall that $|F|=O(z\log n)$, and $z=\frac{\phi m}{\hat c(\log m)^{\hat c r}}$ for some large enough constant $\hat c$). In order to be able to use the theorem, we need to verify that $|F|\leq \psi'\cdot |E(G'')|/10$. Since $\psi'=\Omega\left (\phi/(\log m)^{O(r)}\right )$, and $|E(G'')|\geq \Omega(m)$, by letting $\hat c$ be a large enough constant, we can ensure that this condition holds. 
Applying \ref{thm: expander pruning} to graph $G''$, with conductance parameter $\psi'$, and the set $F$ of fake edges, we obtain a subgraph $G'\subseteq G''\setminus F$, of conductance at least $\psi'/6=\Omega\left (\phi/(\log m)^{O(r)}\right )$. Moreover, if we denote by $A=V(G')$ and $\tilde B=V(G'')\setminus V(G')$, then $|E_{G''}(A,\tilde B)|\leq 4k$ and:

\begin{equation}
 \vol_{G''}(\tilde B)\leq 8k/\psi'\leq O\left (k \cdot (\log m)^{O(r)}/\phi\right ), \label{eq: tilde B has small volume}
 \end{equation}
 
  where $k=|F|=O(z\log n)$ is the number of the fake edges. The running time of the algorithm from \ref{thm: expander pruning} is $\tilde O\left( m/\psi' \right)= O\left( m (\log m)^{O(r)}/\phi \right)$.
Let $B=V(G)\setminus A$. The algorithm then returns the cut $(A,B)$. We now verify that the cut has all required properties. 
We have already established that $G[A]$ has conductance at least $\phi/(\log m)^{O(r)}$.

Let $\tilde S=B\setminus \tilde B$. Then equivalently, we can obtain the set $\tilde S\subseteq V(G)$ of vertices from the set $S\subseteq V(\hG)$ of vertices (recall that $S=\bigcup_{i'=1}^iS_{i'}$) by adding to $\tilde S$ every vertex $v_j\in V(G)$ with $V_j\subseteq S$. Since, from \ref{eq: few edges from S}, $|E_{\hG}(S,\overline S)|\leq \mu \psi |V(\hG)|$ for some constant $\mu$, it is easy to verify that:

\begin{equation}
|E_G(\tilde S, V(G)\setminus \tilde S)|\leq \mu \psi\cdot \vol(G)=\mu \phi \cdot\vol(G)/\hat c. \label{eq: few edges from S in original graph}
\end{equation}

 From the above discussion, we are also guaranteed that $|E_{G''}(A,\tilde B)|\leq 4|F|\leq O(z\log n)$. Since $z=\frac{\phi m}{\hat c(\log m)^{\hat c r}}$, by letting $\hat c$ be a large enough constant, we can ensure that $|E_{G''}(A,\tilde B)|< \phi m/100\leq \phi \vol(G)/100$. Therefore, altogether, we get that:

\[|E_G(A,B)|\leq |E_G(A,\tilde B)|+|E_G(\tilde S,V(G)\setminus \tilde S)|\leq \phi\cdot \vol(G)/100 + \phi \mu \cdot \vol(G)/\hat c\leq  \phi\cdot \vol(G), \]

if $\hat c$ is chosen to be a large enough constant.

Lastly, it remains to verify that $\vol_G(A)\geq \frac{7}{12}\cdot \vol(G)$.
Recall that $|\hat V(G_{i+1})|\geq 2|V(\hat G)|/3\geq 2\vol(G)/3$. Therefore, if we denote by $U=V(G'')=V(G)\setminus \tilde S$, then $\vol_G(U)\geq 2\vol(G)/3$. Recall that $A=U\setminus \tilde B$, and, from \ref{eq: tilde B has small volume}, $\vol_{G''}(\tilde B)\leq O\left (k \cdot (\log m)^{O(r)}/\phi\right )\leq O\left (z \cdot (\log m)^{O(r)}/\phi\right )$. Moreover, $\vol_G(\tilde B)\leq \vol_{G''}(\tilde B)+E_{G}(\tilde S,\tilde B)\leq \vol_{G''}(\tilde B)+E_{G}(U,\tilde S)$. From \ref{eq: few edges from S in original graph}, we get that:

\[\vol_G(\tilde B)\leq O\left (z \cdot (\log m)^{O(r)}/\phi\right )+ O(\mu\phi\vol(G)/\hat c).  \]

Since $z=\frac{\phi m}{\hat c(\log m)^{\hat c r}}$, by letting $\hat c$ be a large enough constant, we can ensure that $\vol_G(\tilde B)\leq \vol(G)/12$. We then get that $\vol_G(A)\geq |\hat V(G_{i+1})|-\vol_G(\tilde B)\geq 2\vol(G)/3-\vol(G)/12\geq 7\vol (G)/12$.

It now remains to analyze the running time of the algorithm. 
The time required to construct graph $\hG$ from graph $G$ is $O(m)$.
Recall that, if an iteration terminates with a cut, then we delete from $\hat G'$ a set of at least $\Omega(z)$ vertices. Therefore, the total number of iterations is bounded by $O(|V(\hat G)|/z)=O(m/z)=O\left ((\log m)^{O(r)}/\phi\right )$. The running time of each iteration is:
 
 $$\tilde O\left(m^{1+O(1/r)}\cdot (\log m)^{O(r^2)}+m/\psi\right )=\tilde O\left(m^{1+O(1/r)}\cdot (\log m)^{O(r^2)}+m/\phi\right ).$$
  
 At the end of each iteration, we employ \ref{lem: degree reduction balanced cut case} to turn the resulting cut into a canonical one, in time $O(m)$.
 Therefore, the total running time of the iterations is $\tilde O\left(m^{1+O(1/r)}\cdot (\log m)^{O(r^2)}/\phi^2\right )$. 
Lastly, if Case 2 happens, we employ the algorithm from \ref{thm: expander pruning}, whose running time, as discussed above, is $\tilde O\left (m (\log m)^{O(r)}/\phi\right )$. Altogether, the running time of the algorithm is $\tilde O\left(m^{1+O(1/r)}\cdot (\log m)^{O(r^2)}/\phi^2\right )$.
\end{proofof}

%-----------------------------------------------------------
%-----------------------------------------------------------
%-----------------------------------------------------------
%-----------------------------------------------------------
%-----------------------------------------------------------
%-----------------------------------------------------------
%-----------------------------------------------------------
%-----------------------------------------------------------
%-----------------------------------------------------------

\section{Applications of $\BCut$}

\label{sec:app}

In this section, we provide applications of the algorithm for $\BCut$ from \ref{thm:BCut phi}. Some of the results are summarized in \Cref{tabel:applications-static,tabel:applications-dyn}. 
We use the $\Ohat(\cdot)$ notation to hide sup-polynomial
lower order terms. Formally $\Ohat(f(n))=O(f(n)^{1+o(1)})$; equivalently,
for any constant $\theta>0$, we have $\Ohat(f(n))\le O(f(n)^{1+\theta})$.
This notation can be viewed as a direct generalization of the $\tilde{O}(\cdot)$
notation for hiding logarithmic factors, and behaves in a similar manner.

\subsection{Expander Decomposition}

An \emph{$(\epsilon,\phi)$-expander decomposition} of a graph $G=(V,E)$
is a partition $\P=\{V_{1},\dots,V_{k}\}$ of the set $V$ of vertices, such that for all $1\leq i\leq k$, the conductance of graph $G[V_i]$ is at least $\phi$, and $\sum_{i-1}^k\delta_{G}(V_{i})\le\epsilon\vol(G)$. This decomposition was introduced in \cite{KannanVV04,GoldreichR99} and has been used as a key tool in many applications, including the ones mentioned in this paper.

Spielman and Teng \cite{SpielmanT04} provided the first near-linear time algorithm, whose running time is $\Otil(m/\poly(\epsilon))$, for computing a \emph{weak} variant of the $(\epsilon,\epsilon^2 / \poly(\log n))$-expander decomposition, where, instead of ensuring that each resulting graph $G[V_i]$ has high conductance, the guarantee is that for each such set $V_i$ there is some larger set $W_i$ of vertices, with $V_i\subseteq W_i$, such that $\Phi(G[W_i]) \ge \epsilon^2 / \poly(\log n)$.
 This caveat was first removed in \cite{NanongkaiS17}, who showed an algorithm for computing  an $(\epsilon,\epsilon / n^{o(1)})$-expander decomposition in time $O(m^{1+o(1)})$ (we note that \cite{Wulff-Nilsen17} provided similar results with somewhat weaker parameters). More recently, \cite{SaranurakW19} provided an algorithm for computing $(\epsilon,\epsilon / \poly(\log n))$-expander decomposition in  time $\Otil(m/\epsilon)$. Unfortunately, all algorithms mentioned above are randomized. 

The only previous subquadratic-time deterministic algorithm  for computing an expander decompositions is implicit in~\cite{GaoLNPSY19}.
It computes an  $(\epsilon,\epsilon / n^{o(1)})$-expander decomposition in time $O(m^{1.5+o(1)})$.
We provide the first deterministic algorithm for computing expander decomposition in almost-linear time:
% which in turn implies many other deterministic algorithms in next subsections.

\begin{cor}
\label{cor:expander decomp} There is a deterministic algorithm that, given a graph
$G=(V,E)$ with $m$ edges, and parameters $\epsilon\in(0,1]$ and $1\leq r\leq O(\log m)$, computes
a $\left(\epsilon,\phi \right)$-expander decomposition of $G$ with $\phi=\Omega(\epsilon/(\log m)^{O(r^2)})$,
 in time  $O\left ( m^{1+O(1/r)}\cdot (\log m)^{O(r^2)} / \epsilon^2 \right )$.
\end{cor}

\begin{proof}
	We maintain a collection $\hset$ of disjoint sub-graphs of $G$ that we call \emph{clusters}, which is partitioned into two subsets, set $\hset^A$ of \emph{active clusters}, and set $\hset^I$ of \emph{inactive clusters}. We ensure that for each inactive cluster $H\in \hset^I$, $\Phi(H)\geq \phi$. We also maintain a set $E'$ of ``deleted'' edges, that are not contained in any cluster in $\hset$. At the beginning of the algorithm, we let $\hset=\hset^A=\set{G}$, $\hset^I=\emptyset$, and $E'=\emptyset$. The algorithm proceeds as long $\hset^A\neq \emptyset$, and consists of iterations. 
	For convenience, we denote $\alpha=(\log m)^{r^2}$, and we set $\phi=\epsilon/(c\alpha\cdot \log m)$, for some large enough constant $c$, so that 
$\phi=\Omega(\epsilon/(\log m)^{O(r^2)})$ holds.

	In every iteration, we apply the algorithm from \ref{thm:BCut phi} to every graph $H\in \hset^A$, with the same parameters $\alpha$, $r$, and $\phi$. Consider the cut $(A,B)$ in $H$ that the algorithm returns, with $|E_H(A,B)|\leq \alpha \phi\cdot \vol(H)\leq \frac{\epsilon\cdot \vol(H)}{c\log m}$. We add the edges of $E_H(A,B)$ to set $E'$. If $\vol_{H}(A),\vol_H(B)\ge \vol(H)/3$, then we replace $H$ with $H[A]$ and $H[B]$ in $\hset$ and in $\hset^A$. Otherwise, we are guaranteed that 
	$\vol_H(A)\geq \vol(H)/2$, and graph $H[A]$ has conductance at least $\phi$. Then we remove $H$ from $\hset$ and $\hset^A$, add $H[A]$ to $\hset$ and $\hset^I$, and add $H[B]$ to $\hset$ and $\hset^A$.
	
	When the algorithm terminates, $\hset^A=\emptyset$, and so every graph in $\hset$ has conductance at least $\phi$. Notice that in every iteration, the maximum volume of a graph in $\hset^A$ must decrease by a constant factor. Therefore, the number of iterations is bounded by $O(\log m)$. It is easy to verify that the number of edges added to set $E'$ in every iteration is at most $\frac{\epsilon\cdot \vol(G)}{c\log m}$. Therefore, by letting $c$ be a large enough constant, we can ensure that $|E'|\leq \epsilon \vol(G)$. The output of the algorithm is the partition $\pset=\set{V(H)\mid H\in \hset}$ of $V$. From the above discussion, we obtain a valid $(\epsilon, \phi)$-expander decomposition, for $\phi=\Omega\left (\epsilon/(\log m)^{O(r^2)}\right )$. 
	
	It remains to analyze the running time of the algorithm.  The running time of a single iteration is bounded by $O\left ( m^{1+O(1/r)}\cdot (\log m)^{O(r^2)} / \phi^2 \right )=O\left ( m^{1+O(1/r)}\cdot (\log m)^{O(r^2)} / \epsilon^2 \right )$. Since the total number of iterations is bounded by $O(\log m)$, we get that the total running time of the algorithm is $O\left ( m^{1+O(1/r)}\cdot (\log m)^{O(r^2)} / \epsilon^2 \right )$.
\end{proof}

We provide another algorithm for the expander decomposition, whose running time no longer depends on $\eps$ in \ref{sec:BCut_nophi} (see \ref{cor:expander decomp2}).

\subsection{Dynamic Connectivity and Minimum Spanning Forest}

In this section we provide a deterministic algorithm for dynamic Minimum Spanning Forest (\MSF) with $n^{o(1)}$
worst-case update time. 
   
\begin{cor}\label{cor:dynconn}
There is a deterministic algorithm that, given an $n$-vertex graph
$G$ undergoing edge insertions and deletions, maintains a minimum spanning forest of $G$ with $n^{o(1)}$%$2^{O(\log n\log\log n)^{2/3}}$
worst-case update time.
\end{cor}

By implementing the link-cut tree data structure \cite{SleatorT83} on top of the minimum spanning forest, 
this algorithm immediately implies  a deterministic algorithm for \DC with the same update time and  $O(\log n)$ query time for answering connectivity queries between pairs of vertices, proving \ref{thm:intro:dynConn}.
Thus, we resolve the longstanding open problem of improving the $O(\sqrt{n})$  worst-case update time from the classical algorithm by Frederickson \cite{Frederickson85,EppsteinGIN97}. The previous best deterministic algorithm for \DC, due to  \cite{Kejlberg-Rasmussen16},  has  worst-case update time $O(\sqrt{n (\log\log n)^2/\log n})$. Below, we prove \ref{cor:dynconn}.

\paragraph{Reduction to Expander Decomposition.}
From now, we write NSW to refer to \cite{NanongkaiSW17}.
The algorithm by NSW can be viewed as a reduction to expander decomposition
as follows. For
any $\gamma>1$, suppose that, given an $n$-vertex graph $G$ with
maximum degree $3$, we can compute an $(\epsilon,\phi)$-expander
decomposition where $\epsilon=1/\gamma^{2}$ and $\phi=1/\gamma^{3}$
in $n\poly(\gamma)$ time.  
Then, NSW show that there is an algorithm for maintaining a minimum spanning forest on a graph with at most $n$ vertices with worst-case update time
\begin{equation}
t_{u}(n)=\tilde{O}(\pi(\gamma)\poly(\gamma))+O(\log n)\cdot t_{u}(O(n/\gamma))\label{eq:reduction NSW}
\end{equation}
where $\pi(\gamma)$ is a function such that $\pi(\gamma)=n^{o(1)}$
as long as $\gamma=n^{o(1)}$. This follows from the proof of Lemma
9.28 of NSW. 
Solving this recursion, we obtain

\begin{equation}
t_{u}(n)=O(\pi(\gamma)\poly(\gamma))\cdot O(\log n)^{O(\log_{\gamma}n)}.
\label{eq:update time NSW}
\end{equation}
This reduction is deterministic (after a slight modification which
we will describe later). Observe that, for any $\gamma$ where $\gamma = \omega(\polylog(n))$ and $\gamma = n^{o(1)}$, we have $t_u(n) = n^{o(1)}$.

In NSW (Lemma 8.7), they show a randomized algorithm for computing  a $(1/\gamma^{2},1/\gamma^{3})$-expander
decomposition of a bounded degree graph with running time
$n\poly(\gamma)$ where $\gamma=n^{O(\log\log n/\sqrt{\log n})}$.
The above reduction then implies a randomized dynamic minimum spanning
forest algorithm with $n^{o(1)}$ update time. 
We can immediately derandomize this algorithm using \Cref{cor:expander decomp} as follows:

\begin{lem}[Deterministic Version of Lemma 8.7 of \cite{NanongkaiSW17} for Bounded-degree Graphs]\label{lem: expander decomposition for NSW} There
	is a deterministic algorithm $\mathcal{A}$ that, given an $n$-vertex graph $G=(V,E)$ with maximum degree $3$, and a parameter $\alpha>0$, computes an $(\alpha \gamma,\alpha)$-expander decomposition of $G$ in $O(n\gamma/\alpha^{2})$ time where $\gamma = 2^{O(\log n \log\log n)^{2/3}}$.
\end{lem}
%
%\begin{lem}[Deterministic Version of Lemma 8.7 of \cite{NanongkaiSW17} for Bounded-degree Graphs]\label{lem: expander decomposition for NSW} There
%	is a deterministic algorithm $\mathcal{A}$ that, given an $n$-vertex graph $G=(V,E)$ with maximum degree $3$, and an expansion parameter $\alpha>0$, computes a partition $\mathcal{Q}=\{V_{1},\dots,V_{k}\}$ of $V$,
%	such that:
%	\begin{itemize}
%		\item the number of edges of $G$ whose endpoints lie in different sets of the partition ${\mathcal Q}$ is at most $\alpha\gamma n$, where $\gamma=2^{O(\log n\log\log n)^{2/3}}$; and
%		\item for all $1\leq i\leq k$, $\Psi(G[V_{i}])\ge\alpha$.
%	\end{itemize}
%	The running time of the algorithm is $O(n\gamma/\alpha^{2})$.
%\end{lem}

\begin{proof}
	Let $r=\log^{1/3}n$ and $\epsilon=c_{0}\alpha(\log n)^{O(r^{2})}$
	for a large enough constant $c_{0}$. 
	%The proof of \ref{lem: expander decomposition for NSW} follows directly by applying the algorithm from \Cref{cor:expander decomp} to graph $G$, with the parameter $\epsilon$.
	The algorithm with parameter $\epsilon$ from \Cref{cor:expander decomp} returns a partition $\mathcal{Q}=\{V_{1},\dots,V_{k}\}$ of $V$, such that for all $1\leq i\leq k$, each $\Phi(G[V_{i}])\ge\Omega(\epsilon/(\log n)^{O(r^{2})})\ge \Omega(c_0\alpha) \ge \alpha$, if we assume that $c_{0}$ is a large enough constant. Moreover, the number of edges whose endpoints lie in different sets of the partition is at most $O(\epsilon n)=O(\alpha n(\log n)^{O(\log^{2/3}n)})\le\alpha\gamma n$, since $\gamma=2^{O(\log n\log\log n)^{2/3}}$.
	The running time of \Cref{cor:expander decomp} is $O(n^{1+O(1/r)}(\log n)^{O(r^{2})}/\epsilon^{2})=O(n\gamma/\alpha^{2})$, since $\gamma=2^{O(\log n\log\log n)^{2/3}}$.
\end{proof}

By plugging the above algorithm with $\alpha = 1/\gamma^3$ into the reduction of NSW, we obtain the \emph{deterministic}
dynamic minimum spanning forest algorithm with $n^{o(1)}$ update
time.

\paragraph{Overview of the Reduction.}

In this rest of this section, we explain how the above reduction by NSW works
in high-level and, in particular, why \ref{eq:reduction NSW} holds. We also describe
the slight modification of the reduction so that there is no randomized
component in it. 
%\paragraph{Strategy.}
%We give a high level description of the algorithm of \cite{NanongkaiSW17} below and then describe the randomized components and how to derandomize them. 

Let $G$ be a weighted $n$-vertex graph undergoing edge insertions
and deletions and let $\msf(G)$ denote the minimum spanning forest
of $G$. There are two high-level steps. The first step is to maintain
a \emph{sketch graph} $H$ satisfying two properties: (1) $\msf(H)=\msf(G)$
and (2) $H$ is a subgraph of $G$ containing only $n+k$ edges where
$k=o(n)$. The second step is to maintain the minimum spanning forest
$\msf(H)$ of $H$. 

Below, we will sketch how NSW implement this strategy in the special case
when $G$ is initially an expander and we sketch how they generalize the algorithm. 

We emphasize that the problem is non-trivial even when we have a promise that the underlying graph is always an expander throughout the updates. 
Indeed, if our goal is only maintaining connectivity of an expander, then the problem becomes trivial as an expander must be connected. 
However, suppose we want to maintain a spanning forest (not necessarily minimum) and a tree edge is deleted. Then, there is no known simple deterministic method to find a replacement edge and update the forest accordingly, even if the graph is an expander.\footnote{On the contrary, if we only need a randomized algorithm \emph{against an adaptive adversary}, there is a simple algorithm based on random sampling as shown in \cite{NanongkaiS17}.} 

Moreover, even if we want to maintain only connectivity of a general graph, the NSW algorithm still needs to a subroutine for maintaining a spanning forest on expanders. So, to explain the NSW algorithm, we need to explain how to maintain a spanning forest in an expander. As the algorithm for minimum spanning forest is not much more complicated, we give the overview for maintaining minimum spanning forests below.

Using a standard
reduction, we will assume that $G$ has maximum degree $3$ and $\msf(G)$
is unique. %Let $t_{u}(n)$ be the update time for maintaining the minimum spanning forest on a graph with at most $n$ vertices. 

\paragraph{Special case: Using Expander Pruing.}

Suppose that $G=(V,E)$ is a $(1/\gamma^3)$-expander for some $\gamma=n^{o(1)}$.
At the preprocessing step, NSW simply set the initial sketch graph $H_{0}=\msf(G)$. Then, given
a sequence of edge updates to $G$, they employ the \emph{dynamic expander
	pruning }algorithm (Theorem 5.1 of NSW) that maintains a \emph{pruned
	set} $P\subset V$ such that, 
\begin{itemize}
	\item $P = \emptyset$ initially and vertices only join $P$ and are never removed,
	\item for some $P'\subseteq P$, $G[V-P']$ is connected\footnote{$G[V-P']$ actually has conductance at least $1/n^{o(1)}$ but NSW do not exploit that.}, and 
	\item $P$ can be updated in $\pi=\pi(\gamma)$ worst-case time, such that %$\pi(\gamma)$ is such that $\pi(\gamma)\gg\gamma$ and
	$\pi(\gamma)=n^{o(1)}$ as
	long as $\gamma=n^{o(1)}$.
%	\footnote{In fact, by using the improved implementation of the expander pruning from \cite{SaranurakW19}, we can have $\pi(\gamma)=2^{O(\sqrt{\log n})}\gamma^{2}$. The dependency on $\gamma$ of the function $\pi(\gamma)$in NSW is much worse but it is enough to obtain $n^{o(1)}$ final update time.} 
	In particular, $|P|\le i\cdot\pi$ after the $i$-th update.
\end{itemize}
At any time, let $I$ be the set of inserted edges and $D$ be the
set of deleted edges. They maintain $H=H_{0}\cup E_{G}(P,V)\cup I\setminus D$.
That is, $H$ contains all edges from the original minimum spanning
forest $H_{0}$, all edges incident to the pruned set $P$, and all
newly inserted edges, and we exclude all the deleted edges from $H$. 
Because $G[V-P']$ is connected, it is not hard to see that $H$ contains all edges of the current $\msf(G)$
and so $\msf(H)=\msf(G)$. Therefore, when $G$ is initially an expander, the task of maintaining the sketch
graph $H$ is only amount to maintaining the pruned set $P$. Moreover,
if the length of update sequence is at most $T=O(n/(\gamma\pi))$,
then we have that the number of edges in $H$ is $n+k$ where $k=T\pi=O(n/\gamma)$ which is sublinear in $n$ as desired.

Next, the goal is to maintain $\msf(H)$. As $H$ has at most $n+k$
edges, NSW observe that the \emph{contraction technique} by \cite{HolmLT01} allows them
to recursively reduce the problem to graphs with
$O(k)$ vertices. In slightly more detail, observe that given an edge
update in the original graph $G$, this update in $G$ corresponds at most $O(\pi)$ edge insertions in $H$
and one edge deletion to $H$. This is because $P$ only grows by at most $\pi$ vertices per step. 

The strength of the contraction technique is as follows: it can handle a batch of edge insertions without recursive calls, and each edge deletion in $H$ corresponds to $O(\log n)$ recursive calls to smaller graphs of size $O(k)$. 
That is, the $\msf(H)$ can be maintained with update time $\tilde{O}(\pi)+O(\log n)\cdot t_{u}(O(k))$
per one update to the original graph $G$. 
%, where the term $t_u(O(k))$ corresponds to the update time for maintaining a minimum spanning forest in $H'$. 
The detail can be found in
Section 7 of NSW.

To summarize, suppose that $G$ is a $(1/\gamma^3)$-expander, each of the $T=O(n/(\gamma\pi))$ updates can be handled in time at most 
\[
\tilde{O}(\pi)+O(\log n)\cdot t_{u}(O(n/\gamma)).
\]
If this was true even for an arbitrary $n$-vertex graph $G$, then we would have a recursive algorithm with worst-case update time
\[
t_{u}(n)=\tilde{O}(t_{pre}/T)+\tilde{O}(\pi)+O(\log n)\cdot t_{u}(O(n/\gamma))
\]
where $t_{pre}$ denotes the preprocessing time. The first term above follows from the fact that we need to restart the data structure after every $T$ updates.
Note that this does not make the update time amortized because the time required for restarting can be distributed using the standard \emph{building-in-the-background} technique.
If the preprocessing time is $t_{pre}=O(n\poly(\gamma))$, then the recursion implies that 
\[
t_{u}(n)={O}(\pi \poly(\gamma))\cdot O(\log n)^{O(\log_{\gamma}n)}.
\]
which is the same as \Cref{eq:update time NSW} as we desired.
%So as long as $\gamma$ is super-polylogarithmic, we would have $O(\log n)^{O(\log_{\gamma}n)}=n^{o(1)}$ and so $t_{u}(n)=n^{o(1)}$.

\paragraph{General case: Using $\protect\msf$ Decomposition.}
The analysis above overly simplifies the NSW algorithm because $G$
might not be an expander. 

The key tool that allows NSW to work with general graphs is called
the $\msf$ decomposition (Theorem 8.3 of NSW). The $\msf$ decomposition
is an intricate hierarchical decomposition of a graph tailored for
the dynamic minimum spanning forest problem. It is a combination of
three kinds of graph partitioning including (1) the expander decomposition,
(2) a partitioning of edges into groups sorted by the edge weights,
and (3) the $M$-clustering (introduced in \cite{Wulff-Nilsen17})
which partitions $\msf(G)$ into small subtrees. See Section 8 of
NSW for detail. For us the only important point is that, the $\msf$
decomposition calls the $(1/\gamma^{2},1/\gamma^{3})$-expander decomposition of bounded degree graphs
as a subroutine, and if the expander decomposition runs in $O(n\poly(\gamma))$ time,
then so does the $\msf$ decomposition.\footnote{In Theorem 8.3, there are actually other parameters $d,\alpha,s_{low},s_{high}$.
	But the NSW algorithm sets $\alpha=1/\gamma^{3},d=\gamma,s_{low}=\gamma,s_{high}=n/\gamma$
	(see NSW on Page 30 below Theorem 8.3). That is, all properties of
	Theorem 8.3 are dictated by the parameter $\gamma$.}

The strategy of the algorithm remains the same: to first maintain a sketch graph $H$ which is a very sparse graph where $\msf(H)=\msf(G)$, and then maintain the minimum spanning forest $\msf(H)$ of $H$. Given a general weighted graph
$G$, the NSW algorithm proceeds as follows. First, they preprocess
the graph $G$ by (mainly) applying the $\msf$ decomposition using $n\poly(\gamma)$
time. This decomposition will define the initial sketch graph $H_{0}$.
The precise definition of the sketch graph $H$ is complicated and
is omitted here. The important point for us is that the initial number
of edges in $H_{0}$ directly exploits the guarantee of  $(1/\gamma^{2},1/\gamma^{3})$-expander decomposition. More precisely, they have $|E(H_{0})|\le n+k$ where $k=O(n/\gamma)$. 

Then, given a sequence of updates, the size of $H$ will grow in a
similar way as described in the case when the graph is a expander.
That is, given a single edge update in $G$, this corresponds to at
most $O(\pi\poly(\gamma))$ edge insertions to $H$ and $O(1)$ edge
deletions in $H$. This bounds follows by a careful definition of
$H$ and a strong guarantee of the $\msf$ decomposition. (See Section 9.1 of NSW for the definition of $H$ and Theorem 8.3
of NSW for precise guarantee of the $\msf$ decomposition.) For each
update in $G$, the time for updating $H$ is 
\[
\tilde{O}(\pi\poly(\gamma))+O(1)\cdot t_{u}(O(n/\gamma)).
\]
Note that,
if the update sequence has length $T=O(n/(\pi\poly(\gamma)))$, then
we can guarantee that $H$ always has at most $n+O(n/\gamma)$ edges. 

Given that we can maintain the sketch graph $H$, the NSW algorithm
maintains $\msf(H)$ in the same way we have described when we know
that the graph is an expander. That is, NSW apply the contraction
technique by \cite{HolmLT01} to recursively solve the problem in smaller graphs of size $O(n/\gamma)$. In the end, the update time can
be written as
\[
t_{u}(n)=\tilde{O}(t_{pre}/T)+\tilde{O}(\pi\poly(\gamma))+O(\log n)\cdot t_{u}(O(n/\gamma))
\]
where $t_{pre}=n\poly(\gamma)$ and $T=O(n/(\pi\poly(\gamma)))$.
This implies \Cref{eq:reduction NSW} as we desired. 
It remains to point out the randomized components in this algorithms and show how to derandomize them.

\paragraph{Derandomization.}

The NSW algorithm has only two randomized components.
%There are only two randomized components
%in the fully dynamic minimum spanning forests algorithm of \cite{NanongkaiSW17}. 
The first randomized component is the $\msf$ decomposition algorithm from Theorem 8.3 in Section 8 of NSW. 
%The algorithm is applied to a bounded-degree graph, and its running time is  $n\poly(\gamma_{nsw})$.
The only source of randomization in Theorem 8.3 comes from the algorithm for computing the expander decomposition (Lemma 8.7 in NSW). By replacing this algorithm with \ref{cor:expander decomp}, we obtain a deterministic implementation of Theorem 8.3.%; we discuss its parameters in more detail below.

The second randomized component is Theorem 6.1 from Section 6 in NSW,
which is an extension of the dynamic expander decomposition from Theorem 5.1 in NSW. The algorithm from Theorem 5.1 is deterministic but needs to assume that its input graph is a high-conductance graph.
Unfortunately, the algorithm for computing the $\msf$ decomposition from Theorem 8.3 is randomized, and only ensures that the resulting sub-graphs have high conductance with high probability. Theorem 6.1 is an extension of Theorem 5.1 that allows it to work even if the input graph has low conductance. Since the new deterministic algorithm from  \ref{cor:expander decomp} guarantees that the sub-graphs obtained in the $\msf$ decomposition have high conductance, we no longer need to use Theorem 6.1, and the algorithm from Theorem 5.1, which is deterministic, is now sufficient.

To summarize, we only need to modify the NSW algorithm as follows: (1) bypassing Theorem 6.1 of NSW and directly applying Theorem 5.1 of NSW for expander pruning, and (2) replacing the randomized expander decomposition algorithm from Lemma 8.7 of NSW by the deterministic version described in \ref{lem: expander decomposition for NSW}.

\subsection{Spectral Sparsifiers }

Our deterministic algorithm for computing expander decompositions from \ref{cor:expander decomp} immediately implies a deterministic algorithm for  the original application of expander decompositions: constructing spectral sparsifiers~\cite{SpielmanT11-SecondJournal}.
Suppose we are given a undirected weighted $n$-vertex graph $G=(V,E,\ww)$ (possibly with self-loops). The Laplacian
$L_{G}$ of $G$ is a matrix of size $n\times n$ whose entries are defined as follows: 
\[
L_{G}(u,v)=\begin{cases}
0 & u\neq v, (u,v)\not\in E\\
-\ww_{uv} & u\neq v, (u,v)\in E\\
\sum_{\stackrel{(u,u')\in E:} {u\neq u'}}\ww_{uu'} & u=v.
\end{cases}
\]
 We say that a graph $H$ is an \emph{$\alpha$-approximate spectral
sparsifier} for $G$ iff for all $\xx\in\mathbb{R}^{n}$, $\frac{1}{\alpha}\xx^{\top}L_{G}\xx\le \xx^{\top}L_{H}\xx\le\alpha\cdot \xx^{\top}L_{G}\xx$ holds. 

All previous deterministic algorithms for graph sparsification, including those computing cut sparsifiers,
exploit explicit potential function-based approach of
Batson, Spielman and Srivastava~\cite{BatsonSS12}.
All previous algorithms that achieve faster running time either perform
random sampling~\cite{SpielmanS08:journal}, or use random
projections, in order to estimate the importances of edges~\cite{ZhuLO15}.
We provide the first deterministic, almost-linear-time algorithm for computing a spectral sparsifier of a \emph{weighted} graph. We emphasize that although all algorithms from previous sections are designed for unweighted graphs, the fact that spectral sparsifiers are ``decomposable'' allows us to easily reduce the problem on weighted graphs to the one on unweighted graphs.

%\jnote{why does it say that $|E(H)|\leq O(n\log^2 n)$? I think we get $O(n\log n)$?}
\begin{cor}\label{cor:sparsifier}
There is a deterministic algorithm, that we call $\mathtt{SpectralSparsify}$ that,
given an undirected 
$n$-node $m$-edge graph $G=(V,E,\ww)$ with integral edge weights $\ww$ bounded by $U$, and a parameter $1\le r \le O(\log m)$, computes a $(\log m)^{O(r^2)}$-approximate
spectral sparsifier $H$ for $G$, with $|E(H)|\leq O\left(n \log n\log U\right)$, in time 
%$O(m2^{O\left(\log m\log\log m\right)^{2/3}})=\Ohat(m)$. 
$O\left ( m^{1+O(1/r)}\cdot (\log m)^{O(r^2)}  \log U\right )$. %\textcolor{red}{TODO: support dependency on U. Size should be $n \log n \log U$, approx factor is unchanged.}
\end{cor}

\begin{proof}
	We first assume that $G$ is unweighted.
We compute a $(1/2,\phi)$-expander decomposition  $\P=\{V_1, V_2, \ldots, V_k\}$ of $G$, for $\phi=1/(\log m)^{O(r^2)}$, using the algorithm from \ref{cor:expander decomp}. Let $\hat E$ denote the set of all edges $e\in E(G)$, whose endpoints lie in different sets in the partition $\P$. If $\hat E\neq \emptyset$, then we continue the expander decomposition recursively on $G[\hat E]$. Notice that the depth of the recursion is bounded by $O(\log m)$. When this process terminates, we obtain a collection $\set{G_1,\ldots,G_z}$ of sub-graphs of $G$, that are disjoint in their edges, such that $\bigcup_{j=1}^zE(G_j)=E(G)$. Moreover, we are guaranteed that for all $1\leq j\leq z$, graph $G_j$ has conductance are at least $\phi=1/(\log m)^{O(r^2)}$. It is now enough to compute a spectral sparsifier for each of the resulting graphs $G_1,\ldots,G_z$ separately.

We can now assume that we are given a graph $G$ whose conductance is at least $\phi=1/(\log m)^{O(r^2)}$, and our goal is to construct a spectral sparsifier for $G$. In order to do so, we will first approximate $G$ by a ``product demand graph'' $D$, that was defined in \cite{KyngLPSS16}, and then use the construction  of \cite{KyngLPSS16}, that can be viewed as a strengthening of Algorithm  \constructexpander from \ref{thm:explicit expander}, in order to sparsify $D$.

\begin{defn}[Definition G.13, \cite{KyngLPSS16}]
Given a vector $\dd\in (\mathbb{R}_{>0})^n$,  its corresponding \emph{product demand graph} $H(\dd)$, is a complete weighted graph on $n$ vertices with self-loops, where for every pair $i,j$ of vertices, the weight $\ww_{ij}=\dd_i\dd_j$. 
\end{defn}

Given an $n$-node edge-weighted graph $G=(V,E,\ww)$, let $\deg_G \in \mathbb{Z}^n$ be the vector of weighted degrees of every vertex (that includes self-loops), so for all $j\in V$, the $j$th entry of $\deg_G$ is $\deg_G(j)=\sum_{i\in V}w_{i,j}$. Given an input graph $G$, we construct a product demand graph $D=\frac{1}{\vol(G)}H(\deg_G)$.  It is immediate to verify that the weighted degree vectors of $D$ and $G$ are equal, that is,  $\deg_D=\deg_G$.

Next, we need to extend the notion of conductance to weighted graphs with self loops. Consider a weighted graph $H=(V',E',\ww')$ (that may have self-loops), and let $S\subseteq V'$ be a cut in $H$. We then let $\delta_H(S)=\sum_{\stackrel{(u,v)\in E':}{u\in S,v\not\in S}}\ww'_{u,v}$, and we let $\vol_H(S)=\sum_{v\in S}\sum_{u\in V'}\ww'_{u,v}$.
A weighted conductance of the cut $S$ in $H$ is then: $\frac{\delta_H(S)}{\min\set{\vol_H(S),\vol_H(\overline S)}}$, and the conductance of $H$ is the minimum conductance of any cut in $H$.
 We need the following observation:

\begin{observation}\label{obs: cond of D}
	The weighted conductance of graph $D$ is at least $1/2$.
	\end{observation}

\begin{proof}
	Consider any cut $S$ in $D$. Observe that, from our construction, $\delta_D(S)=\vol_G(S)\cdot \vol_G(\overline S)/\vol(G)$. It is also easy to see that $\vol_H(S)=\vol_G(S)$. Assume without loss of generality that $\vol_D(S)\leq \vol_D(\overline S)$, so $\vol_D(\overline S) \geq \vol(G)/2$. Then the conductance of the cut $S$ is: 
	
	\[\frac{\delta_D(S)}{\vol_D(S)}=\frac{\vol_G(S)\cdot \vol_G(\overline S)}{\vol(G)\cdot \vol_G(S)}\geq \frac 1 2.\]
\end{proof}

%By construction, for any cut $S$ of $D$, its conductance is $$\frac{\delta_{D}(S)}{min\{\vol_D(S), \vol_D(V-S)\}}=\frac{\vol_D(S)\cdot \vol_D(V-S)}{\vol(D) \cdot \min\{\vol_D(S), \vol_D(V-S)\}}\ge \frac{1}{2}.$$

In the following lemma, we show that $D$ is a spectral sparsifier for $G$.
\begin{lem}
\label{lem:degree_approx}
Let $D$ and $G$ be two undirected weighted $n$-vertex graphs with $V(D)=V(G)$, such that $\deg_D=\deg_G$. Assume further that $\Phi(D), \Phi(G)\ge \phi$ for some conductance threshold $\phi$. Then for any real vector $\xx\in \reals^n$: $\frac{\phi^2}{4}\xx^\top L_G\xx\le \xx^\top L_D\xx\le \frac{4}{\phi^2}\xx^\top L_G\xx$.
\end{lem}
\begin{proof}
The normalized Laplacian $\Lhat_H$ of a weighted graph $H$ is defined as $W_H^{-1/2}L_H W_H^{-1/2}$, where $L_H$ is the Laplacian of $H$ and $W_H$ is a diagonal weighted-degree matrix, where for every vertex $v$ of $H$, $(W_H)_{vv}=\deg_H(v)$.

Let $\Lhat_D$ and $\Lhat_G$ be normalized Laplacians of $D$ and $G$, respectively. 
It is well-known that eigenvalues of normalized Laplacians are between $0$ and $2$. 
Also, observe that, for any graph $H$, $L_H \vec{1} = 0$. Therefore,  $\Lhat_G (\deg_G)^{1/2} = \Lhat_D (\deg_G)^{1/2} = 0$. That is, $(\deg_G)^{1/2}$ is in the kernel of both $\Lhat_G$ and $\Lhat_D$.

Let $\lambda$ be the second smallest eigenvalue of $\Lhat_H$. Then for any vector $\xx'\perp 
\left(\deg_G\right)^{\frac{1}{2}}$, we have:
\[
\frac{\lambda}{2}\xx'^\top \Lhat_D\xx'\le \lambda \lVert \xx'\rVert^2\le \xx'^\top \Lhat_G \xx',
\]

 since the largest eigenvalue of $\Lhat_D$ is at most $2$. 
This implies that, for every vector $\xx\in \reals^n$,  $\xx^\top \Lhat_G \xx \ge \frac{\lambda}{2}\xx^\top \Lhat_D\xx $ holds.
\newcommand{\xxbar}{\overline{\boldsymbol{\mathit{x}}}}
Indeed, we can write
\[
\xx = \xxbar +  c \left(\deg_G\right)^{\frac{1}{2}}
\]
where $\xxbar \perp (\deg_G)^{\frac{1}{2}}$ and $c$ is a scalar.
This gives:
\begin{align*}
\xx^{\top}\Lhat_{G}\xx
& =\left(\xxbar + c \left(\deg_G\right)^{\frac{1}{2}} \right)^{\top}
	\Lhat_{G} \left(\xxbar + c \left(\deg_G\right)^{\frac{1}{2}} \right)\\
& =\xxbar^{\top} \Lhat_{G} \xxbar\\
& \ge\frac{\lambda}{2}\cdot\xxbar^{\top}\Lhat_{D}\xxbar\\
& =\frac{\lambda}{2} \cdot \left(\xxbar + c \left(\deg_G\right)^{\frac{1}{2}} \right)^{\top}
	\Lhat_{D} \left(\xxbar + c \left(\deg_G\right)^{\frac{1}{2}} \right)\\
& = \frac{\lambda}{2}\cdot\xx^{\top}\Lhat_{D}\xx
\end{align*} 
where the last equality uses the fact that $\deg_G = \deg_D$.
By Cheeger's inequality \cite{Alon86}, we have $\lambda \ge \Phi(G)^2/2 \ge \phi^2/2$. Therefore, for any vector $\xx\in \reals^n$:
\begin{equation}
\xx^{\top}\Lhat_{G}\xx\ge\frac{\phi^{2}}{4}\xx^{\top}\Lhat_{D}\xx\label{eq:cheeger}
\end{equation}

We can now conclude that, for any vector $\xx\in \reals^n$: 
\begin{align*}
\xx^\top L_G \xx&
=\xx^\top W_G^{1/2} \Lhat_G W_G^{1/2}\xx\\
&\ge \frac{\phi^2}{4}\xx^\top W_G^{1/2} \Lhat_D W_G^{1/2}\xx \\
&= \frac{\phi^2}{4}\xx^\top W_G^{1/2}W_D^{-1/2} L_D W_D^{-1/2}W_G^{1/2}\xx\\
&=\frac{\phi^2}{4}\xx^\top L_D\xx
\end{align*}

where the first inequality follows by applying \ref{eq:cheeger} to vector $x'=W^{1/2}_Gx$, and 
 the last equality follows from the fact that $\deg_G = \deg_D$.
The proof that $\xx^\top L_D \xx \ge \frac{\phi^2}{4}\xx^\top L_H\xx$ is similar.

\end{proof}

Using \ref{lem:degree_approx} with $\phi=1/(\log m)^{O(r^2)}$ implies that $D$ is a $\left((\log m)^{O(r^2)}\right)^2=(\log m)^{O(r^2)}$-approximate spectral sparsifier of $H$. 
Finally, a spectral sparsifier for graph $D$ can be constructed in nearly linear time using the following lemma.

\begin{lem}[Lemma G.15, \cite{KyngLPSS16}]
\label{lem:approx_prod_demand}
There exists a deterministic algorithm that, given any demand vector $\dd\in \mathbb{R}^n$, computes, in time $O(n\epsilon^{-4})$, a graph $K$ with $O(n\epsilon^{-4})$ edges such that $e^{-\epsilon} K$ is an $e^{2\epsilon}$-approximate spectral sparsifier of $H(\dd)$.
\end{lem}

By letting $\epsilon=2$ and $\dd=\deg_D$ in \ref{lem:approx_prod_demand}, we obtain an $100$-approximate spectral sparsifier for graph $D$ (by scaling $K$), which is in turn a $(\log m)^{O(r^2)}$-approximate spectral sparsifier for graph $G$. %Each expander is preserved up to an approximation factor of $(\log m)^{O(r^2)}$,
By combining the spectral sparsifiers that we have computed for all sub-graphs of the original input graph $G$, we obtain an $(\log m)^{O(r^2)}$-approximate spectral sparsifier of the original graph $G$. The total number of edges in the sparsifier is $O(n\log n)$, as every level of the recursion contributes $O(n)$ edges. 

We now analyze the running time of the algorithm.
Since the depth of the recursion is $O(\log m)$, running~\ref{cor:expander decomp} takes $O\left ( m^{1+O(1/r)}\cdot (\log m)^{O(r^2)}  \right )$ time in total. Sparsifying the resulting expanders takes $O(m\polylog(m))$ time. Therefore, the overall running time is bounded by $O\left ( m^{1+O(1/r)}\cdot (\log m)^{O(r^2)}  \right )$.

For the general (weighted) case, it suffices to decompose the graph by the binary representations of the edge weights and sum the results up: For every edge $e\in E(G)$, let $\bb_e$ be the binary representation of the weight $w_e$. For all $1\leq i\leq \ceil{\log(\max_e\ww_e)}$, we construct an unweighted graph $G^{(i)}$, whose vertex set is $V$, and edge set contains every edge $e\in E(G)$, such that the $i$th bit of $\bb_e$ is $1$.  Since $\ww_e\le U$ for every $e\in E(G)$, there are at most $\lceil \log U\rceil$ such $G^{(i)}$s. By the algorithm for the unweighted case, we compute $(\log m)^{O(r^2)}$-approximate spectral sparsifiers for each $G^{(i)}$. The desired $(\log m)^{O(r^2)}$-approximate spectral sparsifier for $G$ is $\sum_{i=1}^{\ceil{\log(\max_e\ww_e)}}2^iG^{(i)}$. This sparsifier contains $\sum_{i=1}^{\ceil{\log(\max_e\ww_e)}} |E(G^{(i)})|=O(n\log n\log U)$ edges. The total running time is $O\left ( m^{1+O(1/r)}\cdot (\log m)^{O(r^2)}  \log U\right )$.
\end{proof}

%
%\subsection{$p$-norm Flow}
%
%{*}{*}ADD: Define, state, proof{*}{*}
%

\subsection{Laplacian Solvers and Laplacian-based Graph Algorithms}

The fastest previous deterministic Laplacian solver, due to Spielman and Teng \cite{SpielmanT03}, has running time $\Otil\left(m^{1.31}\log \frac{1}{\epsilon}\right)$. 
All faster solvers with near-linear running time are based on randomized spectral sparsifiers (e.g.~\cite{SpielmanT14-ThirdJournal}) or are inherently randomized \cite{KyngS16}. 
By applying the deterministic algorithm for computing spectral sparsifiers from \ref{cor:sparsifier}, we immediately obtain deterministic Laplacian solvers with almost linear running time.

Formally stating such results requires the definition of errors,
which are based on matrix norms.
For any matrix $A$, an $A$-norm of a vector $x$ is defined by
$\norm{\xx}_{A}=\sqrt{\xx^{\top}A\xx}$.
Let $A^{\dag}$ denote the Moore-Penrose pseudoinverse of $A$, which
is the matrix with the same nullspace as $A$ that acts as the inverse
of $A$ on its image.
\begin{cor}
\label{cor:laplacian}There is a deterministic algorithm that, given
a Laplacian $L$ size $n\times n$ with $m$ non-zeroes and a vector
$\bb\in\mathbb{R}^{n}$, computes a vector $\xx$ such that $\norm{\xx-L^{\dag}\bb}_{L}\le\epsilon||L^{\dag}\bb||_{L}$
in time $\Ohat\left(m\log\frac{1}{\epsilon}\right)$.
\end{cor}

This result follows because spectral sparsifiers are the only randomized components in the Spielman-Teng
Laplacians solvers~\cite{SpielmanT14-ThirdJournal}.
Although Spielman and Teng employ $(1+\epsilon)$-approximate spectral
sparsifiers in their solvers, by paying $n^{o(1)}$ factor in the running time, one can show that exactly the same approach works even if we use $n^{o(1)}$-approximate spectral sparsifiers from \ref{cor:sparsifier}.

There are several graph algorithms \cite{Madry10,Madry16,CohenMSV17}
based on interior point method that need to iteratively solve Laplacian
systems several times. In those algorithms, solving Laplacians is
the only randomized subroutine.
Therefore, the $\Ohat\left(m^{3/2} \log W\right)$ bound of interior point methods
for graph structured matrices by Daitch and Spielman~\cite{DaitchS08}
becomes deterministic.
This immediately implies algorithms for:
\begin{itemize}
        \item maximum flow in directed graphs with $m$ edges and edge capacities
        up to $W$,
        \item minimum-cost, and loss generalized flows in directed graphs with
        $m$ edges and edge capacities in $[0, W]$ and edge costs in $[-W, W]$,
\end{itemize}
that run in deterministic $\Ohat\left(m^{3/2}\log W\right)$ time. (See \cite{DaitchS08} for the discussion about the history of these problems.)
Furthermore, by derandomizing the interior-point-methods-based results from~\cite{Madry13,Madry16,CohenMSV17}, the following problems in directed $m$-edge graphs with edge costs/weights
in the range $[-W, W]$ can also be solved in deterministic $\Ohat\left(m^{10/7}\log W\right)$ time:
\begin{itemize}
\item unit-capacity maximum flow and maximum bipartite matching,
\item  single-source shortest path (with negative weight),\item minimum-cost bipartite perfect matching,
\item minimum-cost bipartite perfect $b$-matching, and
\item minimum-cost unit-capacity maximum flow.
\end{itemize}
A discussion about the history of these problems can be found in~\cite{CohenMSV17}.

%!TEX root = main_det_cut.tex

\subsection{Congestion Approximators and Approximate Maximum Flow}

In this subsection we discuss applications of our results to approximate maximum $s$-$t$ flow in undirected edge-capacitated graphs. Given an edge-capacitiated graph $G=(V,E)$, and a target flow value $b$, together with an accuracy parameter $0<\eps<1$, the goal is to either compute an $s$-$t$ flow of value at least $(1-\epsilon)b$, or to certify that the maximum $s$-$t$ flow value is less than $b$, by exhibiting an $s$-$t$ cut of capacity less than $b$. We note that the problem can equivalently be defined using a demand function $\bb: V\rightarrow \reals$ with $\sum_{v\in V}\bb_v=0$, by setting $\bb_s=-1$, $\bb_t=1$, and, for all $v\in V\setminus\set{s,t}$, $\bb_v=0$. %Given a cut $(S,\overline S)$ in $G$, we denote 
In general, given an arbitrary demand function $\bb: V\rightarrow \reals$ with $\sum_{v\in V}\bb_v=0$, we say that a flow $f$ \emph{satisfies} the demand $\bb$ iff, for every vertex $v\in V$, the \emph{excess flow at $v$}, which is the total amount of flow entering $v$ minus the total amount of flow leaving $v$, is precisely $\bb_v$.

The maximum $s$-$t$ flow problem is among the most basic and extensively studied problems. There are several  near-linear time randomized algorithms for computing $(1+\epsilon)$-approximate maximum flows \cite{Sherman13,KelnerLOS14,Peng16}; the fastest current randomized algorithm, due to Sherman \cite{Sherman17}, has running time $\Otil(m/\epsilon)$. Our results imply a deterministic algorithm for approximate maximum flow in undirected edge-capacitated graphs, with almost-linear running time.

%A $(1+\epsilon)$-approximate max flow in undirected edge-capacitated graphs can be computed in near-linear time by randomized algorithms \cite{Sherman13,KelnerLOS14,Peng16}. The currently fastest randomized algorithm takes  $\Otil(m/\epsilon)$ time and is by Sherman \cite{Sherman17}. Below, we present a $\Ohat(m/\epsilon^2)$-time deterministic algorithm. 
%Combining the congestion approximators with the area convexity based accelerated methods by Sherman~\cite{Sherman17} gives a deterministic, almost linear time, approximate max-flow algorithm with $1 / \epsilon$ dependencies.
\begin{cor}
\label{cor:maxflow}
There is a deterministic algorithm that, given an $m$-edge connected graph $G$ with capacities $\cc_e\geq 0$ on edges $e\in E$, such that $\frac{\max_{e}\cc_{e}}{\min_{e}\cc_{e}}\leq O(\poly(m))$, a demand function $\bb\in\mathbb{R}^{V}$  with  $\sum_{v\in V}\bb_{v}=0$,  and an integer $1\le r\le O(\log m)$, and an accuracy parameter $0<\epsilon\leq 1$, computes, in time $T_{\Maxflow}(m,\epsilon)=O(m^{1+O(1/r)}(\log m)^{O(r^3)}\epsilon^{-2})$,
either:\thatchaphol{TODO: write a precise dependency on $\log U$}
\begin{itemize}
\item \textbf{(Flow): }a flow satisfying the demand $\bb$ with $\abs{f(e)}\le (1+\epsilon)\cc_e$ for every edge $e$; or
\item \textbf{(Cut):} a cut $S$ such that $\sum_{e\in E(S,\overline{S})}\cc_{e}<\abs{\sum_{v\in S}\bb_{v}}$.
\end{itemize}
In particular, choosing $r \leftarrow (\log m)^{1/4} (\log\log{m})^{-3/4}$
gives a total running time of
\[
m \cdot \exp\left(O\left( \left( \log{m} \log\log{m}\right)^{3/4} \right) \right) \epsilon^{-2}
<
O\left( m^{1 + o\left(1\right)} \epsilon^{-2} \right).
\]
\end{cor}

%We first use \ref{cor:sparsifier} to sparsify $G$ into $G'$ with $O(n\log ^2 n)$ edges. $G'$ is an $n^{o(1)}$-approximate spectral sparsifier construct deterministically a oblivious routing scheme 
%We explain the proof of this application in more detail because it will be used later for our balanced cut algorithm in the low conductance regime.

Our proof of \ref{cor:maxflow} closely follows the algorithm of~\cite{Sherman13}, and proceeds by constructing a congestion approximator. Unlike the algorithm in~\cite{Sherman13}, our algorithm for computing the congestion approximator is deterministic, and is obtained by replacing a randomized procedure in~\cite{Sherman13} for constructing a cut sparsifier with the deterministic algorithm from \ref{cor:sparsifier}. We then use the reduction from $(1+\eps)$-approximate maximum flow to congestion approximators by~\cite{Sherman13}. 

Suppose we are given a graph $G=(V,E)$, where $E=\set{e_1,\ldots,e_m}$, with capacities (or weights) $\cc_e>0$ on edges $e\in E$. Let $C$ be the diagonal $(m\times m)$ matrix, such that for all $1\leq i\leq m$, the entry $(i,i)$ of the matrix is $\cc_{e_i}$. Assume now that we are given a demand vector $\bb:V\rightarrow \reals$, with $\sum_{v\in V}\bb_v=0$. Let $f$ be any flow that satisfies the demand $b$. The \emph{congestion} $\eta$ of this flow is the maximum, over all edges $e\in E$, of $f(e)/\cc_e$. Equivalently, $\eta=\norm{C^{-1}f}_{\infty}$. By scaling flow $f$ by factor $\eta$, we obtain a valid flow routing demand $\bb/\eta$. Therefore, the maximum $s$-$t$ flow problem is equivalent to the following problem: given a graph $G=(V,E)$ with capacities $\cc_e$ on edges $e\in E$ and a demand vector $\bb:V\rightarrow \reals$, with $\sum_{v\in V}\bb_v=0$, compute a flow $f$ satisfying the demand $\bb$, while minimizing the congestion  $\eta=\norm{C^{-1}f}_{\infty}$ among all such flows.
For convenience, we define an \emph{incidence} matrix $B$ associated with graph $G$. We direct the edges of graph $G$ arbitrarily. Matrix $B$ is an $(m\times n)$ matrix, whose rows are indexed by edges and columns by vertices of $G$. Entry $(e_i,v_j)$ is $-1$ if $e_i$ is an edge that leaves $v_j$, it is $1$ if $e_i$ is an edge that enters $v_j$, and it is $0$ otherwise. Notice that, given any flow $f$, the $j$th entry of $B\cdot f$ is the \emph{excess flow} on vertex $v_j$: the total amount of flow entering $v_j$ minus the total amount of flow leaving $v_j$. If flow $f$ satisfies a demand vector $\bb$, then $B\cdot f=\bb$ must hold.
Next, we recall the definition of congestion approximators from \cite{Sherman13}.

\begin{defn}
Let $G$ be a graph with $n$ vertices and $m$ edges, let $C$ be the diagonal matrix containing of edge capacities and let $B$ be the $n\times m$ incidence matrix. An $\alpha$-congestion approximator for $G$ is a matrix $R$ that contains $n$ columns and an arbitrary number of rows, such that for any demand vector $\bb$, 
\[
\norm{R\bb}_{\infty}\le \mathsf{opt}_G(\bb)\le \alpha\norm{R\bb}_{\infty},
\] where $\mathsf{opt}_G(\bb)$ is the value of the optimal solution for the minimum-congestion flow problem
\[
\min\ \norm{C^{-1}f}_{\infty} s.t.\ Bf=\bb.
\]
\end{defn}

The main goal of this section is to prove the following lemma, that provides a deterministic construction of congestion approximators, and is used to replace its randomized counterpart, Theorem 1.5 in \cite{Sherman13}. We note that we obtain somewhat weaker parameters in the approximation factor and the running time.

\begin{lem}[Deterministic version of Theorem 1.5,~\cite{Sherman13}]
\label{lem:det15sherman13}
%Let $f(x)$ be a family of functions that are bounded by
%$x\cdot 2^{O(\log (x)\log\log (x))^{2/3}}$. $O\left(f(m)+f(n)\right)$.
There is a deterministic algorithm, that we call $\mathtt{CongestionApproximator}(G, r)$, that, given as input an $n$-vertex $m$-edge graph $G$ with capacity ratio $U$ and an integer $1\le r\le O(\log m)$, constructs a $(\log m)^{O(r^3)}$-congestion
approximator $R$, in  time $O(m^{1+O(1/r)} \log^{O(r^3)} (mU))$.
Once constructed, we can compute a multiplication of $R$ and of $R^\top$ by a vector, in time $O(m^{1+O(1/r)} \log^{O(r^3)} (mU))$ each.
\end{lem}

Given a cut $(S,\overline S)$ in an edge-capacitated graph $G=(V, E, \cc)$, we denote $\cc_{S}=\sum_{e\in E(S, \overline{S})}\cc_e$. Given a demand vector $\bb$ for $G$, we also denote $\bb_S=\abs{\sum_{v\in S}\bb_v}$.

Sherman \cite{Sherman13} provides the following method for turning an
efficient congestion approximator into an algorithm for computing approximate maximum flow. \jl{I think Sherman's Theorem 1.2 has a typo, and it should be $\alpha^2/\eps^2$, not $\alpha/\eps^2$? At least, MWU gives $\alpha^2/\eps^2$. It doesn't change anything, of course.}

\begin{lem}[Theorem 2.1~\cite{Sherman13}]
        \label{lem:almost} There is a deterministic algorithm that, given a graph $G=(V,E)$ with edge weights $\cc_e$ for $e\in E$ and a  demand vector $\bb: V\rightarrow \reals$ with $\sum_{v\in V}\bb_v=0$, together with an  
        access to an $\alpha$-congestion-approximator $R$ of $G$, makes $\tilde{O}(\alpha^2 \eps^{-2})$ iterations, and returns a flow $f$ and cut $S$ in $G$, with $Bf = \bb$ and $\norm{C^{-1}f}_{\infty} \le (1 + \eps)\bb_S/\cc_S$. Each iteration requires $O(m)$ time, plus time needed to multiply a vector by $R$, and time needed to multiply a vector by $R^\top$.
\end{lem}

Note that combining~\ref{lem:det15sherman13} with~\ref{lem:almost} immediately yields a proof of~\ref{cor:maxflow}. Indeed, if the algorithm from~\ref{lem:almost} outputs a cut $S$ with $\bb_S>\cc_S$, then we return the cut $S$ as the outcome of our algorithm for~\ref{cor:maxflow}. Otherwise, $\bb_S/\cc_S\le 1$ must hold, and so the flow $f$ output by~\ref{lem:almost} satisfies $\norm{C^{-1}f}_{\infty}\le 1+\epsilon$. We then return $f$ as the outcome of~\ref{cor:maxflow}. 
The algorithm from \ref{lem:almost}  performs $\tilde O(\alpha^2/\eps^2)$ iterations,\jl{also changed this to $\alpha^2/\eps^2$} where $\alpha=\log ^{O(r^3)}m$. Every iteration takes time $O(m^{1+O(1/r)}\log ^{O(r^3)}m)$, since $U=\poly(m)$ by assumption. Therefore, the total running time
of the algorithm from \ref{lem:almost} is at most $O\left (m^{1+O(1/r)}\log ^{O(r^3)}m/\eps^2\right )$. So the total running time of the algorithm is bounded by $O\left (m^{1+O(1/r)}\log ^{O(r^3)}m/\eps^2\right )$.

In order to complete the proof of \ref{cor:maxflow}, it is now enough to prove \ref{lem:det15sherman13}.

\subsection*{Proof of \ref{lem:det15sherman13}}

%\jnote{Please check if the next paragraph is correct. If not please explain how $R$ is constructed from the cuts}

\iffalse
\jnote{I think the following is not quite correct because we actually use cuts in the sparsifier and not in the graph $G$.}
\takeout{
Instead of constructing the matrix $R$ directly, we construct a collection $\cset=\set{(A_i,B_i)}_i$ of cuts in graph $G$. The corresponding matrix $R$ will then contain a row for each cut $(A_i,B_i)$, where for each $1\leq j\leq n$, the $j$th entry of the row corresponding to cut $(A_i,B_i)$ is $1/\cc_{A_i}$ if vertex $v_j\in A_i$, and $0$ otherwise. Note that for any vector $\bb\in \reals^n$, the value of the $i$th entry of $R\cdot \bb$ is either $\bb_{A_i}/\cc_{A_i}$ or $-\bb_{A_i}/\cc_{A_i}$. Therefore, from the maximum-flow / minimum-cut theorem, we are guaranteed that, if $f$ is a flow satisfying $B$, that is, $Bf=\bb$, then $\norm{R\bb}_{\infty}\leq \norm{C^{-1}f}_{\infty}$ (recall that $ \norm{C^{-1}f}_{\infty}$ is the congestion caused by flow $f$). Therefore, our goal is to define the set $\cset$ of cuts, such that, on the one hand, $|\cset|$ is small, and on the other hand, there exists some cut $(A_i,B_i)\in \cset$, such that $\bb_{A_i}/\cc_{A_i}\geq \mathsf{opt}_G(\bb)/\alpha$.}

\jnote{Replace previous paragraph with amended paragraph below?}

\fi
Instead of constructing the matrix $R$ directly, we (implicitly) construct a graph $H$  with $V(H)=V(G)$, that can be thought of as a cut sparsifier for $G$, and then construct a collection $\cset=\set{(A_i,B_i)}_i$ of cuts in graph $H$. The corresponding matrix $R$ will then contain a row for each cut $(A_i,B_i)$, where for each $1\leq j\leq n$, the $j$th entry of the row corresponding to cut $(A_i,B_i)$ is $1/\cc_{A_i}$ if vertex $v_j\in A_i$, and $0$ otherwise (the values $\cc_{A_i}$ are defined with respect to the sparsifier $H$). Note that for any vector $\bb\in \reals^n$, the value of the $i$th entry of $R\cdot \bb$ is either $\bb_{A_i}/\cc_{A_i}$ or $-\bb_{A_i}/\cc_{A_i}$. Therefore, from the maximum-flow / minimum-cut theorem, we are guaranteed that, if $f$ is a flow satisfying $\bb$ in graph $H$, then $\norm{R\bb}_{\infty}\leq \norm{C^{-1}f}_{\infty}$ (recall that $ \norm{C^{-1}f}_{\infty}$ is the congestion caused by flow $f$). Therefore, our goal is to define the set $\cset$ of cuts in $H$ such that, on the one hand, $|\cset|$ is small, and on the other hand, there exists some cut $(A_i,B_i)\in \cset$, such that $\bb_{A_i}/\cc_{A_i}\geq \mathsf{opt}_H(\bb)/\alpha$. The sparsifier $H$ is in fact a convex combination of a collection $\mathcal{F}$ of spanning trees of $G$, and the cuts in $\cset$ are the cuts defined by these trees (that is, for each tree $T\in \mathcal{F}$, for every edge $e$ of $T$, we add the cut defined by $T\setminus \set{e}$ to $\cset$). As the total number of all such cuts is large, the matrix $R$ itself is also large, so we cannot afford to construct it explicitly. Instead, the recursive procedure that we use in order to construct the collection $\mathcal{F}$ of trees can also be employed in order to efficiently compute a multiplication of $R$ and of $R^{\top}$ by a vector. We note that this algorithm is only a slight modification of the algorithm of Sherman~\cite{Sherman13}. 

The construction of the cut sparsifier $H$ for $G$ and the corresponding collection $\mathcal{F}$ of trees is a direct modification of the $j$-tree based construction of cut approximators by Madry~\cite{Madry10-jtree}.

\iffalse

\jnote{I don't think we need the following if we use the above 2 paragraphs}
This view takes the dual view of maximum flow:

\begin{fact}(max-flow / min-cut)
In a graph $G = (V, E)$ with capacities $\cc$,
the minimum congestion of routing a demand vector $\bb$ is given by
\[
\max_{S \subseteq V} \frac{\abs{\bb_{S}}}{\cc_{S}}.
\]
\end{fact}
As $\cc_{S}$ are fixed by the choice of the graph $G$,
this view says that for an exact congestion approximator 
we can let the rows of $R$ be indicator vectors for the vertices in $S$.
Also, note that because $\bb$ is a valid demand, we have $\bb_{S}
= -\bb_{V \setminus S}$.
%\yp{added why $R$ correpsonds to cuts here.} \jnote{I have added a more detailed explanation above, please see if it's OK}

Madry's congestion approximator construction~\cite{Madry10-jtree}
finds a family of such cuts where the demands across them can
be efficiently computed: note that we cannot afford to look at
all $n$ cuts that form the rows of $R$ due to the total number
of vertices involved being potentially as large as $\Theta(n^2)$.

\jnote{End of part we no longer need.}
\fi

The algorithm of~\cite{Madry10-jtree} constructs the cut sparsifier by gradually reducing $G$ to (several) tree-like
objects, called $j$-trees. Each such object consists of a collection of disjoint trees and another relatively small graph called a core. 
We can then consider all cuts defined by the edges of the tree, and then recursively sparsify the core after contracting the trees, in order to consider the cuts that partition the core. 
%and gradually contract out the tree-like parts to keep
%the total sizes of the graphs small.
In order to accomplish this, the algorithm alternates between reducing the number of vertices and the number of edges in a graph whose cuts we need to approximate. The former is achieved by routing along adaptively generated low-stretch spanning tree, while the latter uses a randomized algorithm 
 of Benczur and Karger \cite{BK}. 

\iffalse
he needs to alternatively
route along adaptively generated low-stretch spanning trees to reduce
the number of vertices,
and reduce the number of edges by calling graph sparsifiers.
For the latter, sparsifiers as given in~\ref{cor:sparsifier} suffice,
while the former rely on routing along trees.
\fi

Our modifications of the constructions of Madry~\cite{Madry10-jtree}
and Sherman~\cite{Sherman13} consist of three main components:
\begin{enumerate}
        \item observing that the routing along the trees is entirely deterministic;
        \item showing that the randomized  algorithm of  \cite{BK} for computing  cut sparsifiers can be directly replaced by its deterministic counterpart from~\ref{cor:sparsifier} (though with a larger approximation factor); and
        \item showing that, instead of sampling the distribution over partial routings,
        we can recursively construct congestion approximators for each of them, and analyzing the total cost of the
        recursion. 
\end{enumerate}

We start with several definitions that we need.

\begin{defn} [Embeddings of weighted graphs]
        Let $G$, $H$ be two graphs with $V(G)= V(H)$. Assume also that we are given edge weights $\ww_e$ for all edges $e\in E(
H)$ be $\ww_e$. An \emph{embedding} of $H$ into $G$ is a collection $\pset=\set{P(e)\mid e\in E(H)}$ of paths in $G$, such that
 for each edge $e\in E(H)$, path $P(e)$ connects the endpoints of $e$ in $G$. We say that the embedding causes \emph{congestion
} $\cong$ iff for every edge $e'\in E(G)$: \[\sum_{\stackrel{e\in E(H):}{e'\in P(e)}}\ww_e\le \cong.\]
\end{defn}

We also need the following definitions for composition of graphs
and cut domination.

\begin{defn}
        Let $G=(V, E, \ww)$ and $H=(V, E', \ww')$ be two edge-weighted graphs defined over the same vertex set. We define their
 composition $K=G+H$ be an edge-weighted graph $K=(V, F, \zz)$, where $F$ is the disjoint union of the edge sets $E$ and $E'$,
and for every edge $e\in F$, its weight is defined to be $\zz_e=\ww_e$ if $e\in E$ and $\zz_e=\ww'_e$ if $e\in E'$.

        Given an edge-weighted graph $G=(V, E, \ww)$ and a scalar $\alpha$, we let $\alpha G$ be the edge-weighted graph $(V, E
, \alpha \ww)$.
\end{defn}

\begin{defn}
        Suppose we are given two edge-weighted graphs $G=(V,E,\ww)$ and $H=(V,E',\ww')$ that are defined over the same vertex s
et. We say that \emph{$G$ cut-dominates $H$}, and denote $G\geq_c H$ iff for every partition $(S,\overline S)$ of $V$, $\sum_{e
\in E_G{S,\overline S}}\ww_e\geq \sum_{e\in E_H{S,\overline S}}\ww_e'$.
\end{defn}

We use the following definition of $j$-trees of \cite{Madry10-jtree}: %  below Figure 2 of~\cite{Madry10-jtree}.

\begin{defn} \label{def:jtree}
        A graph $H$ is a $j$-tree if it is a union of:
        \begin{itemize}
                \item a subgraph $H'$ of $H$ (called the \emph{core}), induced by a set $V_{H'}$ of at most $j$ vertices; and
                \item a forest (that we refer to as \emph{peripheral forest}), where each connected component of the forest contains exactly one vertex of $V_{H'}$.
                For each core vertex $v \in V_{H'}$, we let $T_H(v)$ denote the unique
                tree in the peripheral forest that contains $v$. When the $j$-tree $H$ is 
                                unambiguous, we may use $T(v)$ instead.
        \end{itemize}
\end{defn}

\iffalse
\jnote{I could not understand this. Do you mean that each edge is embedded or something? what does it mean to move off-tree edges?}\yg{The edges in  the original graph are moved. Is it less confusing if we remove Madry's construction as we use it as a blackbox?}\jnote{since we are using $j$-trees as a blackbox I think we shouldn't include this intuitive description. The lemma seems self-explanatory}
At a high level, Madry constructs $j$-trees by finding a \jl{low-stretch?} spanning
tree, and then moving the off-tree edges to
a suitably identified subset of vertices.
This incurs congestion along the tree edges, and this congestion
is in turn distributed over the trees from the distribution via
a multiplicative update scheme simliar to the construction of nearly-optimal
oblivious routings by Racke~\cite{Racke08}.
Such guarantee can be encapsulated by the following rephrasing of
Theorem 3.6 from~\cite{Madry10-jtree}.
\fi

We use the following theorem, which is a restatement of Theorem 3.6 from~\cite{Madry10-jtree}.

\begin{lem}[\cite{Madry10-jtree}] \label{thm:madrythm}
        There is a deterministic algorithm that, given an edge-weighted graph $G=(V, E, \ww)$ with $|E|=m$ and capacity ratio $U=\frac{\max_{e\in E}\ww_e}{\max_{e\in E}\ww_e}$, together with a parameter $t \geq 1$, computes, in time $\tilde{O}(tm)$, a distribution $\set{\lambda_i}_{i=1}^{t}$ over a collection of $t$ edge-weighted graphs $G_1,\ldots,G_{t}$, where for each $1\leq i\leq t$, $G_i=(V, E_i, \ww_i)$, and the following hold:
        \begin{itemize}
                \item for all $1\leq i\leq t$, graph $G_i$ is an $(\frac{m\log^{O(1)}m\log U}t)$-tree, whose core contains at most $m$ edges;
                \item for all $1\leq i\leq t$, $G$ embeds into $G_i$ with congestion $1$; and %, i.e. for each cut $C\subseteq V$, $\sum_{e\in E_i(C, \overline{C})}(\ww_i)_e\ge \sum_{e\in E(C, \overline{C})}\ww_e$.
                %\thatchaphol{Either define dominate, or just write the mathematical definition here.}
                \item the graph that's the average of these graphs over the distribution,
                $\tilde G= \sum_{i} \lambda_i G_i$ can be embedded into $G$ with congestion $O(\log m (\log\log m)^{O(1)})$.
        \end{itemize}
                Moreover, the capacity ratio of each $G_i$ is at most $O(mU)$. 
\end{lem}

We remark that Madry's algorithm calls \emph{low-stretch spanning trees} as a black-box, and is deterministic outside of the low-stretch spanning tree algorithm. Madry calls the fastest algorithm at the time of his paper~\cite{AbrahamN08} which is randomized, but the more recent algorithm of~\cite{AbrahamN12} is deterministic and even produces better bounds, so we can simply use that instead and keep it deterministic.

 Notice that the distribution over the $j$-trees from the above theorem essentially provides a construction of a cut sparsifier $\tilde G= \sum_{i} \lambda_i G_i$ for graph $G$. Next, we show that it is sufficient to consider two types of cuts in this sparsifier: cuts defined by the edges of the trees in the $j$-trees $G_1,\ldots,G_{\tilde t}$, and cuts that are obtained by first contracting each tree of a graph $G_i$, and then partitioning its core. This observation was also shown by Madry in Theorem 4.1 and Lemma 6.1 of ~\cite{Madry10-jtree}. In order to do so, we need to define ``truncated'' versions of the demand vector, which we do next.

\iffalse

This distribution of $O(m / t)$-trees essentially allows us reduce the cuts that we need to consider in order to construct the congestion approximator to the following cuts:
\begin{itemize}
        \item cuts obtained by removing a single edge from a tree in the peripheral forest; and
        \item cuts partitioning the set $V(H')$ of $O(m/t)$ vertices of the core,
        with all the tree/forest edges omitted.
\end{itemize}

Madry ~\cite{Madry10-jtree} proves this formally in  Theorem 4.1 and Lemma 6.1 
of \cite{Madry10-jtree}.

This is formally proven by first showing that the embedability conditions
gives that checking the capcaities of these $j$-trees across all cuts gives
an $O(\log{n})$-approximation to the optimum congestion of the max-flow.
This is shown in the proofs of Theorem 4.1 and Lemma 6.1.
of~\cite{Madry10-jtree}.

\jnote{I don't understand this. Checking the trees is not enough, because we need to check the core as well? So something is missing from the explanation below?}
Specifically, these lemmas show that a $j$-tree sampled from the distribution
with probability proportional to $\lambda_{i}$ gives a good bound to the
optimum with constant probability.
That in turn implies checking through all the $O(m / t)$-trees in the
collection ensures that we examine that instance as well.

Next, we show that it suffices to check all core cuts,
plus all tree edge cuts in each of the $j$-trees.
\fi

\begin{defn}
\label{defn:DemandMigrate}
If $H = (V_H, E_H, \cc_{H})$ is a $j$-tree with core $V_{H'}$ and $\bb$ is a demand vector, 
 \begin{itemize}
                \item for each peripheral tree $T(v)$ (with $v \in V_{H'}$),
                the demand vector $\bb$ truncated to the tree, $\bb_{T(v)}\in \mathbb{R}^{V_{T(v)}}$, is defined as:
                \[
                \bb_{T\left(v\right), w}
                \defeq
                \begin{cases}
                        \bb_{w} & \qquad \text{if $w \neq v$}\\
                        -\sum_{u \in T\left( v \right), u \neq v}\bb_u & \text{otherwise}.
                \end{cases}
                \]
                \item the demand vector on $V_{H'}$ that's the sum
                of $\bb$ over the corresponding trees, $\bb_{H'}$, is defined as 
                \[
                        \bb_{H', v}
                        \defeq
                        \sum_{w \in V\left( T \left( v \right) \right)} \bb_{w}
                                \qquad \forall v \in V_{H'}.
                \]
                
        \end{itemize}
\end{defn}

The following lemma allows us to restrict our attention to only a small number of cuts in a $j$-tree.

\begin{lem}
\label{lem:route_in_core}
        Let $H = (V_H, E_H, \cc_{H})$ be a $j$-tree with core $V_{H'}$ 
        and $\bb$ be a demand vector. If the following two conditions hold, 
 \begin{itemize}
                \item for each peripheral tree $T(v)$ (with $v \in V_{H'}$),
                $\bb_{T\left(v\right)}$ can be routed on $T(v)$ with congestion at most $\rho$; and
                \item $\bb_{H'}$ can be routed in $H'$ with congestion at most $\rho$,
        \end{itemize}
        then $\bb$ can be routed in $H$ with congestion at most $\rho$ as well.
\end{lem}

This lemma is phrased in terms of routings due to max-flow/min-cut.
Note that all the cuts in the peripheral trees $T(v)$
and the core graph $H'$ are valid cuts to be considered in $H$ as well.
So checking the congestions of these two routings on the peripheral
trees and core graphs provide
a lower bound on the congestion needed to route $\bb$ as well.

\begin{proof}
        This follows from the fact that the minimum cut must have both pieces connected:
        if a tree edge is disconnected, then one of the pieces must fall entirely
        within one of the peripheral trees, and is captured by the tree cuts.
        Otherwise the only edges cut are within the core graph, and the cut is also
        one of the cuts on $H'$.
        
        Alternatively, we can use a flow based proof where we route $\bb$.
        Consider routing $\bb$ in two stages:
        first, all the demand in each tree moves to its root.
        The congestion of this is exactly the max congestion of a
        tree edge, since the flow across a tree is uniquely determined
        by the vertex demands.
        Then we route all the demands at the root vertices
        across the core graph.
\end{proof}

%\thatchaphol{TODO: explain more the output of \Cref{lem:almost} gives gives approximate max flow in \Cref{cor:maxflow}.}
%\thatchaphol{TODO: compute the subpolynomial factor in the running time because it will be used later in other section.}
%\yu{This requires the running time of~\ref{cor:sparsifier} which requires that of~\ref{cor:expander decomp}}

We also verify that all the cuts checked on the trees and cores
are also valid cuts in $G$, and thus what we find is also a lower
bound to the optimum congestion as well.

\begin{lem}
        \label{lem:core_lower_bound}
        Let $H = (V_H, E_H, \cc_{H})$ be a $j$-tree with core $V_{H'}$ and 
        and $\bb$ be a demand vector. 
        Then the demand vector migrated onto $H'$ as given
        in \ref{defn:DemandMigrate} satisfies
        \[
        \opt_{H'}\left(\bb_{H'}\right)
        \le
        \opt_H\left(\bb\right).
        \]
\end{lem}
\begin{proof}
        Take the cut based interpretation of min-cut,
        where we maximize over the congestion of cut.
        
        For every cut on $V_{H'}$, we can extend it to a cut in $G$
        by putting, for every core vertex $v \in V_{H'}$,
        $T(v)$ on the same side of $v$ in the cut.
        
        Because each $T(v)$ is on the same side of the cut,
        the total number of edges cut is unchanged.
        Furthermore, the sum of the demand $\bb$s on the side
        of the cut is unchanged by the way we constructed $\bb_{H', v}$.
        Thus, the set of cuts we examined on $V_{H'}$ is a subset of
        the cuts of $G$, which means the max congested cut in $G$
        has a higher or equal value.
\end{proof}

Thus, we can recurse on all core graphs after sparsifying them.
This leads to the algorithm $\mathtt{CongestionApproximator}$
whose pseudocode is given~\ref{alg:CongestionApproximator}.
For sake of brevity, we do not explicitly define the functions
needed to apply $R$ and $R^{\top}$ (i.e. to compute $R\bb$ and $R^{\top}\yy$ for vectors $\bb$ and $\yy$), but only implicitly
describe how they are computed together with this recursion.

\begin{algorithm}
        $\mathtt{CongestionApproximator}(G=(V_G, E_G, \cc_G), t, r):$
        \begin{itemize}
                \item Implicitly append a row corresponding to the cut $\{u\}$ for the only edge $uv$ if exists. Return.
                
                \item Using theorem \ref{thm:madrythm} with parameter $t=\Theta(m^{1/r}\log^{O(1)}m\log ^{2}U)$,
                compute distribution $(\lambda_i, H_i)_{i=1}^{t}$ of $\max(1, \frac{m\log^{O(1)}m\log ^{2}U}t)$-trees.  
                \item For $i = 1 \ldots t$:
                \begin{itemize}
                        \item For each edge $e$ in the forest of $H_i$:
                        \begin{itemize}
                        \item Let $S$ be the set of vertices (which form a subtree) that are disconnected with the core of $H_i$ when cutting $e$.
                        Let $\cc_{H_i}$ be the capacity vector of $H_i$. Implicitly append a row that measures $\bb_{S} / \cc_{H_i, S}=\bb_{S}/\cc_{H_i, e}$ to 
                        the congestion approximator $R$:
                        \begin{itemize}
                                \item This corresponds to a row $\rr^{\top}=\chi_S/\cc_{H_i, S}$, the indicator vector of vertices in $S$ times $1 / \cc_{H_i, S}$.
                                \item For a demand vector $\bb$, the corresponding row
                                in $R \bb$, 
                                $\rr^\top\bb$ can be computed by data structures that compute sum of values in a subtree in $O(\log n)$ time. \jnote{why do we need to compute $\rr^\top\bb$? Is it for some general vector $\bb$ or is it the original demand vector?}
                                \item For computing $R^\top \yy$ for some vector $\yy$, the new row contributes $\rr \yy_j$ to the result. %for some entry $y$ in the dual \jnote{may be good to explain what you mean by dual. Also you need to multiply all entries by $y$, not add some value?} 
                                This can be computed by adding the value $1 / \cc_{H_i, S}$
                                to all nodes in a subtree, which also
                                takes $O(\log{n})$ time using tree data structures.
                        \end{itemize}
                        \end{itemize}
                        \item Let $H'_i$ be the core graph of $H_i$ (with edges in the forest of $H_i$ removed). Set up mappings from rows of congestion approximator on $V_{H'_i}$ to rows of congestion approximator on $G$:
                        %Set up mappings of demands / dual varialbes to and from $V_{H'_i}$ respectively:

                        \begin{itemize}
                                \item A row $\tilde{\rr}$ of congestion approximator $R$ on $V_{H'_i}$ is mapped to a row $\rr$ on $V_{G}$ by duplicating the value on $v$ to all vertices in $T(v)$, i.e. $\rr_u=\tilde{\rr}_v$ for all $u\in T(v)$. This implies the following mappings of demands / dual variables to and from $V_{H'_i}$: 
                                \item For each $u \in T(v)$, $\bb_{u}$ gets added to $\bb_{v}$. We call the new vector on $V_{H'_i}$ $\bb_{H_i'}$.
                                \item Any congestion approximator $R'$ on $V_{H'_i}$,
                                will be extended back to all vertices by
                                duplicating the value at vertex $v$,        
                                $({R'}^{\top} \yy)_v$ to all vertices in $T(v)$.
                                This is equivalent to extending every $\rr'$ by duplicating
                                $\rr'_{v}$ to all entries in $T(v)$.
                        \end{itemize}
                        
                        \item Set $\tilde{H}'_i \leftarrow \mathtt{Sparsify}(H'_i, r)$, where
                                $\mathtt{Sparsify}$ was introduced in~\ref{cor:sparsifier}.
                        \item Recursively call $\mathtt{CongestionApproximator}(\tilde{H}_i', t, r)$. The rows implicitly created by this call is implicitly mapped back by the mapping above. 
                \end{itemize}
        \end{itemize}
\caption{Pseudocode for Constructing Congestion Approximator}
\label{alg:CongestionApproximator}
\end{algorithm}

We analyze the correctness of the congestion approximator
produced, and the overall running time, below.

\begin{proofof}{\Cref{lem:det15sherman13}}
        Let $t=\ceil{m^{1/r}}$ and run $\mathtt{CongestionApproximator(G, t, r)}$.
        We may assume $m=t^k$ by padding edges and prove the bounds by induction on $k$. \jnote{is it OK to pad the edges? won't this change the cuts?}
        When $k$ is $0$, the statement is true as $R$ is an $1$-congestion approximator of $G$. Assume $k\neq 0$.
        We first bound the quality of the congestion approximator
        produced. We will show 
        \[
        \norm{R\bb}_{\infty}(\log m)^{-kO(r^2)}\le \opt_G(\bb) \le (\log m)^{kO(r^2)}\norm{R\bb}_{\infty}
        \] for every integer $k\le r$. In the top level, $k$ is equal to $r$. We scale $R$ up by $(\log m)^{O(r^3)}$ in the top level to meet the definition of congestion approximator. 
        
        We first begin by showing that $\norm{R\bb}_{\infty}\le \opt_G(\bb)(\log m)^{kO(r^2)}$.
        
        For each cut $S$ corresponding to some forest edge $e$ with capacity $\cc_{H_i, e}$: We have 

        \[
        \abs{\chi_S^\top/\cc_{H_i, e}\bb}=\bb_S/\cc_{H_i, e}\le \bb_S/\cc_{G, S} \le \opt_G(\bb),
        \] where $\chi_S$ is the indicator vector of vertices in $S$ and $\cc_G$ is the capacity vector of $G$,
        and $\chi_S$ is the indicator vector for $S$ with $1$ at
        all $u \in S$ and $0$ everywhere else.
        
        Recall that $\tilde{H}'_i = \mathtt{Sparsify}(H'_i, r)$ is a spectral sparsifier of $H'_i$.
        Let $\tilde{\rr}$ be any row of the congestion approximator computed by the recursive call $\mathtt{CongestionApproximator}(\tilde{H}_i', k)$ (i.e. one row implicitly added in that recursive call). By~\ref{cor:sparsifier}, $\tilde{H}_i'$ has at most $m/\Omega(t\log m\log U)\cdot O(\log m\log U)\le m/t$ edges. By inductive hypothesis, $\abs{\tilde{\rr}\bb_{H_i'}}\le (\log m)^{(k-1)O(r^2)}\opt_{\tilde{H}_i'}(\bb_{H_i'})$. As $\tilde{\rr}$ is mapped to $\rr$ by duplicating the value on $v$ to all vertices in $T(v)$, $\abs{\rr\bb}=\abs{\tilde{\rr}\bb_{H_i'}}$. Since $\tilde{H}_i'\le_c (\log m)^{O(r^2)} H_i'$, by the multicommodity max-flow/min-cut theorem~\cite{LeightonR99},\jl{what is $\le_c$? I don't see it defined anywhere. I assume $G\le_c H$ means all cuts in $G$ are less than their corresponding cuts in $H$? In that case, why multicommodity? isn't this just (single-commodity) maxflow/mincut which is exact?}\jnote{Def 6.16, all cut in $G$ are less than the cuts in $H$. Probably the definition should be moved closer to here because it's not used before.} $\opt_{\tilde{H}_i'}(\bb_{H_i'}) \le (\log m)^{O(r^2)}(\log n)\opt_{H_i'}(\bb_{H_i'})=(\log m)^{O(r^2)}\opt_{H_i'}(\bb_{H_i'})$.  By \ref{lem:core_lower_bound}, $\opt_{H_i'}(\bb_{H_i'}) \le \opt_{H_i}(\bb)$. Since $G$ embeds into $H_i$, $\opt_{H_i}(\bb)\le \opt_G(\bb)$. Thus, 
        \begin{multline*}
                \abs{\rr\bb}
                =
                \abs{\tilde{\rr}\bb^{\left(i\right)}}
                \le
                \left(\log m\right)^{\left(k-1\right)O\left(r^2\right)}
                \opt_{\tilde{H}_i'}(\bb^{\left(i\right)})
                \le
                \left(\log m\right)^{kO\left(r^2\right)}\opt_{H_i'}
                        \left(\bb^{\left(i\right)}\right)\\
                \le
                \left(\log m\right)^{kO\left(r^2\right)}
                        \opt_{H_i}\left(\bb\right)
                \le
                \left(\log m\right)^{k O\left(r^2\right)}
                        \opt_{G}\left(\bb\right).
        \end{multline*}
        
        Next, we show that $\opt_G(\bb)\le \norm{R\bb}_{\infty} (\log m)^{kO(r^2)}$. Let $S$ be a subset of $V$ such that $\opt_G(\bb)=\frac{\bb_S}{\cc_S}$. Since $\sum\lambda_iH_i$ can be embedded into $G$ with congestion $\Otil(\log m)$, there exists an $i$ such that $\cc_{H_i, S}\le \Otil(\log m)\cc_{G, S}$ where $\cc_{H_i, S}$ is the total capacity of edges leaving $S$ in $H_i$. Thus, 
        \[
        \opt_G(\bb)=\frac{\bb_S}{\cc_{G, S}}\le \Otil(\log m)\frac{\abs{\bb_S}}{\cc_{H_i, S}}\le \Otil(\log m)\opt_{H_i}(\bb).
        \] 
        By~\ref{lem:route_in_core}, $\opt_{H_i}(\bb)\le \max\{\opt_{H_i'}(\bb_{H_i'}), \rho\}$. where $\rho=\max\{\bb_S/\cc_{H_i, S}|S\text{ is the tree cut corresponds to }e\}$. $\bb_S/\cc_{H_i, S}$ is upper bounded by $\norm{R\bb}_{\infty}$ since $\chi_S/\cc_{H_i, S}$ is a row of $R$. $\opt_{H_i'}(\bb_{H_i'})$ is upper bounded by $(\log m)^{O(r^2)}\opt_{\tilde{H}_i'}(\bb_{H_i'})$ since $\tilde{H}_i'\le_c (\log m)^{O(r^2)}H_i'$.
        By inductive hypothesis,
        \[
        \opt_{\tilde{H}_i'}(\bb_{H_i'})\le \norm{\tilde{R}\bb_{H_i'}}_{\infty} (\log m)^{(k-1)O(r^2)}
        \]
        where $\tilde{R}$ is the congestion approximator computed by $\mathtt{CongestionApproximator}(\tilde{H}_i', t)$.
        $\tilde{R}$ is mapped to a submatrix $R^{(i)}$ in $R$ such that $R^{(i)}\bb=\tilde{R}\bb_{H_i'}$.
        Thus,
        \[
        \norm{\tilde{R}\bb_{H_i'}}_{\infty}
        =
        \norm{R^{(i)}\bb}_\infty\le \norm{R\bb}_{\infty}.
        \]
        Connecting these inequalities gives 
        \[\opt_{H_i'}(\bb_{H_i'}) \le (\log m)^{O(r^2)}\opt_{\tilde{H}_i'}(\bb_{H_i'}) \le \norm{\tilde{R}\bb_{H_i'}}_{\infty} (\log m)^{kO(r^2)} \le \norm{R\bb}_{\infty}(\log m)^{kO(r^2)}.
        \] Together with $\rho=\max\{\bb_S/\cc_e|S\text{ is the tree cut corresponds to }e\}\le \norm{R\bb}_{\infty}$, this completes the inductive step. 
        %Sparsified cores route $\bb_{core}$s with congestion $\norm{R\bb}_{\infty} (\log m)^{(k-1)O(k^2)}$. Cores route the same demand with $(\log m)^{O(k^2)}$ times more congestion. Since the convex combination of cores and forests can be embedded in $G$ with congestion $\Otil(\log m)$, $G$ routes $x\bb$. 
        
        %The recursion will proceed for $k$ levels.
        %Each layer, the converstion to $j$-trees incur a multiplicative error of $O(\log{n})$,
        %while the sparsifier by~\ref{cor:sparsifier} and the setting of $r \leftarrow k$
        %incurs a multiplicative error of $O(\log^{O(k^2)}n)$.
        %So compounding across the $k$ layers gives that the resulting set of cuts
        %is an $O(\log^{k^3}n)$ approximator.
        
        Then it remains to bound the total running time and size of approximators. There are $r\le O(\log m)$ levels of recursion and the capacity ratio increases by a factor of $O(m)$ at each level by \Cref{thm:madrythm}, so the capacity ratio is always at most $U_{\max}\le O(m^{O(\log m)}U)$.
        At each level, each of the at most $\max\{1, \frac{m\log^{O(1)}m\log U_{\max}}t\}$-trees have at most $O(\frac{m\log^{O(1)}m\log U_{\max}}t)$ vertices,
        which by the increases in sizes in the cut sparsifiers from~\ref{cor:sparsifier}
        gives at most $ O(\frac{m\log^{O(1)}m\log U_{\max}}t) \log m\log U_{\max}\le O(\frac{m\log^{O(1)}m\log^2 U}t) \le m^{1-1/r}$ edges for large enough $t=\Omega(m^{1/r}\log^{O(1)}m\log^{2} U)$.
        Summing across the $t$ graphs in the distribution
        gives a total size of $O(m\log^{O(1)}m\log^{2} U)\le O(m\log^{O(1)}   (mU))$ edges, which is an increase by a factor of
        $O\left(\log^{O(1)}(mU)\right)$ across each level of the recursion.
        Thus, the total sizes of these graph after $r$ levels of recursion is
        $O\left(m \log^{O(r)} (mU)\right)$. The total running time is dominated by sparsifying the graphs after $r$ levels which is $O\left(m^{1+O(1/r)}\log ^{O(r^3)} m\right)$.
        
        This size serves as an upper bound on the total number of cuts examined in the trees.
        Multiplying in the cost of running the sparsifier then gives the overall runtime as well.

        Furthermore, the reduction of $\bb$ to the core graph consists
        of summing over the trees, and takes time linear in the size of the trees.
        Thus the cost of computing matrix-vector products in $R$ and $R^{\top}$
        follow as well.
\end{proofof}

\section{Faster Algorithm for $\protect\BCut$ in Low Conductance Regime}

\label{sec:BCut_nophi}

In this section we complete the proof of \ref{thm:intro:main}, by strengthening the results of \ref{thm: main slower alg}, so that the running time no longer depends on the conductance parameter $\phi$. In order to do so, we start by introducing a new algorithm for the matching player. The algorithm is an analogue of the algorithm from \ref{thm:push match}, except that its running time no longer depends on $\phi$. This is achieved by using \ref{cor:maxflow} in order to compute flows and cuts, instead of the push-relabel algorithm from \ref{lem:unitflow}. However, the new algorithm for the matching player does not return the routing paths, and instead only returns a partial matching, for which we are guaranteed that the routing paths exist.  The second obstacle is that we can no longer use the algorithm for expander pruning from \ref{thm: expander pruning}, since its running time also depends on the parameter $\phi$. Instead, we use a different high-level approach. We define the Most-Balanced Cut problem, and provide a bi-criteria approximation algorithm for it, that exploits the cut-matching game, the algorithm from \ref{thm: cut player} for the cut player, and the new algorithm for the matching player. We then show that this approximation algorithm can be used in order to approximately solve the $\BCut$ problem.
We start with providing a new algorithm for the matching player in \ref{subsec: new matching player}.

\subsection{New Algorithm for the Matching Player}\label{subsec: new matching player}
The new algorithm for the matching player is summarized in the following theorem. The theorem and its proof are similar to Lemma
B.18 of \cite{NanongkaiS17}.

\begin{thm}\label{thm:efficient matching player}
	There is a deterministic algorithm, that we call \matchorcut, that, given an $m$-edge graph $G=(V,E)$, two disjoint subsets $A,B$ of its vertices, where $|A|\leq |B|$ and $|A|=N$, and parameters $z\geq 0$, $0<\psi<1/2$, computes one of the following:
	
	\begin{itemize}
		\item either a partial matching $M\subseteq A\times B$ with $|M|> N-z$, such that there exists a set $\pset=\set{P(a,b)\mid (a,b)\in M}$ of paths in $G$, where for each pair $(a,b)\in M$, path $P(a,b)$ connects $a$ to $b$, and the paths in $\pset$ cause congestion at most $O\left (\frac{ \log n}{\psi}\right )$; or
		
		\item a cut $(X,Y)$ in $G$, with $|X|,|Y|\geq z/2$, and $\Psi_G(X,Y)\leq \psi$.
	\end{itemize}
	
	The running time of the algorithm is $O\left (m^{1+o(1)}\right)$.
	\thatchaphol{All $m^{o(1)}$ factors in this section can be written concretely as $\exp\left(O\left( \left( \log{m} \log\log{m}\right)^{3/4} \right) \right)$ using \Cref{cor:maxflow}}
\end{thm}

We note that, if the algorithm returns the matching $M$, then it does not explicitly compute the corresponding set $\pset$ of paths.
Note also that, if the parameter $z=1$, and the algorithm returns a matching $M$, then $|M|=N$ must hold, that is, every vertex of $A$ is matched.

\begin{proof}
Let $x=\ceil{1/\psi}$, and let $\psi'=1/x$, so that $1/\psi'$ is an integer, and $\psi/2\leq \psi'\leq \psi$. Notice that it is enough to prove \ref{thm:efficient matching player} for parameter $\psi'$ instead of $\psi$, so for simplicity, we denote $\psi'$ by $\psi$ from now on. 
We use the following lemma as a subroutine.

\begin{lem}\label{lem: matching player 1 it max flow}
	There is a deterministic algorithm, that, given an $m$-edge graph $G=(V,E)$, two disjoint subsets $A',B'$ of its vertices, such that $|A'|\leq |B'|$ and $|A'|=N'$, and a parameter $0<\psi<1$ such that $1/\psi$ is an integer, computes one of the following:
	
	\begin{itemize}
		\item either a partial matching $M'\subseteq A'\times B'$ with $|M'|\geq \Omega(N')$, such that there exists a set $\pset'=\set{P(a,b)\mid (a,b)\in M'}$ of paths in $G$, where for each pair $(a,b)\in M$, path $P(a,b)$ connects $a$ to $b$, and the paths in $\pset'$ cause congestion $O(1/\psi)$; or
		
		\item a cut $(X,Y)$ in $G$, with $|X|,|Y|\geq N'/2$, and $\Psi_G(X,Y)\leq \psi$.
	\end{itemize}
	
	The running time of the algorithm is $O(m^{1+o(1)})$.
\end{lem}

We provide the proof of \ref{lem: matching player 1 it max flow} below, after we prove \ref{thm:efficient matching player} using it.
Our algorithm performs a number of iterations. We maintain a current matching $M$, starting with $M=\emptyset$, and subsets $A'\subseteq A,B'\subseteq B$ of vertices that do not participate in the current matching $M$, starting with $A'=A,B'=B$. For the sake of analysis, we keep track of a set $\pset=\set{P(a,b)\mid (a,b)\in M}$ of paths in $G$, where for each $(a,b)\in M$, path $P(a,b)$ connects $a$ to $b$ in $G$ (but this set of paths is not explicitly computed by the algorithm). We perform iterations as long as $|A'|\geq z$. In every iteration, we apply the algorithm from \ref{lem: matching player 1 it} to the current two sets $A',B'$ of vertices. If the outcome of the algorithm is a cut $(X,Y)$ with  $|X|,|Y|\geq N'/2$, and $\Psi_G(X,Y)\leq \psi$, then we say that the current iteration terminated with a cut. We then terminate the algorithm, and return the cut $(X,Y)$ as its output. Since $|A'|\geq z$, we are guaranteed that $|X|,|Y|\geq z/2$ as required.

Otherwise, the algorithm returns a partial matching $M'\subseteq A'\times B'$ with $|M'|\geq \Omega(N')$, such that there exists a set $\pset'=\set{P(a,b)\mid (a,b)\in M'}$ of paths in $G$, where for each pair $(a,b)\in M$, path $P(a,b)$ connects $a$ to $b$, and the paths in $\pset'$ cause congestion $O(1/\psi)$. We then say that the current iteration terminated with a matching. We add the pairs in $M'$ to the matching $M$, and delete from sets $A',B'$ all vertices that participate in $M'$. Also, for the sake of analysis, we implicitly add the paths in $\pset'$ to the set $\pset$.

Notice that in every iteration, $|A'|$ is guaranteed to reduce by a constant factor, so the number of iterations is $O(\log n)$, and the total running time of the algorithm is $O( m^{1+o(1)})$. If every iteration of the algorithm terminated with a matching, then at the end of the algorithm, $|A'|<z$, and so $|M|> N-z$. Moreover, there exists a set $\pset=\set{P(a,b)\mid (a,b)\in M}$ of paths in $G$, where for each pair $(a,b)\in M$, path $P(a,b)$ connects $a$ to $b$ --- the set $\pset$ of paths that we have implicitly maintained. The congestion caused by this set of paths is $O( \log n/\psi)$, since there are $O(\log n)$ iterations, and the set of paths corresponding to each iteration causes congestion $O(1/\psi)$. In order to complete the proof of  \ref{thm:efficient matching player}, it now remains to prove \ref{lem: matching player 1 it max flow}.

\begin{proofof}{\ref{lem: matching player 1 it max flow}}
\iffalse
We start by constructing an arbitrary spanning tree $T$ of $G$, using standard greedy algorithm whose running time is $O(m)$. Next, we employ the standard grouping technique (see e.g.~\cite{ANF,CKS,RaoZhou,Andrews,Chuzhoy11}), except that we show that it can be implemented efficiently, in the next lemma.
	
\begin{lem}\label{lem: grouping}
	There is a deterministic algorithm, that, given a tree $T$ on $n$ vertices, a parameter $q>1$, and a subset $U\subseteq V(T)$ of at least $q$ vertices called terminals computes, in time $O(n)$, a collection $\tset=\set{T_1,\ldots,T_r}$ of disjoint sub-trees of $T$, such that for all $1\leq i\leq r$, $q\leq |T_i\cap U|\leq 4\Delta q$, and for all $u\in U$, some tree $T_i$ contains $u$.
\end{lem}	
\fi
%
	We can assume that $\psi\geq 1/n$, as otherwise the problem is trivial to solve, since we are allowed to compute a routing of $A',B'$ with congestion $n$.
We construct a new edge-capacitated graph $\hat G$, as follows. We start with $\hat G=G$, and we set the capacity $c_e$ of every edge $e\in E$ to be $1/\psi$ (recall that $1/\psi$ is an integer). Next, we introduce a source vertex $s$, that connects to every vertex in $A'$ with an edge of capacity $1$, and a destination vertex $t$, that connects to every vertex in $B'$ with an edge of capacity $1$. We set the demand $b(s)=N'/2$, $b(t)=-N'/2$, and for every vertex $v\neq s,t$, we set $b(v)=0$. We then apply the algorithm from \ref{cor:maxflow} to this new capacitated graph $\hat G$, the resulting demand vector $b$, and accuracy parameter $\epsilon=1/2$. Note that the ratio of largest to smallest edge capacity is $O(1/\psi)=O(n)$, and so the running time of the algorithm from   \ref{cor:maxflow} is $O\left (m^{1+o(1)}\right )$.

We now consider two cases. Assume first that the algorithm computes a cut $S$ with $\sum_{e\in E(S,\overline{S})}c_{e}<|\sum_{v\in S}b_{v}|$. Since $\sum_{v\in V}b_v=0$, we can assume w.l.o.g. that $\sum_{e\in E(S,\overline{S})}c_{e}<\sum_{v\in S}b_{v}$, by switching the sides of the cut if necessary.  Clearly, $s\in S$, $t\not\in S$ must hold, and so $\sum_{e\in E(S,\overline{S})}c_{e}<N'/2$. In particular, at least $N'/2$ vertices of $A'$ must lie in $S$ (as otherwise edges connecting them to $s$ will contribute capacity at least $N'/2$ to the cut), and similarly at least $N'/2$ vertices of $B'$ must lie in $\overline S$. We set $X=S\setminus \set{s}$, and $Y=\overline S\setminus \set{t}$, and return the cut $(X,Y)$ as the outcome of the algorithm. As observed above, $|X|,|Y|\geq N'/2$ must hold. Moreover, since every edge in $E_{G}(X,Y)$ contributes capacity $1/\psi$ to the cut $(S,\overline S)$, we get that $|E_G(X,Y)|\leq \psi \cdot N'/2$. We conclude that $\Psi_G(X,Y)\leq \psi$.

We now consider the second case, where the algorithm from \ref{cor:maxflow} returns an $s$-$t$ flow $f$ of value at least $N'
/2$ in $\hat G$, that causes edge-congestion at most $(1+\epsilon)\leq 2$.

We use the \emph{link-cut tree data} structure of Sleator and Tarjan \cite{SleatorT83}.
The data structure maintains a forest $F$ of rooted trees, over a set $V$ of $n$ vertices, with edge costs $w(e)$ for $e\in E(F)$, and supports the following operations (we only list operations that are relevant to us):

\begin{itemize}
	\item $\rootx(v)$: return the root of the tree containing the vertex $v$;
	\item $\mincost(v)$: return the vertex $x$ closest to $\rootx(v)$, such that the edge connecting $x$ to its parent in the tree has minimum cost among all edges lying on the unique path connecting $v$ to $\rootx(v)$.
	\item $\parent(v)$: return the parent of $v$ in the tree containing $v$;
	\item $\update(v,w)$: update the costs of all edges lying on the path connecting $v$ to $\rootx(v)$, by adding $w$ to the cost of each edge (we note that $w$ may be negative);
	\item $\link(u,v,w)$: for vertices $u$, $v$ that lie in different trees, add an edge $(u,v)$ of cost $w$;
	\item $\evert(v)$: make $v$ the root of the tree containing $v$; and
	\item $\cut(v)$: delete the edge connecting $v$ to $\parent(v)$ (this operation assumes that $v\neq \rootx(v)$).
\end{itemize}

Sleator and Tarjan  \cite{SleatorT83} showed a deterministic algorithm for maintaining the link-cut tree data structure, with $O(\log n)$ worst-case time per operation.

We also use the algorithm of \cite{KangP15}, that, given any graph $H$ with integral capacities $c_e\geq 0$ on its edges, with two special vertices $s$ and $t$, and any (possibly fractional) $s$-$t$ flow $f$ of value $\Lambda$ in $H$, that does not violate the edge capacities, computes, in time 
$\tilde O(m)$, an {\bf integral} $s$-$t$ flow $f'$ of value at least $\Lambda$, that does not violate the edge capacities. The algorithm of \cite{KangP15} is deterministic, and relies on link-cut trees. 

We apply the algorithm of \cite{KangP15} to graph $\hat G$, and the flow $f'=f/2$. Note that flow $f'$ does not violate the capacities of edges in $\hat G$, and that its value is at least $N'/4$. We denote by $f''$ the integral flow of value at least $N'/4$ that is computed by this algorithm. Notice that flow $f''$ naturally defines a flow $\hat f$ in the original graph $G$, of value at least $N'/4$, from vertices of $A'$ to vertices of $B'$. Since flow $f''$ obeys the edge capacities, every vertex in $A'$ sends either $0$ or $1$ flow units in $\hat f$, and every vertex in $B'$ receives either $0$ or $1$ flow units in $\hat f$. We denote by $A''\subseteq A'$ the set of vertices that send one flow unit, and by $B''\subseteq B'$ the set of vertices that receive one flow unit in $f$; observe that $|A''|=|B''|\geq N'/4$.

Intuitively, we would like to compute the flow-paths decomposition of $\hat f$. However, computing the flow-paths explicitly may take too much time, so instead, we would like to only compute pairs of vertices that serve as endpoints of the paths in the decomposition. We do so using link-cut trees. The algorithm, that we refer to as $\computematch$ is very similar to the algorithm for computing blocking flows in \cite{SleatorT83} (see Section 6 of \cite{SleatorT83}) and is included here for completeness.

We will gradually construct the desired matching $M'\subseteq A''\times B''$, and we will implicitly maintain a set of paths $\pset'=\set{P(a,b)\mid (a,b)\in M'}$ routing the pairs in $M'$. The set $\pset'$ of paths can be obtained by computing the flow decomposition of $\hat f$. However, our algorithm will not compute the path set $\pset'$ explicitly (as this would take too much time), and instead will only guarantee its existence.

We maintain a directed graph $H$ with $V(H)=V(G)$ and $E(H)$ containing all edges $e$ of $G$ with $\hat f(e)\neq 0$ (the direction of the edge is in the direction of the flow). For every vertex $v$, we denote by $\out(v)$ the set of all edges that leave $v$ in $H$.

We will also maintain an (undirected) forest $F$ with $V(F)=V(H)$, whose edges are a subset of $E(H)$ (though they do not have direction anymore). Further, we will ensure that the following invariants hold throughout the algorithm:

\begin{properties}{I}
	\item  for every vertex $v\in V(F)$, at most one edge of $\out(v)$ belongs to $F$; \label{inv: at most one out-edge}
	\item for every tree $T$ in the forest $F$, if $v$ is the root of $T$, then no edge of $\out(v)$ lies in $F$; \label{inv: root no out edge}; and
	\item for every tree $T$ and vertex $u\in V(T)$, there is a {\bf directed} path in graph $H$ connecting $u$ to the root of $T$, that only contains edges that lie in $T$.\label{inv: directed path to root}
\end{properties}

Throughout the algorithm, we (implicitly) maintain a valid integral flow from vertices of $A''$ to vertices of $B''$, as follows. For every edge $e\in E(F)$, the flow on $e$ is the cost $w(e)$ of the edge in $F$, and the direction of the flow is the same as the direction of the edge in $H$. For an edge $e\in E(H)$ that does not lie in $F$, the flow on $e$ is the value $\hat f(e)$ (that may be updated over the course of the algorithm). An edge that carries $0$ flow units is deleted from both $H$ and $F$. We will ensure that for every edge $e\in F$, $w(e)\geq 0$ holds at all times.

The algorithm uses a procedure $\updateflow(x,w)$ that receives as input a vertex $x$ of the forest $F$ and an integer $w$, such that, if we denote by $P_x$ the unique path connecting $x$ to the root $r$ of the tree of $F$ containing $x$, then for every edge $e\in P_x$, $w(e)\geq w$. The procedure decreases the cost of every edge on path $P_x$ by $w$. Additionally, it deletes every edge $e\in P_x$ whose new cost becomes $0$, while maintaining all invariants; each such edge is also deleted from $H$. The procedure starts by executing  $\update(x,-w)$, that decreases the weight of every edge on path $P_x$ by $w$. Next, we iteratively remove edges from $F$ whose new cost becomes $0$. In order to do so, we maintain a current vertex $u$, starting with $u=\mincost(x)$. An iteration is executed as follows. Let $T$ denote the tree of $F$ containing $u$, let $r$ be the root of $T$, and let $u'$ be the parent of $u$ in the tree. If $w(u,u')\neq 0$, then we terminate the algorithm. Otherwise, we execute $\cut(u)$, deleting the edge $(u,u')$ from the tree $T$, that decomposes into two subtrees: tree $T'$ containing $u$ and tree $T''$ containing $u'$ and $r$. The root of tree $T''$ remains $r$, while the root of $T'$ becomes $u$. Observe that all invariants continue to hold. We also delete the edge $(u,u')$ from the graph $H$. We then set $u=\mincost(x)$, and continue to the next iteration.

\begin{algorithm}
	$\updateflow(x,w)$:
	\begin{itemize}
		\item Execute $\update(x,-w)$.
		\item Set $u\leftarrow \mincost(x)$.
		\item let $T$ be the tree containing $u$, let $r$ be the root of $T$ and let $u'$ be the parent of $u$ in the tree $T$.
		\item while $w(u,u')=0$ and $u\neq \rootx(u)$:
		\begin{itemize}
			\item Delete edge $(u,u')$ from $H$.
			\item Execute $\cut(u)$. This decomposes $T$ into two sub-trees: tree $T'$ containing $u$, and tree $T''$ containing $u'$ and $r$. The root of $T'$ becomes $u$ and the root of $T''$ becomes $r$.
			\item Update $u\leftarrow \mincost(x)$. Update $T$ to be the tree of $F$ containing the new vertex $u$, let $r$ be the root of $T$, and let $u'$ the parent of $u$ in $T$ (if $u=\rootx(u)$, set $u'=u$).
			\end{itemize}
		\end{itemize}
\end{algorithm}

We now describe the algorithm $\computematch$. The algorithm initializes the forest $F$ to contain the set $V(G)$ of vertices and no edges. Notice that all invariants hold for $F$. It then iteratively considers every vertex $a\in A''$ one-by-one and applies procedure $\process(a)$ to each such vertex. 

\begin{algorithm}
	$\computematch$:
	\begin{itemize}
		\item Initialize $F$ to contain the set $V(G)$ of vertices and no edges.
		\item For all $a\in A''$: execute $\process(a)$.
	\end{itemize}
\end{algorithm}

We now describe procedure $\process(a)$. The goal of the procedure is to find a vertex $b\in B''$ such that some path $P$ connecting $a$ to $b$ carries one flow unit in the current flow. We do not compute the path $P$ explicitly, but we reduce, for each edge $e\in P$, the amount of flow that $e$ carries by one unit. The procedure consists of a number of iterations. At the beginning of every iteration, we start with the current vertex $v$, which is set to be $\rootx(a)$. The iterations are executed as long as $v\not\in B''$. Notice that, from our invariant, no edge in $\out(v)$ lies in $F$. Let $T$ denote the tree of $F$ containing $v$. We now consider three cases. First, if $\out(v)=\emptyset$, then it is impossible to reach vertices of $B''$ from 
$v$ in the current graph $H$. We iteratively delete every edge $(y,v)$ that belongs to the tree $T$, using operation $\cut(y)$. This operation splits the tree $T$ into two subtrees, one whose root remains $v$, and one whose root becomes $y$. It is easy to verify that all invariants continue to hold. The second case is when some edge $(v,u)\in \out(v)$ exists in $H$, and $u\not\in V(T)$ (this can be checked by running $\rootx(u)$ and comparing the outcome with $v$). Let $T'$ be the tree of $F$ that contains $u$, and let $r'$ be its root. We join the two trees by using operation $\link(v,u,w)$, where $w$ is the current flow on edge $(v,u)$ in graph $H$. The root of the new tree becomes $r'$.  
It is immediate to verify that Invariants \labelcref{inv: at most one out-edge} and \labelcref{inv: root no out edge} continue to hold. In order to verify that Invariant \labelcref{inv: directed path to root} continues to hold, note that for every vertex $x\in V(T')$, there is a directed path $P_x$ in graph $H$, connecting $x$ to $r'$, that only contains edges of $T'$. 
Consider now some vertex $y\in V(T)$. From Invariant \labelcref{inv: directed path to root}, 
there is a directed path $P_y$ in graph $H$, connecting $y$ to $v$, that only contains edges of $T$. By using the edge $(v,u)$ and the path $P_u$ connecting $u$ to $r'$, we obtain a directed path in graph $H$, connecting $y$ to $r'$, that only uses edges that lie in the new tree.
The third and the last case is when the endpoint $u$ edge $(v,u)\in \out(v)$ lies in the tree $T$. In this case, there is a directed cycle in graph $H$, that includes the edge $e=(v,u)$ and the path $P_u$ that is contained in $T$ and connects $u$ to $v$. We let $w$ be the minimum between the current flow $\hat f_e$ on edge $e$, and the smallest value $w(e')$ of an edge $e'\in P_u$, that can be computed by executing $x\leftarrow \mincost(u)$ and inspecting the edge that connects $x$ to its parent in $T$. We decrease the value $\hat f_e$ by $w$; if the value becomes $0$, then we delete the edge $e$ from $H$. Additionally, we execute $\updateflow(u,w)$.

The iterations are terminated once $b=\rootx(a)$ is a vertex of $B''$. We then add the pair $(a,b)$ to $M'$. The intended path for routing this pair is the path $P_a$ in the tree $T$ containing $a$ that connects $a$ to $b$. We execute $\updateflow(a,1)$ in order to decrease the flow on this path by one unit.

This completes the description of the algorithm. It is easy to verify that the algorithm simulates the standard flow-paths decomposition, without explicitly computing the paths, and that there is a collection $\pset'=\set{P(a,b)\mid (a,b)\in M'}$ of paths in $G$, where path $P(a,b)$ connects $a$ to $b$, such that the paths in $\pset'$ cause congestion at most $O(1/\psi)$. 

\begin{algorithm}
	$\process(a)$:
	\begin{itemize}
		\item Let $v\leftarrow \rootx(a)$, and let $T$ be the tree of $F$ containing $v$.
		\item While $v\not\in B''$:
		
		\begin{itemize}
		  	\item If $\out(v)=\emptyset$:
			\begin{itemize}
				\item for every child $z$ of $v$, execute $\cut(z)$; vertex $z$ becomes the root of the newly created tree.
			\end{itemize}
		\item Otherwise, let $(v,u)\in \out(v)$ be any edge of $\out(v)$.
		\begin{itemize}
		\item If $u\not \in T$:
		\begin{itemize}
			\item Let $T'$ be the tree of $F$ containing $u$ and let $r'$ be its root.
		    \item Execute $\link(u,v,w)$, where $w$ is the current flow value $\hat f(e)$ of the edge $e=(v,u)$. The root of the new merged tree becomes $r'$.
		\end{itemize}
	
	\item Otherwise:
	\begin{itemize}
		\item Let $x\leftarrow \mincost(u)$, and let $w_1$ be the cost of the edge connecting $x$ to its parent in $T$.
		\item Let $w_2$ be the flow $\hat f(e)$ on the edge $e=(v,u)$ in $H$.
		\item Set $w=\min\set{w_1,w_2}$.
		\item Set $\hat f(e)\leftarrow \hat f(e)-w$. If $\hat f(e)=0$, delete $e$ from $H$.
		\item Execute $\updateflow(u,-w)$.%\thatchaphol{Should this be $-w$?}\jnote{yes, fixed.}
	\end{itemize}
			
		\end{itemize}	
		\end{itemize}
		\item Add $(a,v)$ to $M'$.
		
		\item Execute $\updateflow(v,-1)$.
	\end{itemize}
\end{algorithm}

In order to analyze the running time of the algorithm, observe that every edge may be inserted at most once into $F$ and deleted at most once from $F$ (this is since an edge is only deleted from $F$ when the flow on the edge becomes $0$; at this point the edge is also deleted from $H$ and is never again inserted into $H$ or $F$.)
Observe that the number of update operations of the link-cut tree data structure due to a single call to $\updateflow$ subroutine is $O(1+n')$, where $n'$ is the number of edges that were deleted from $F$ during the procedure (notice that it is possible that no edge is deleted from $F$ during the procedure). Whenever procedure $\updateflow$ is called, we either delete at least one edge from $F$, or we delete at least one edge from $H$ (when we eliminate a flow cycle), or we add one pair to matching $M'$. Therefore, the total number of update operations of the link-cut tree data structure due to $\updateflow$ subroutine is $O(m)$.

We now consider the execution of procedure $\process(a)$, ignoring the calls to $\updateflow$ subroutine. In every iteration of this procedure, the number of updates to the link-cut tree data structure is proportional to the number of edges that were inserted into $F$ or deleted from $F$. 
It is then easy to see that the total running time of $\computematch$ is $O(m\log n)$.

\end{proofof}

\end{proof}

\subsection{Most-Balanced Sparse Cut}
\label{sec:most-bal cut}
In the Most Balanced Sparse Cut problem, the input is a graph $G=(V,E)$, and a parameter $0<\psi\leq 1$. The goal is to compute a cut $(X,Y)$ in $G$, with $\Psi_G(X,Y)\leq \psi$, of maximum \emph{size}, which is defined to be $\min\set{|X|,|Y|}$. The problem (or its variations) was defined independently by \cite{NanongkaiS17}
and \cite{Wulff-Nilsen17}, and it was also used in \cite{ChuzhoyK19} and \cite{ChangS19}. As observed in these works, one can obtain a bi-criteria approximation algorithm for this problem by using the cut-matching games. The following two lemmas summarize these algorithms, where we employ Algorithm  \cutorcert from \ref{thm: cut player} for the cut player, and Algorithm \matchorcut from \ref{thm:efficient matching player} for the matching player, in order to implement them efficiently.

\begin{lem}\label{thm: sparse edge cut of large profit or witness}
There are universal constants $N_0,c_0$, and a deterministic algorithm, that, given an $n$-vertex and $m$-edge graph $G=(V,E)$ and parameters $0<\psi\leq 1$, $0<z\leq n$ and $r\geq 1$, such that $n^{1/r}\geq N_0$:
	
	\begin{itemize}
		\item either returns a cut $(X,Y)$ in $G$ with $\Psi_G(X,Y)\leq \psi$ and $|X|,|Y|\geq z$; 
		
		\item or correctly establishes that for every cut $(X',Y')$ in $G$ with $\Psi_G(X',Y')\leq \psi/ (\log n)^{c_0r}$, $\min\set{|X'|,|Y'|}<c_0z\cdot (\log n)^{c_0r}$ holds.
	\end{itemize}
	
	The running time of the algorithm is $O\left (m^{1+O(1/r)+o(1)}\cdot (\log n)^{O(r^2)}\right )$.
\end{lem}

\begin{proof}
	The algorithm employs the cut-matching game, and will maintain a set $F$ of fake edges.
	We assume that $n$ is an even integer; otherwise we add a new isolated vertex $v_0$ to $G$, and we add a fake edge connecting $v_0$ to an arbitrary vertex of $G$ to $F$. We also maintain a graph $H$, that initially contains the set $V$ of vertices and no edges. We then perform a number of iterations, that correspond to the cut-matching game. In every iteration $i$, we will add a matching $M_i$ to graph $H$. We will ensure that the number of iterations is bounded by $O(\log n)$, so the maximum vertex degree in $H$ is always bounded by $O(\log n)$. At the beginning of the algorithm, graph $H$ contains the set $V$ of vertices and no edges. We now describe the execution of the $i$th iteration.
	
	In order to execute the $i$th iteration, we apply Algorithm  \cutorcert from \ref{thm: cut player} to graph $H$, where the constant $N_0$ and the parameter $r$ remain unchanged. Assume first that the output of the algorithm from \ref{thm: cut player} is a cut $(A_i,B_i)$ in $H$ with $|A_i|,|B_i|\geq n/4$ and $|E_H(A,B)|\leq n/100$. We then compute an arbitrary partition $(A'_i,B'_i)$ of $V(G)$ with $|A'_i|=|B'_i|$ such that $A_i\subseteq A'_i$ and $B_i\subseteq B'_i$. We treat the cut $(A'_i,B'_i)$ as the move of the cut player. 
	Then, we apply Algorithm \matchorcut from \ref{thm:efficient matching player} to the sets $A'_i,B'_i$ of vertices, a sparsity parameter $\psi'=\psi/2$ and parameter $z'=4z$. If the algorithm returns a cut $(X,Y)$ in $G$, with $|X|,|Y|\geq z'/2\geq 2z$, and $\Psi_G(X,Y)\leq \psi'$, then we terminate the algorithm and return the cut $(X,Y)$, after we delete the extra vertex $v_0$ from it (if it exists). It is easy to verify that $|X|,|Y|\geq z$ and $\Psi_G(X,Y)\leq \psi$ must hold. 
	Otherwise, the algorithm from \ref{thm:efficient matching player} computes a partial matching $M'_i\subseteq A'_i\times B'_i$ with $|M'_i|\geq N-4z$, such that there exists a set $\pset'_i=\set{P(a,b)\mid (a,b)\in M'_i}$ of paths in $G$, where for each pair $(a,b)\in M'_i$, path $P(a,b)$ connects $a$ to $b$, and the paths in $\pset'_i$ cause congestion at most $O\left (\frac{ \log n}{\psi}\right )$.
	We let $A''_i\subseteq A'_i$, $B''_i\subseteq B'_i$ be the sets of vertices that do not participate in the matching $M'_i$, and we let $M''_i$ be an arbitrary perfect matching between these vertices. We define a set $F_i$ of fake edges, containing the edges of $M''_i$, and an embedding $\pset''_i=\set{P(e)\mid e\in F_i}$ of the edges in $M''_i$, where each fake edge is embedded into itself. Lastly, we set $M_i=M'_i\cup M''_i$, add the edges of $M_i$ to $H$, and continue to the next iteration.
Notice that $|F_i|\leq 4z$.

We perform the iterations as described above, until Algorithm  \cutorcert from \ref{thm: cut player}  returns a subset $S\subseteq V$ of at least $n/2$ vertices, such that $\Psi(G[S])\geq 1/(\log n)^{O(r)}$. 
Recall that \ref{thm:KKOV-new} guarantees that this must happen after at most $O(\log n)$ iterations. We then perform one last iteration, whose index we denote by $q$.

We let $B_q=S$ and $A_q=V(G)\setminus S$, and apply Algorithm  \matchorcut from \ref{thm:efficient matching player} to the sets $A_q,B_q$ of vertices,  a sparsity parameter $\psi'=\psi/2$ and parameter $z'=4z$. As before, if the algorithm returns a cut $(X,Y)$ in $G$, with $|X|,|Y|\geq z'/2\geq 2z$ and $\Psi_G(X,Y)\leq \psi'$, then we terminate the algorithm and return the cut $(X,Y)$, after we delete the extra vertex $v_0$ from it (if it exists). As before, we get that $|X|,|Y|\geq z$ and $\Psi_G(X,Y)\leq \psi$. 
Otherwise, the algorithm from \ref{thm:efficient matching player} computes a partial matching $M'_q\subseteq A'_q\times B'_q$ with $|M'_q|\geq N-4z$, such that there exists a set $\pset'_q=\set{P(a,b)\mid (a,b)\in M'_q}$ of paths in $G$, where for each pair $(a,b)\in M'_q$, path $P(a,b)$ connects $a$ to $b$, and the paths in $\pset'_q$ cause congestion at most $O\left (\frac{ \log n}{\psi}\right )$.
We let $A'_q\subseteq A_q$, $B'_q\subseteq B_q$ be the sets of vertices that do not participate in the matching $M'_q$, and we let $M''_q$ be an arbitrary matching that connects every vertex of $A'_q$ to a distinct vertex of $B'_q$ (such a matching must exist since $|A_q|\leq |B_q|$). As before, we define a set $F_q$ of fake edges, containing the edges of $M''_q$, and an embedding $\pset''_q=\set{P(e)\mid e\in F_q}$ of the edges in $M''_q$, where each fake edge is embedded into itself. Lastly, we set $M_q=M'_q\cup M''_q$, and we add the edges of $M_q$ to graph $H$.

From now on we assume that the algorithm never terminated with a cut $(X,Y)$ with $|X|,|Y|\geq z$ and $\Psi_G(X,Y)\le \psi$. Note that, from \ref{obs: exp plus matching is exp}, the final graph $H$ is a $\psi'$-expander, for $\psi'\geq 1/(\log n)^{O(r)}$. Moreover, we are guaranteed that there is an embedding of $H$ into $G+F$ with congestion $O\left (\frac{ \log^2 n}{\psi}\right )$, where $F=\bigcup_{i=1}^rF_i$ is a set of $O(z\log n)$ fake edges. Notice that, in the embedding that we constructed, every edge of $H$ is either embedded into a path consisting of a single fake edge, or it is embedded into a path in the graph $G$; every fake edge in $F$ serves as an embedding of exactly one edge of $H$.
	
We now claim that there is a large enough universal constant $c_0$, such that, for every cut $(X',Y')$ in $G$ with $\Psi_G(X',Y')\leq \psi/ (\log n)^{c_0r}$, $\min\set{|X'|,|Y'|}<c_0z\cdot (\log n)^{c_0r}$ holds.
Indeed, consider any cut $(X',Y')$ in $G$ with $|X'|,|Y'|\geq c_0z \cdot (\log n)^{c_0r}$. It is enough to show that $\Psi_G(X',Y')>\psi/ (\log n)^{c_0r}$. We assume w.l.o.g. that $|X'|\leq |Y'|$.

Notice that $(X',Y')$ also defines a cut in graph $H$, and, since $H$ is a $\psi'$-expander, $|E_{H}(X',Y')|\geq \psi'\cdot |X'|\geq \psi'\cdot c_0 z\cdot (\log n)^{c_0r}$. We partition the set  $E_{H}(X',Y')$ of edges into two subsets. The first subset, $E_1$, is a set of edges corresponding to the fake edges (so each edge $e\in E_1$ is embedded into a path consisting of a single fake edge), and $E_2$ contains all remaining edges (each of which is embedded into a path of $G$). Recall that the total number of the fake edges, $|F|\leq O(z\log n)$, while $\psi'=1/(\log n)^{O(r)}$. Therefore, by letting $c_0$ be a large enough constant, we can ensure that $|E_1|\leq |E_H(X',Y')|/2$. %, and that $|E_2|\geq \psi'\cdot c_0 z/2(\log n)^{c_0r}$ holds.

The embedding of $H$ into $G+F$ defines, for every edge $e\in E_2$ a corresponding path $P(e)$ in $G$, that must contribute at least one edge to the cut $E_G(X',Y')$. Since the embedding causes congestion $O\left (\frac{ \log^2 n}{\psi}\right )$, we get that:

\[
\begin{split}
|E_G(X',Y')|&\geq \Omega \left (\frac{|E_H(X',Y')|\cdot \psi}{\log^2 n}\right )\\
& \geq \Omega \left (\frac{\psi' \cdot \psi\cdot |X'|}{\log^2 n}\right )\\
&\geq \Omega  \left (\frac{\psi\cdot |X'|}{(\log n)^{O(r)}}\right ).
\end{split}
\]

By letting $c_0$ be a large enough constant, we get that $\Psi_G(X',Y')\geq \psi/ (\log n)^{c_0r}$, as required (we note that we have ignored the extra vertex $v_0$ that we have added to $G$ if $|V(G)|$ is odd, but the removal of this vertex can only change the cut sparsity and the cardinalities of $X'$ and $Y'$ by a small constant factor that can be absorbed in $c_0$).

Lastly, we bound the running time of the algorithm. The algorithm consists of $O(\log n)$ iterations. Every iteration employs Algorithm  \cutorcert from \ref{thm: cut player}, whose running time is $O\left (n^{1+O(1/r)}\cdot (\log n)^{O(r^2)}\right )$, and Algorithm \matchorcut from \ref{thm:efficient matching player}, whose running time is $O\left (m^{1+o(1)}\right)$. Therefore, the total running time is $O\left (m^{1+o(1)+O(1/r)}\cdot (\log n)^{O(r^2)}\right)$.
	
\end{proof}

\begin{lem}\label{thm: sparse edge cut or expander}
	There are universal constants $N_0,c_0$, and a deterministic algorithm, that, given an $n$-vertex $m$-edge graph $G=(V,E)$ and parameters $0<\psi\leq 1$ and $r\geq 1$, such that $n^{1/r}\geq N_0$:
	
%	There is a randomized algorithm, that, given a sub-graph $W'$ of $W$ containing at least half the vertices of $W$, together with a parameter $0<\alpha\leq 1/(64\log^9n)$:
	
	\begin{itemize}
		\item either returns a cut $(X,Y)$ in $G$ with $\Psi_G(X,Y)\leq \psi$;
		
		\item or correctly establishes that $G$ is a $\psi'$-expander, for $\psi'= \psi/ (\log n)^{c_0r}$.
	\end{itemize}
	
	The running time of the algorithm is $O\left (m^{1+O(1/r)+o(1)}\cdot (\log n)^{O(r^2)}\right )$.
\end{lem} 

\jnote{I switched it back from corollary to lemma. It's not a corollary of the previous lemma because we don't use it as a black box. We repeat the proof.}

\begin{proof}
The proof is almost identical to the proof of \ref{thm: sparse edge cut of large profit or witness}. The only difference is that we set the parameter $z$ that is used in the calls to Algorithm \matchorcut from \ref{thm:efficient matching player} to $1$. This ensures that no fake edges are introduced. The remainder of the proof is unchanged.
\end{proof}

We note that \ref{thm: sparse edge cut or expander} immediately gives a deterministic $(\log n)^r$-approximation algorithm for the
Sparsest Cut problem with running time $O\left (m^{1+O(1/r)+o(1)}\cdot (\log n)^{O(r^2)}\right )$, for all $r\leq O(\log n)$, proving \ref{thm: sparsest and low cond} for the Sparsest Cut problem.

\subsection{Completing the Proof of \ref{thm:intro:main}}

The goal of this subsection is to prove the following theorem.

\begin{thm}\label{thm:main restatement}
	There is a universal constant $N_1$, and a deterministic algorithm, that, given a graph $G$ with $m$ edges,
	and parameters $0<\phi\leq 1$ and $r\geq 1$, such that $m^{1/r}\geq N_1$, computes, in time $O\left ( m^{1+O(1/r)+o(1)}\cdot (\log m)^{O(r^2)}\right )$, a cut $(A,B)$ in $G$ with $|E_G(A,B)|\leq \phi\cdot (\log m)^{O(r^2)}\cdot \vol(G)$, such that one of the following holds:
	
	\begin{itemize}
		\item either $\vol_{G}(A),\vol_G(B)\ge \vol(G)/3$; or
		\item  $\vol_G(A)\geq \vol(G)/2$, and graph $G[A]$ has conductance at least $\phi$.
	\end{itemize}
	\end{thm}

Notice that \ref{thm:intro:main} immediately follows from \ref{thm:main restatement}. The remainder of this subsection is dedicated to the proof of \ref{thm:main restatement}. We set $N_1=8N_0$, where $N_0$ is the universal constant used in \ref{thm: sparse edge cut of large profit or witness} and  \ref{thm: sparse edge cut or expander}.

We start by using Algorithm \reducedegree from \ref{subsec:constant degree}, in order to construct, in time $O(m)$, a graph $\hG$ whose maximum vertex degree is bounded by $10$, and $|V(\hG)|=2m$. Denote $V(G)=\set{v_1,\ldots,v_n}$. Recall that graph $\hG$ is constructed from graph $G$ by replacing each vertex $v_i$ with an $\alpha_0$-expander $H(v_i)$ on $\deg_G(v_i)$ vertices, where $\alpha_0=\Theta(1)$. For convenience, we denote the set of vertices of $H(v_i)$ by $V_i$. Therefore, $V(\hG)$ is a union of the sets $V_1,\ldots,V_n$ of vertices. Consider now some cut subset $S$ of vertices of $\hG$. As before, we say that $S$ is a \emph{canonical} vertex set iff for every $1\leq i\leq n$, either $V_i\subseteq S$ or $V_i\cap S=\emptyset$ holds. 
The main subroutine in the proof of \ref{thm:main restatement} is summarized in the following lemma.

\begin{lem}\label{lem: single stage}
	There is a universal constant $c_1$ and a deterministic algorithm, that, given a canonical vertex subset $V'\subseteq V(\hat G)$ containing at least $|V(\hat G)|/2$ vertices of $\hat G$, 
	and parameters $0<\psi<1$, $0<z'<z$, such that for every partition $(A,B)$ of $V'$ with $|E_{\hat G}(A,B)|\leq \psi\cdot  \min\set{|A|,|B|}$, $\min\set{|A|,|B|}\leq z$ holds,
	computes a partition $(X,Y)$ of $V'$, where both $X,Y$ are canonical subsets of $V(\hat G)$, $|X|\leq |Y|$ (where possibly $X=\emptyset$),  $|E_{\hG}(X,Y)|\leq \psi \cdot|X|$, and one of the following holds:
	
	\begin{itemize}
		\item either $|X|,|Y|\geq |V'|/3$ (note that this can only happen if $z\geq |V'|/3$); or
		\item for every partition $(A',B')$ of the set $Y$ of vertices with $|E_{\hG}(A',B')|\leq \frac{\psi}{c_1(\log n)^{c_1r}}\cdot \min\set{|A'|,|B'|}$, $\min\set{|A'|,|B'|}\leq z'$ must hold (if $z'<1$, then graph $\hat G[Y]$ is guaranteed to be a $\frac{\psi}{c_1(\log n)^{c_1r}}$-expander).
	\end{itemize}

The running time of the algorithm is $O\left (\frac{z}{z'}\cdot m^{1+O(1/r)+o(1)}\cdot (\log n)^{O(r^2)}\right )$.
\end{lem}

\begin{proof}
We let $c_1$ be a large enough constant, whose value we set later, and we let $\psi'=\psi/c_1$. We also use a parameter $z^*=\frac{z'}{c_0c_1 (\log n)^{c_0r}}$, where $c_0$ is the constant from \ref{thm: sparse edge cut of large profit or witness} and  \ref{thm: sparse edge cut or expander}. Assume first that $z^*\geq 1$; we will discuss the other case later.

Our algorithm is iterative. At the beginning of iteration $i$, we are given a subgraph $G_i\subseteq \hat G$, such that $V(G_i)\subseteq V'$ is a canonical subset of vertices, and $|V(G_i)|\geq 2|V'|/3$; at the beginning of the first iteration, we set $G_1=\hat G[V']$. 
At the end of iteration $i$, we either terminate the algorithm with the desired solution, or we compute a canonical subset $S_i\subseteq V(G_i)$ of vertices, such that $|S_i|\leq |V(G_i)|/2$, and $|E_{G_i}(S_i, V(G_i)\setminus S_i)|\leq \psi\cdot |S_i|/2$. We then delete the vertices of $S_i$ from $G_i$, in order to obtain the graph $G_{i+1}$, that serves as the input to the next iteration. The algorithm terminates once the current graph $G_i$ contains fewer than $2|V'|/3$ vertices (unless it terminates with the desired output beforehand).

We now describe the execution of the $i$th iteration. We assume that the sets $S_1,\ldots,S_{i-1}$ of vertices are already computed, and that $\sum_{i'=1}^{i-1}|S_{i'}|\leq |V'|/3$. Recall that $G_i$ is the sub-graph of $\hat G[V']$ that is obtained by deleting the vertices of $S_1,\ldots,S_{i-1}$ from it. Recall also that we are guaranteed that $V(G_i)$ is a canonical set of vertices, and $|V(G_i)|\geq 2|V'|/3\geq |V(\hat G)|/3\geq 2m/3$. From the definition of parameter $N_1$, we are guaranteed that $|V(G_i)|^{1/r}\geq N_0$. We apply \ref{thm: sparse edge cut of large profit or witness} to graph $G_i$, with parameters $\psi'$ and $z^*$. We now consider two cases.

In the first case, the algorithm from \ref{thm: sparse edge cut of large profit or witness} returns a cut $(X',Y')$ in graph $G_i$ with $|X'|,|Y'|\geq z^*$, and $|E_{G_i}(X',Y')|\leq \psi'\cdot\min\set{|X'|,|Y'|}$. We then use  Algorithm \makecanonical from \ref{lem: degree reduction balanced cut case} in order to obtain a cut $(X'',Y'')$ of $G_i$, such that both $X'',Y''$ are canonical vertex sets, $|X''|,|Y''|\geq \min\set{|X'|,|Y'|}/2$, and $|E_{G_i}(X'',Y'')|\leq O(E_{G_i}(X',Y'))$. We assume w.l.o.g. that $|X''|\leq |Y''|$. Notice that, in particular, $|X''|\geq \Omega(z^*)$, and $|E_{G_i}(X'',Y'')|\leq O(\psi'\cdot |X''|)$. Recall that $\psi'=\psi/c_1$. By letting $c_1$ be a large enough constant, we can ensure that $|E_{G_i}(X'',Y'')|\leq \psi\cdot |X''|/2$. We set $S_i=X''$. If $\sum_{i'=1}^i|S_{i'}|\leq |V'|/3$ continues to hold, then we let $G_{i+1}=G_i\setminus S_i$, and continue to the next iteration. Otherwise, we terminate the algorithm, and return the partition $(X,Y)$ of $V'$ where $X=\bigcup_{i'=1}^iS_{i'}$, and $Y=V'\setminus X$. Recall that we are guaranteed that $|X|\geq |V'|/3$. Moreover, since $|V(G_i)|\geq 2|V'|/3$ held, and $|S_i|\leq |V(G_i)|/2$, we are guaranteed that $|Y|\geq |V(G_i)|/2\geq |V'|/3$. Lastly, our algorithm guarantees that $|E_{\hat G}(X,Y)|=\sum_{i'=1}^i|E_{G_{i'}}(S_{i'},V(G_{i'}\setminus S_{i'})|\leq \frac{\psi}{2} \cdot \sum_{i'=1}^i|S_{i'}|\leq \frac{\psi |X|}{2}$. Since $|Y|\geq |X|/2$ must hold, we get that $|E_{\hat G}(X,Y)|\leq \psi |Y|$, and altogether, $|E_{\hat G}(X,Y)|\leq \psi\cdot\min\set{|X|,|Y|}$. 

Next, we assume that the algorithm from \ref{thm: sparse edge cut of large profit or witness}, when applied to graph $G_i$, 
correctly establishes that for every cut $(A',B')$ in $G_i$ with $\Psi_{G_i}(A',B')\leq \psi'/ (\log n)^{c_0r}$, $\min\set{|A'|,|B'|}<c_0z^*\cdot (\log n)^{c_0r}$ holds. Then we terminate the algorithm and return a partition $(X,Y)$ of $V'$, where $Y=V(G_i)$ and $X=\bigcup_{i'=1}^{i-1}S_{i'}$.  From the above discussion, both $X,Y$ are canonical subsets of $V(\hat G)$, $|X|\leq |Y|$, and $|E_{\hG}(X,Y)|\leq \psi |X|$. It is now enough to show that for every partition $(A',B')$ of the set $Y$ of vertices with $|E_{\hG}(A',B')|\leq \psi/c_1(\log n)^{c_1r}$, $\min\set{|A'|,|B'|}\leq z'$ holds.

Consider any partition $(A',B')$ of $Y$ with $|E_{\hG}(A',B')|\leq \frac{\psi}{c_1(\log n)^{c_1r}}\cdot \min\set{|A'|,|B'|}$.
Observe that $(A',B')$ is also a cut in $G_i$, whose sparsity is $\Psi_{G_i}(A',B')\leq \frac{\psi}{c_1(\log n)^{c_1r}}\leq \frac{\psi'} {(\log n)^{c_0r}}$, if $c_1\geq c_0$. 
 Therefore, \ref{thm: sparse edge cut of large profit or witness} guarantees that $\min\set{|A'|,|B'|}<c_0z^*\cdot (\log n)^{c_0r}$ holds. Since $z^*=\frac{z'}{c_0c_1 (\log n)^{c_0r}}$, we get that $\min\set{|A'|,|B'|}\leq z'$, as required.

It now remains to analyze the running time of the algorithm. Observe that we are guaranteed that for all $i$, $|S_i|\geq \Omega(z^*)$. Notice however that throughout the algorithm, if we set $A=\bigcup_{i'=1}^iS_{i'}$ and $B=V'\setminus A$, then 
$|A|<|B|$ holds, and $|E_{\hat G}(A,B)|\leq \psi\cdot  |A|$. Therefore, from the condition of the lemma, $|A|\leq z$ must hold. Overall, the number of iterations in the algorithm is bounded by $O(z/z^*)=(\log n)^{O(r)}\cdot z/z'$, and, since every iteration takes time $O\left (m^{1+O(1/r)+o(1)}\cdot (\log n)^{O(r^2)}\right )$, the total running time of the algorithm is bounded by $O\left (\frac z {z'}\cdot m^{1+O(1/r)+o(1)}\cdot (\log n)^{O(r^2)}\right )$. 

Lastly, we need to consider the case where $z^*<1$. In this case, in every iteration, we employ  \ref{thm: sparse edge cut or expander} instead of 
\ref{thm: sparse edge cut of large profit or witness}. The two main differences are that (i) we are no longer guaranteed that each set $S_i$ has large cardinality (the cardinality can be arbitrarily small); and (ii) if the lemma does not return a cut $(X',Y')$, then it correctly establishes that the current graph $G_i$ is a $\psi''$-expander, for $\psi''=\frac{\psi'}{(\log n)^{c_0r}}\geq \frac{\psi}{c_1 (\log n)^{c_1r}}$, if we choose $c_1$ to be at least $c_0$. This affects our analysis in two ways. First,  we need to bound the number of iterations differently -- it is now bounded by $O(z)$. However, since $z^*\leq 1$, $z\leq (\log n)^{O(r)}$, and so the number of iterations is bounded by $(\log n)^{O(r)}$ as before, and the running time remains $O\left (m^{1+O(1/r)+o(1)}\cdot (\log n)^{O(r^2)}\right )$. Second, if, in the last iteration, the algorithm from \ref{lem: single stage} establishes that graph $G_i$ is a $\psi''$-expander, then we obtain the cut $(X,Y)$ as before, but now we get the stronger guarantee that $\hat G[Y]$ is a $\frac{\psi}{c_1(\log n)^{c_1 r}}$-expander.
\end{proof}

We are now ready to complete the proof of \ref{thm:main restatement}. Our algorithm will consist of at most $r$ iterations and uses the following parameters. First, we set $z_1=|V(\hat G)|/2=m$, and for $1<i\leq r$, we set $z_i=z_{i-1}/m^{1/r}$; in particular, $z_r=1$ holds. We also define parameters $\psi_1,\ldots,\psi_{r}$, by letting $\psi_r=\phi$, and, for all $1\leq i<r$, setting $\psi_{i}=8c_1(\log |V(\hat G)|)^{c_1r}\cdot \psi_{i+1}$, where $c_1$ is the constant from \ref{lem: single stage}. Notice that $\psi_1\leq \phi\cdot (\log m)^{O(r^2)}$.

In the first iteration, we apply \ref{lem: single stage} to the set $V'=V(\hat G)$ of vertices, with the parameters $\psi=\psi_1$, $z=z_1$, and $z'=z_2$.
Clearly, for every partition $(A,B)$ of $V'$ with $|E_{\hat G}(A,B)|\leq \psi_1\cdot \min\set{|A|,|B|}$, $\min\set{|A|,|B|}\leq z_1=m/2$ holds.
 Assume first that the outcome of the algorithm from  \ref{lem: single stage} is a partition $(X,Y)$ of $V(\hat G)$, where $X,Y$ are canonical subsets of $V(\hat G)$, $|X|,|Y|\geq |V(\hat G)|/3\geq \vol(G)/3$, and $|E_{\hG}(X,Y)|\leq \psi_1\cdot \min\set{|X|,|Y|}$. Let $(A,B)$ be the partition of $V(G)$, defined as follows: for every vertex $v_i\in V(G)$, we add $v_i$ to $A$ if $V_i\subseteq X$, and we add it to $B$ otherwise. Clearly, $\vol_G(A)=|X|\geq \vol(G)/3$, and similarly, $\vol_G(B)\geq \vol(G)/3$. Moreover, $|E_G(A,B)|=|E_{\hat G}(X,Y)|\leq \psi_1\cdot \min\set{|X|,|Y|}\leq \phi\cdot (\log m)^{O(r^2)} \vol(G)$. We then return the cut $(A,B)$ and terminate the algorithm.

We assume from now on that the algorithm from \ref{lem: single stage} returned a partition $(X,Y)$ of $V(\hat G)$, where both $X,Y$ are canonical subsets of $V(\hat G)$, $|X|\leq |Y|$ (where possibly $X=\emptyset$),  $|E_{\hG}(X,Y)|\leq \psi_1 \cdot|X|$, and the following guarantee holds: For every partition $(A',B')$ of the set $Y$ of vertices with $|E_{\hG}(A',B')|\leq 8\psi_2\cdot \min\set{|A'|,|B'|}$, $\min\set{|A'|,|B'|}\leq z_2$ must hold. We set $S_1=X$, and we let $\hat G_2=\hG\setminus S_1$.

The remainder of the algorithm consists of $r-1$ iterations. The input to the $i$th iteration is a subgraph $\hat G_i\subseteq \hat G$, containing at least half the vertices of $\hat G$, such that for every cut $(A',B')$ of $\hat G_i$ with $|E_{\hG}(A',B')|\leq \psi_i\cdot \min\set{|A'|,|B'|}$, $\min\set{|A'|,|B'|}\leq z_i$ must hold. (Observe that, as established above, this condition holds for graph $\hat G_2$). 
The output is a canonical subset $S_{i}\subseteq V(\hat G_i)$ of vertices, such that $|E_{\hat G_i}(S_i,V(\hat G_i)\setminus S_i)|\leq \psi_i \cdot |S_i|$, and, if we set $\hat G_{i+1}=\hat G_i\setminus S_{i}$, then we are guaranteed that for every cut $(A'',B'')$ of $\hat G_{i+1}$ with $|E_{\hG}(A'',B'')|\leq 8\psi_{i+1}\cdot \min\set{|A''|,|B''|}$, $\min\set{|A''|,|B''|}\leq z_{i+1}$ holds.
In particular, if $|E_{\hG}(A'',B'')|\leq \psi_{i+1}\cdot \min\set{|A''|,|B''|}$, then $\min\set{|A''|,|B''|}\leq z_{i+1}$ holds.
 In order to execute the $i$th iteration, we simply apply \ref{lem: single stage} to the set $V'=V(\hat G_i)$ of vertices, with parameters $\psi=\psi_i$, $z=z_i$ and $z'=z_{i+1}$. As we show later, we will ensure that $|V(\hat G_i)|\geq |V(\hat G)|/2\geq m$. Since, for $i>1$, $z_i<m/3\leq |V(\hat G_i)|/3$, the outcome of the lemma must be a partition $(X,Y)$ of $V'$, where both $X,Y$ are canonical subsets of $V(\hat G)$, $|X|\leq |Y|$ (where possibly $X=\emptyset$),  $|E_{\hG}(X,Y)|\leq \psi_i \cdot|X|$, and we are guaranteed that, for every partition $(A'',B'')$ of the set $Y$ of vertices with $|E_{\hG}(A'',B'')|\leq 8\psi_{i+1}\cdot \min\set{|A''|,|B''|}$, $\min\set{|A'|,|B'|}\leq z_{i+1}$ holds. Therefore, we can simply set $S_i=X$, $\hat G_{i+1}=\hat G_i\setminus S_i$, and continue to the next iteration, provided that $|\hat G_{i+1}|\geq |V(\hat G)|/2$ holds.

 We next show that this indeed must be the case, if $|V(\hat G_i)|\geq |V(\hat G)|/2$. Indeed, recall that for all $2\leq i'\leq i$, we guarantee that $|E_{\hat G_{i'}}(S_{i'},V(\hat G_{i'})\setminus S_{i'})|\leq \psi_{i'} \cdot |S_{i'}|\leq \psi_2\cdot |S_{i'}|$.
 Therefore, if we denote by $Z=\bigcup_{i'=2}^iS_{i'}$ and $Z'=V(\hat G_2)\setminus Z$, then $|E_{\hG}(Z,Z')|\leq \psi_2|Z|$. Since $|V(\hat G_i)|\geq |V(\hat G)|/2$, we get that $|Z'|\geq |V(\hat G)|/4$ must hold.
 Assume now that $|V(\hat G_{i+1})|=|Z'|<|V(\hat G)|/2$. 
 %Then $|Z|\geq |V(\hat G)|/2-|S_1|\geq |V(\hat G)|/2 -|V(\hat G)|/3\geq |V(\hat G)|/8$, and so $|Z'|\geq |Z|/2$. 
 Then $|Z'|\geq |Z|/4$, as $|Z|\leq |V(\hat G)|$ and $|Z'|\geq |V(\hat G)|/4$.
 Therefore, $|E_{\hG}(Z,Z')|\leq \psi_2|Z|\leq 4\psi_2|Z'|\leq 4\psi_2\min\set{|Z|,|Z'|}$. We have thus obtained a cut $(Z,Z')$ of $\hat G_2$, of sparsity less than $8\psi_2$, such that $|Z|,|Z'|>z_2$, contradicting the fact that such a cut does not exist.
 
 We continue the algorithm until we reach the last iteration, where $z_r=1$ holds. When we apply  \ref{lem: single stage} to the final graph $\hat G_r$, we obtain a partition $(X,Y)$ of $V(G_r)$, such that graph $\hat G[Y]$ is guaranteed to be a $\psi_r$-expander (recall that $\psi_r=\phi$). We let $B'=Y$ and $A'=V(\hat G)\setminus B'$. Using the same reasoning as before, we are guaranteed that $|B'|\geq |V(\hat G)|/2$, and that $|E_{\hat G}(A',B')|\leq \psi_1\cdot |A'|\leq \phi\cdot (\log m)^{O(r^2)}\cdot \vol(G)$. As discussed above, we are guaranteed that graph $\hat G[B'] $ is a $\phi$-expander.
Next, we define  a cut $(A',B')$ in graph $G$, as follows. For every vertex $v_i\in V(G)$, we add $v_i$ to $A$ if $V_i\subseteq X$, and we add it to $B$ otherwise. Clearly, $\vol_G(A)=|A'|$, and similarly, $\vol_G(B)=|B'|\geq \vol(G)/2$. Moreover, $|E_G(A,B)|=|E_{\hat G}(A',B')|\leq \phi\cdot (\log m)^{O(r^2)}\cdot \vol(G)$.
Since graph $\hat G[B']$ is a $\psi$-expander, it is immediate to verify that graph $G[A]$ has conductance at least $\phi$.

For all $1\leq i\leq r$, the running time of the $i$th iteration is $O\left (\frac{z_i}{z_{i+1}}\cdot m^{1+O(1/r)+o(1)}\cdot (\log n)^{O(r^2)}\right )= O\left (m^{1+O(1/r)+o(1)}\cdot (\log n)^{O(r^2)}\right )$, and the total running time is $O\left (m^{1+O(1/r)+o(1)}\cdot r \cdot (\log n)^{O(r^2)}\right )=O\left (m^{1+O(1/r)+o(1)} \cdot (\log n)^{O(r^2)}\right )$.
 This concludes the proof of \ref{thm:main restatement}.

\subsection{Applications of the Algorithm for Balanced Cut in Low Conductance Regime}
\label{sec:app nophi}

\subsection*{Approximation Algorithms for Sparsest Cut and Lowest Conductance Cut}

As observed already, \ref{thm: sparse edge cut or expander} immediately gives a deterministic $(\log n)^r$-approximation algorithm for the
Sparsest Cut problem on an $n$-vertex $m$-edge graph $G$, with running time $O\left (m^{1+O(1/r)+o(1)}\cdot (\log n)^{O(r^2)}\right )$, for all $r\leq O(\log n)$, proving \ref{thm: sparsest and low cond} for the Sparsest Cut problem.

We now show that we can obtain an algorithm with similar guarantees for the Lowest Conductance Cut problem. Let $G=(V,E)$ be an input to the Lowest-Conductance Cut problem, with $|V|=n$ and $|E|=m$, and let $\phi=\Phi(G)$. We can assume without loss of generality that $\phi<1/\left(c(\log n)^r\right )$ for some large enough constant $c$, since otherwise we can let $v$ be a lowest-degree vertex in $G$, and return the cut $(\set{v},V\setminus\set{v})$, whose conductance is $1$.
We use Algorithm \reducedegree from \ref{subsec:constant degree}, in order to construct, in time $O(m)$, a graph $\hG$, whose maximum vertex degree is bounded by $10$, and $|V(\hG)|=2m$. %Denote $V(G)=\set{v_1,\ldots,v_n}$. Recall that graph $\hG$ is constructed from graph $G$ by replacing each vertex $v_i$ with an $\alpha_0$-expander $H_i$ on $\deg_G(v_i)$ vertices, where $\alpha_0=\Theta(1)$. For convenience, we denote the set of vertices of $H(v_i)$ by $V_i$. Therefore, $V(\hG)$ is a union of the sets $V_1,\ldots,V_n$ of vertices. Consider now some cut subset $S$ of vertices of $\hG$. As before, we say that $S$ is a \emph{canonical} vertex set iff for every $1\leq i\leq n$, either $V_i\subseteq S$ or $V_i\cap S=\emptyset$ holds. 
Note that, if we denote $\psi=\Psi(\hat G)$, then $\psi\leq \phi$ must hold. This is since every cut $(A,B)$ in $G$ naturally defines a cut $(A',B')$ in $\hat G$, with $|A'|=\vol_G(A), |B'|=\vol_G(B)$, and $|E_{\hat G}(A',B')|=|E_G(A,B)|$. We use our approximation algorithm for the Sparsest Cut problem in graph $\hat G$, to obtain a cut $(X',Y')$ of $\hat G$ with $\Psi_{\hat G}(X',Y')\leq (\log n)^r\cdot \psi\leq (\log n)^r\cdot \phi$, in time $O\left (m^{1+O(1/r)+o(1)}\cdot (\log n)^{O(r^2)}\right )$. 
Using  Algorithm \makecanonical from \ref{lem: degree reduction balanced cut case}, we obtain a cut $(X'',Y'')$ of $\hat G$, with $|X''|\geq |X'|/2$, $|Y''|\geq |Y'|/2$, and $|E_{\hat G}(X'',Y'')|\leq O(|E_{\hat G}(X',Y')|)\leq O((\log n)^r\cdot \phi)$, such that both $X''$ and $Y''$ are canonical vertex sets. This cut naturally defines a cut $(X,Y)$ in $G$, with $\vol_G(X)=|X''|$, $\vol_G(Y)=|Y''|$, and $|E_G(X,Y)|=|E_{\hat G}(X'',Y'')|$. Therefore, $\Phi_G(X,Y)\leq O((\log n)^r\cdot \phi)$, and the running time of the algorithm is $O\left (m^{1+O(1/r)+o(1)}\cdot (\log n)^{O(r^2)}\right )$.

\iffalse
First, observe that we can $n^{o(1)}$-approximate 
graph conductance in almost-linear time independent of the value of the conductance itself:
\begin{cor}\label{cor:conductance_nophi}
There is a deterministic algorithm that, given an $m$-edge undirected
graph $G,$ compute a $(2^{O(\log n\log\log n)^{2/3}})$-approximation
to the lowest-conductance cut in time $\Ohat(m)$.
\end{cor}

This follows by  calling the  $\CA$ algorithm from \ref{lem:CA nophi} where $\beta_\cut=0$, $\beta_\aug =0$,  $\phi_{L}\le\frac{\phi_{U}}{(c_{2}\log m)^{12\ell+8}}$,  $\ell=\Theta((\log m/\log\log m)^{1/3})$ and doing binary search on $\phi_U$.  
It is not hard to adjust the proof so that we can approximate sparsity of the graph instead of conductance in $\Ohat(m)$ time as well.
\fi

\subsection*{Expander Decomposition}
Observe that \ref{thm:intro:main} immediately implies an almost-linear time algorithm for computing  $(\epsilon,\phi)$-expander decomposition even for very small $\epsilon$ and $\phi$. The proof of the following corollary is almost identical to that of \ref{cor:expander decomp} and is omitted here.

\begin{cor}
	\label{cor:expander decomp2} There is a deterministic algorithm that, given a graph
	$G=(V,E)$ with $m$ edges and parameters $\epsilon\in(0,1]$ and $1\leq r\leq O(\log m)$, computes
	a $\left(\epsilon,\phi \right)$-expander decomposition of $G$ with $\phi=\Omega(\epsilon/(\log m)^{O(r^2)})$
	in time  $O\left ( m^{1+O(1/r)+o(1)}\cdot (\log m)^{O(r^2)} \right )$.
\end{cor}

{
}

\section{Open Problems}\label{sec:conclusion}

A very interesting remaining open problem is to obtain deterministic algorithms for Minimum Balanced Cut, Sparsest Cut and Lowest-Conductance Cut, that achieve a polylogarithmic approximation ratio, with running time $O(m^{1+o(1)})$. It would also be interesting to obtain deterministic $n^{o(1)}$-approximation algorithms for these problems with running time $\tilde O(m)$. The latter result would imply a near-linear time deterministic algorithm for computing an expander decompositions, matching the performance of the best current randomized algorithm of \cite{SaranurakW19}). 

It is typically desirable for dynamic graph algorithms to have polylogarithmic update time complexity. Our result for dynamic connectivity (\cref{thm:intro:dynConn}) only guarantees $n^{o(1)}$ update time. Designing a deterministic algorithm with polylogarithmic update time for dynamic connectivity remains a major open problem. In fact, it is already very interesting to achieve such bounds with a Las Vegas randomized algorithm. It is also very interesting to design a Monte Carlo randomized algorithm for maintaining a spanning forest in polylogarithmic update time that does not need the so-called {\em oblivious adversary assumption}.\footnote{It was shown by Kapron et~al.~\cite{KapronKM13}, that a spanning forest can be maintained in polylogarithmic update time by a Monte Carlo randomized algorithm under the oblivious adversary assumption.}
%
% Moreover, if a (minimum) spanning forest is maintained, even a Monte Carlo randomized algorithm that does not need the so-called {\em oblivious adversary assumption} is already very interesting.\footnote{It was known, due to Kapron et~al.~\cite{KapronKM13}, that dynamic spanning forest can be maintained in polylogarithmic update time by a Monte Carlo randomized algorithm under the oblivious adversary assumption.}
% 
We remark that even if one can implement an algorithm for \Cref{thm:intro:main} with running time $\tilde{O}(m)$ time and approximation factor $O(\polylog n)$, this would not immediately imply any of the above goals.
The reason is that there are several components in the algorithm of Nanongkai et~al.~\cite{NanongkaiSW17} that each incur the $n^{o(1)}$ factor in the update time.

Our deterministic algorithm for spectral sparsifiers from \ref{cor:sparsifier} only achieves a factor $n^{o(1)}$-approximation. It is an intriguing open question whether $(1+\epsilon)$-approximate cut/spectral sparsifiers can be computed deterministically in almost-linear time. It is also interesting whether there is a deterministic $O(\sqrt{\log n})$-approximation algorithm for Lowest-Conductance Cut, whose running time matches that of the best currently known randomized algorithm, which is $O(m^{1+\epsilon})$, for an arbitrarily small constant $\epsilon>0$ \cite{Sherman09}. We believe that resolving both questions would require significantly new ideas.% and might be related to designing fast SDP solvers. 

%\begin{itemize}
%\item Sparsifier with $(1+\epsilon)$-approximation: the replacement of mildly expanding pieces
%by full expanders leads to a high error of $n^{o(1)}$:
%reducing this error to the $(1 + \epsilon)$ quality of randomized
%algorithms is an intriguing question that we believe requires
%significant new ideas.
%Such lower error routines are also needed for parallel
%solvers for Laplacians~\cite{PengS14,KyngLPSS16,CohenKKPPRS18}.
%
%%\item Polylog approx for $\BCut$ in near linear time.
%
%\item 
%$O(\sqrt{\log{n}})$ approximations to conductance in subqradratic time. This should be related to derandomizing fast SDP solvers.  \end{itemize}

\section*{Acknowledgements} 
This project has received funding from the European Research Council (ERC) under the European Union's Horizon 2020 research and innovation programme under grant agreement No 715672.      
Nanongkai was also partially supported by the Swedish Research Council (Reg. No. 2015-04659).
Chuzhoy was supported in part by NSF grant CCF-1616584.  
Gao and Peng were supported in part by NSF grant CCF-1718533.

\appendix

\section{Proof of \ref{thm:KKOV-new}}\label{sec: appx: proof of CMG}

In this section we prove \ref{thm:KKOV-new}. The proof is practically identical to that in  \cite{KhandekarKOV2007cut}, but, since the algorithm is slightly different, we present it here for completeness. 
We denote by $H_i$ the graph $H$ obtained after $i$ iterations of the cut-matching game. Therefore, graph $H_0$ has a set $V$ of vertices and no edges, and for all $i$, graph $H_i$ is defined over the same set $V$ of vertices, while the set $E(H_i)$ of edges is the union of $i$ matchings $M_1,\ldots,M_i$, where for $1\leq i'\leq i$, matching $M_{i'}$ is a perfect matching between two equal-cardinality subsets $A_{i'},B_{i'}$ of $V$. Notice that for every vertex $v\in V$, for all $1\leq i'\leq i$, there is exactly one edge of $M_{i'}$ that is incident to $v$.

Consider a random walk in graph $H_i$ that starts at an arbitrary vertex $v=v_0$. For all $1\leq i'\leq i$, at step $i'$, with probability $1/2$, the random walk stays at the current vertex $v_{i'-1}$ (so $v_{i'}=v_{i'-1}$), and with probability $1/2$ it moves to the unique vertex $v_{i'}$ that is connected to $v_{i'-1}$ with an edge of $M_{i'}$. We denote by $p_{i'}(v,u)$ the probability that the random walk that starts at $v$ is located at vertex $u$ after $i'$ steps. 

For a vertex $v\in V$ and index $i$, we define the potential $\Phi_i(v)=\sum_{u\in V}p_{i}(v,u)\log(1/p_i(v,u))$. In other words, $\Phi_i(v)$ is the entropy of the distribution $\set{p_i(v,u)}_{u\in V}$. Clearly, $\Phi_0(v)=0$, and for all $i$, $\Phi_i(v)\leq \log n$. Finally, we define the total potential at the end of iteration $i$:

\[\Phi_i=\sum_{v\in V}\Phi_i(v). \]

From the above discussion, $\Phi_0=0$, and for all $i$, $\Phi_i\leq O(n\log n)$.

%We define a different potential function, based on the entropy of the current distribution:

%\[\Phi_t=\sum_{i,j}W^t(i,j)\log(1/W^t(i,j)).\]

%Note that $\phi_0=0$, and  $\phi_t\leq n\log n$ for all $t$, as each index $i$ contributes the entropy of the distribution defined by $W_i$, and this entropy is bounded by $\log n$. The algorithm consists of two phases.

%The first phase continues as long as there are $1/4$-balanced cuts that are sparse. Once the first phase ends, we run one more iteration of the algorithm, as follows. We cut off sparse cuts from the main cluster one-by-one. Suppose the clusters we cut off are $C_1,\ldots,C_r$. Their union must contain fewer than $n/4$ vertices. The remainder of the graph must be an expander. Let $A=\bigcup_iC_i$, and let $A'$ be obtained by adding additional vertices to $V_1$, until its cardinality becomes $n/2$. We then give the cut $(A,\overline A)$ as the last cut in the game. As the result, all vertices of $A$ will become connected to the vertices of the expander, and so we will obtain an expander.

In order to complete the proof of \ref{thm:KKOV-new}, it is sufficient to prove the following claim.

\begin{claim}\label{claim: potential increase in one iteration}
	Let $i$ be any iteration in which the cut player computed a partition $(A_i,B_i)$ of $V(H_{i-1})$ with $|B_i|\geq |A_i|\geq n/4$ and $|E_{H_{i-1}}(A_i,B_i)|\leq n/100$. Let $(A_i',B_i')$ be any partition of $V$ into two equal-cardinality subsets such that $A_i\subseteq A'_i$, and let $M_{i}$ be any perfect matching between $A'_i$ and $B'_i$. Let $H_i$ be the graph obtained from $H_{i-1}$ by adding the edges of $M_i$ to it. Then $\Phi_i\geq \Phi_{i-1}+\Omega(n)$.
	\end{claim}

Since the initial potential $\Phi_0=0$, and the potential increases by $\Omega(n)$ in every iteration, the number of iterations is bounded by $O(\log n)$, as the potential may never exceed $O(n\log n)$. It now remains to prove \ref{claim: potential increase in one iteration}. The proof is almost identical to that in  \cite{KhandekarKOV2007cut}.

\begin{proofof}{\ref{claim: potential increase in one iteration}}
%	Assume that $A_i$ contains $bn$ vertices, where $1/4\leq b\leq 1/2$. 
	For convenience, we denote $E'=E_{H_{i-1}}(A_i,B_i)$.
	Notice that $|E'|\leq n/100\leq  |A_i|/25$ must hold.  
	%For convenience, we let $P(u,v)$ denote the probability that a random walk that starts from $u$, arrives at $v$ after $t$ iterations. 
	For a vertex $v\in V$ and a subset $Y\subseteq V$ of vertices, we let $P(v,Y)=\sum_{u\in Y}p_{i-1}(v,u)$ be the probability that the random walk that we defined above is located at a vertex of $Y$ at the end of the $(i-1)$th step, if it started from $v$.  
	Similarly, for two disjoint subsets $X,Y$ of vertices of $V$, we denote by $P(X,Y)=\sum_{v\in X}P(v,Y)$.
	
	Consider now the following experiment. We place one unit of mass on every vertex $v\in A_i$, and then perform $(i-1)$ iterations. 
	For all $1\leq i'< i$, in order to perform iteration $i'$, we consider every vertex $a\in V$ and the mass $\mu(a)$ that is currently located at vertex $a$. We keep half of this mass at vertex $a$, and the remaining half of the mass is moved to the unique vertex $b\in V$ such that $(a,b)\in M_{i'}$. 
	
	It is easy to verify (using induction) that, over the course of this experiment, at every time step, the amount of mass at every given vertex $a\in V$ is at most $1$, and moreover, at most $1$ unit of mass is moved across any edge in every iteration. Notice that $P(A_i,B_i)$ is precisely the amount of mass that is located at the vertices of $B_i$ after the end of the $(i-1)$th iteration.
	
	%We now run the experiment described above with $X=A_i$ and $Y=B_i$.
	
Consider now some edge $e\in E'$. There must be a unique index $1\leq i'< i$ such that $e\in M_{i'}$  (we consider parallel edges as separate edges). Mass can be transferred along the edge $e$ only in iteration $i'$, and only one unit of mass can be transferred across it then. Therefore, the total amount of mass that is located at the vertices of $B_i$ at the end of iteration $(i-1)$ is at most $|E'|\leq |A_i|/25$. 
Equivalently, $P(A_i,B_i)\leq |A_i|/25$.

%	As each edge of $E'$ belongs to at exactly one of the matchings $M_1,\ldots, M_{i-1}$ (we consider parallel edges as separate edges), it can be traversed at most once, in iteration $i$ of the walk, and so overall at most $|A|/10$ mass of the flow may leave $A$; the remaining mass has to stay inside $A$. That is, $P(A,B)\leq |A|/10$.

We say that a vertex $a\in A_i$ is \emph{interesting} iff $P(a,B_i)\leq 1/4$. Then at least $|A_i|/2$ of the vertices of $A$ are interesting. Indeed, otherwise, we have $|A_i|/2$ vertices $a\in A_i$ with $P(a,B_i)>1/4$, so $P(A_i,B_i)\geq |A_i|/8$ must hold, a contradiction.

%----------------------------------------
Let us now fix an interesting vertex $a\in A_i$. Recall that $P(a,B_i)\leq 1/4$, and, therefore, $P(a,A_i)\geq 3/4$ must hold. Consider now the matching $M_i$, and some matched pair $e=(u,v)\in M_{i}$ with $u\in A_i$, $v\not\in A_i$. Denote $p=p_{i-1}(a,u)$, and  $q=p_{i-1}(a,v)$. We define the weight of the $e$ with respect to $a$ be $w_a(e)=p$. Note that $\sum_{e\in M_{i}}w_a(e)=P(a,A)\geq 3/4$.
	We say that $e$ is a good edge with respect to $a$ iff $p\geq 2q$. Let $E'(a)\subseteq M_{i}$ be the set of all edges that are good with respect to $a$, and let $E''(a)=M_{i}\setminus E'(a)$.
	
	\begin{claim}
		For every interesting vertex $a$, $\sum_{e\in E'(a)}w_a(e)\geq 1/4$.
	\end{claim}
	
	\begin{proof}
		Note that:
		
		 \[\begin{split}
		 \sum_{e\in E''(a)}w_a(e)&= \sum_{e=(v,u)\in E''(A)}p_{i-1}(a,u)\\
		 &\leq 2\sum_{e=(u,v)\in E''(a)}p_{i-1}(a,v)\\
		 &\leq 2P(a,B)\leq 1/2,
		 \end{split}\] 
		 
		 while $P(a,A)\geq 3/4$, so $\sum_{e\in E'(a)}w_a(e)\geq P(a,A)-\sum_{e\in E''(a)}w_a(e)\geq 1/4$.%   We claim that the total weight of all pairs that are bad w.r.t. $u$ is at least $1/10$. Otherwise, the total probability that a walk from $u$ ends up outside of $A$ is more than $(9/10)/2>1/4$.  
	\end{proof}
	
	Consider now an edge $(u,v)\in M_{i}$ that is good with respect to an interesting vertex $a$, with the corresponding probabilities $p$ and $q$, then pairs $(a,u)$ and $(a,v)$ originally contribute $p\log\left (\frac 1 p\right )+q\log \left (\frac 1 q\right )$ to the potential $\Phi_{i-1}(a)$, and will contribute $(p+q)\log\left (\frac{2}{p+q}\right )$ to $\Phi_i(a)$, since $p_i(a,u)=p_i(a,v)=\frac{p+q}{2}$. The key claim is that the increase in the potential due to these pairs is at least $\Omega(p)$:
	
	\begin{claim}\label{claim: mass}
		Let $a$ is an interesting vertex, and let $(u,v)\in M_{i}$ be a good edge for $a$, with $u\in A_i$. Denote $p=p_{i-1}(a,u)$ and  $q=p_{i-1}(a,v)$. Then:
		
		\[(p+q)\log\left(\frac{2}{p+q}\right )- p\log\left (\frac 1 p\right )-q\log \left (\frac 1 q \right ) \geq \Omega(p).\]
	\end{claim}
	
	If the above claim is correct, then for each interesting vertex $a$, we get that $\Phi_i(a)-\Phi_{i-1}(a) \geq \sum_{e\in E'(a)}\Omega(w_a(e))\geq \Omega(1)$.
	Since the number of interesting vertices $a\in A_i$ is $\Omega(n)$, we get that $\Phi_i-\Phi_{i-1}\geq \Omega(n)$.
	
	It now remains to prove the claim. 
Denote $S=(p+q)\log\left(\frac{2}{p+q}\right )- p\log\left (\frac 1 p\right )-q\log \left (\frac 1 q \right )$. By regrouping the terms, we can write:

\[S= p\log \left (\frac{2p}{p+q}\right) +q \log \left (\frac{2q}{p+q}\right).\]

Denoting $q=\alpha p$, for some $0<\alpha\leq 1/2$, it is now enough to show that there is some constant $c>0$ that is independent of $\alpha$, such that:

\[\log \left(\frac{2}{1+\alpha}\right )+\alpha \log \left (\frac{2\alpha}{1+\alpha}\right ) \geq c. \]
	
Rewriting $\log \left(\frac{2}{1+\alpha}\right )=\log \left(1+\frac{1-\alpha}{1+\alpha}\right )$ and $\log \left (\frac{2\alpha}{1+\alpha}\right )=\log \left(1-\frac{1-\alpha}{1+\alpha}\right )$, and using Taylor expansion for $\ln(1+\epsilon)$ completes the proof.
	
	\iffalse
	%\[2\frac{p+q}{2}\log \frac{2}{p+q}-p\log 1/p-q\log 1/q\geq \Omega(p),\]
	
	%or:
	
	\[\log\left(\frac{2^{p+q}\cdot p^p\cdot q^q}{(p+q)^{p+q}}\right )\geq \Omega(p),\]
	
	or:
	
	\[\frac{2^{p+q}\cdot p^p\cdot q^q}{(p+q)^{p+q}}\geq 2^{cp},\]
	
	for some constant $c$.  We denote $q=\alpha p$, where $0<\alpha<1/2$. The LHS becomes:
	
	\[\begin{split}
	\frac{2^{(1+\alpha)p}\cdot p^p\cdot (\alpha p)^{\alpha p}}{((1+\alpha)p)^{(1+\alpha)p}}&
	=\left(\frac{2^{(1+\alpha)}\cdot p^{(1+\alpha)}\cdot \alpha^{\alpha} }{(1+\alpha)^{(1+\alpha)}\cdot p^{(1+\alpha)}}\right )^p\\
	&=\left(\left(\frac{2}{1+\alpha}\right)^{1+\alpha}\cdot \alpha^{\alpha}  \right )^p.
	\end{split}\]
	
	The minimum of the function $x^x$ is $1/e$. As $\alpha\leq 1/2$, we get that $\frac{2}{1+\alpha}\geq \frac 4 3$, and so the expression is at least:
	
	\[\left((4/3)^{1+\alpha}/e^2\right)^p\geq c^p,\]
	
	for some constant $c$.
\fi
	
\end{proofof}

%It is now therefore enough to focus on the first phase of the algorithm. We would like to show that after $O(\log n)$ iterations, each $1/4$-balanced cut has sparsity at least $c$, for some constant $c$. As long as this is not the case, we compute a sparsest $1/4$-balanced cut $(A,B)$, where $|A|<|B|$. We move arbitrary vertices from $B$ to $A$ to ensure that each side contains exactly half the vertices, obtaining a cut $(A',B')$. The crucial part is to show that no matter how the vertices of $A'$ and $B'$ are matched, the potential increases by at least $\Omega(n)$, if the original cut $(A,B)$ was sparse enough.

\bibliographystyle{alpha}
\bibliography{bibliography}

\end{document}